%% file: main.tex
\documentclass{article}

\input{modules/packages}
\input{modules/macros}

\input{modules/tikzstyles}

\usepackage{microtype}
\usepackage{bbm}

\newcommand\cat[1]{\ensuremath{\mathbf{#1}}}

\newcommand{\idm}[1]{\mbox{id}_{#1}}

\newcommand{\tinycounitblack}{\XcounitSym}

\newcommand{\tinycomultblack}{\XcomultSym}
\newcommand\whitefrob[1]{\ensuremath{(#1,\,\ZmultSym,\ZunitSym,\ZcomultSym,\ZcounitSym)}}
\newcommand\blackfrob[1]{\ensuremath{(#1,\,\XmultSym,\XunitSym,\XcomultSym,\XcounitSym)}}
\newcommand{\scpair}{(\hbox{\input{modules/symbols/ZdotSym.tex}}\!,\hbox{\input{modules/symbols/XdotSym.tex}}\!)}

\begin{document}
\bibliographystyle{plain}

\title{Fourier transforms from strongly complementary observables}

\author{Stefano Gogioso \qquad William Zeng \\
    Quantum Group, Department of Computer Science\\
    University of Oxford, UK\\
    \texttt{stefano.gogioso@cs.ox.ac.uk} \qquad \texttt{william.zeng@cs.ox.ac.uk}
}

\maketitle

\begin{abstract}
\noindent Ongoing work in quantum information emphasizes the need for a structural understanding of quantum speedups: in this work, we focus on the quantum Fourier transform and the structures in quantum theory that enable it. We work in the setting of dagger symmetric monoidal categories. Each such category represents a certain kind of process theory, of which the category of finite dimensional Hilbert spaces and linear maps (associated to quantum computation) is but one instance. We elucidate a general connection in any process theory between the Fourier transform and strongly complementary observables, i.e. Hopf algebras in dagger symmetric monoidal categories. 
We generalise the necessary tools of representation theory from $\fdHilbCategory$ to arbitrary dagger symmetric monoidal categories. We define groups, characters and representations, and we prove their relation to strong complementarity. The Fourier transform is then defined in terms of pairs of strongly complementary observables, in both the abelian and non-abelian case. In the abelian case, we draw the connection with Pontryagin duality and provide categorical proofs of the Fourier Inversion Theorem, the Convolution Theory, and Pontryagin duality.
Our work finds application in the novel characterisation of the Fourier transform for the category $\RelCategory$ of sets and relations. This is a result of interest for the study of categorical quantum algorithms, as the usual construction of the quantum Fourier transform in terms of Fourier matrices is shown to fail in $\RelCategory$. Despite this, the process theoretic perspective on the Fourier transform is sensible in this setting. Furthermore, our categorical setting provides a generalisation of the abelian Fourier transform from finite-dimensional Hilbert spaces to finite-dimensional modules over arbitrary semirings, as well as a further generalisation to finite non-abelian groups, including a fully categorical generalisation of the Gelfand-Naimark theorem.
\end{abstract}

\setcounter{tocdepth}{2} 

\section{Background}
\label{section_Background}
\noindent We work with the framework of dagger symmetric monoidal categories, and, in particular with dagger Frobenius algebras on their objects. Some basic background definitions are covered in Appendix~\ref{app:basic}. A reference for the basic relationship between these structures, quantum mechanics, and process theories is Coecke et al.~\cite{coecke2015generalised}, where they are called \emph{generalised compositional theories}.

The key connection between dagger Frobenius algebras and non-degenerate observables in quantum mechanics is provided by Coecke et al.~\cite{coecke2013new}, where it is proven that dagger special commutative Frobenius algebras ($\dagger$-SCFAs) in $\fdHilbCategory$ canonically correspond to orthonormal bases (via their unique basis of \emph{classical states}), and can thus be used to model a basis of eigenstates;  more generally, commutative $\dagger$-Frobenius algebras ($\dagger$-CFAs) correspond to orthogonal bases. In order to model possibly degenerate observables, Coecke and Pavlovic~\cite{CQM-QuantumMeasuNoSums} introduce spectra as coalgebras for the comonoid part of a classical structure. Throughout the paper we will be working in a slightly more general setting than that of classical structures, where specialness and commutativity are not neccessarily assumed. 

Strongly complementary pairs of classical structures appear in \cite{coecke2012strong}\cite{CQM-ZXCalculusComplete} to model non-locality in terms of non-commutative non-degenerate observables, while Kissinger~\cite{kissinger2012pictures} shows their correspondence to finite abelian groups in $\fdHilbCategory$. Some of the material on representation theory in this paper already appears in \cite{vicary-tqa} and \cite{zeng2014abstract} in relation to quantum algorithms for the hidden subgroup problem and the group homomorphism identification problem, respectively; furthermore, it forms the basis of the recent \cite{StefanoGogioso-CategoricalSemanticsSchrodingersEqn} on categorical quantum dynamics. Finally, in addition to~\cite{coecke2015generalised} the upcoming~\cite{zeng2015abstract},~\cite{CQM-QCSnotes}, and~\cite{CQM-CQMnotes} provide a comprehensive reference for many structures and results used here.

\section{The Fourier Transform}
\label{section_FourierTransformIntro}

\noindent While they happen to coincide for qubits and systems composed of qubits, the Fourier transform of a general system is mathematically distinct from a Fourier matrix.  In particular, a Fourier matrix is a Fourier transform equipped with the choice of an isomorphism that is, in general, non-canonical. These Fourier matrices then correspond to strongly complementary observables (with a choice of isomorphism) in the same way that complex Hadamard matrices correspond to complementary ones. The relationship between these concepts is illustrated in Figure~\ref{fig:FTtoHrelationship}, and will be introduced in detail in this section. 

\begin{figure}
\includegraphics[width=\linewidth]{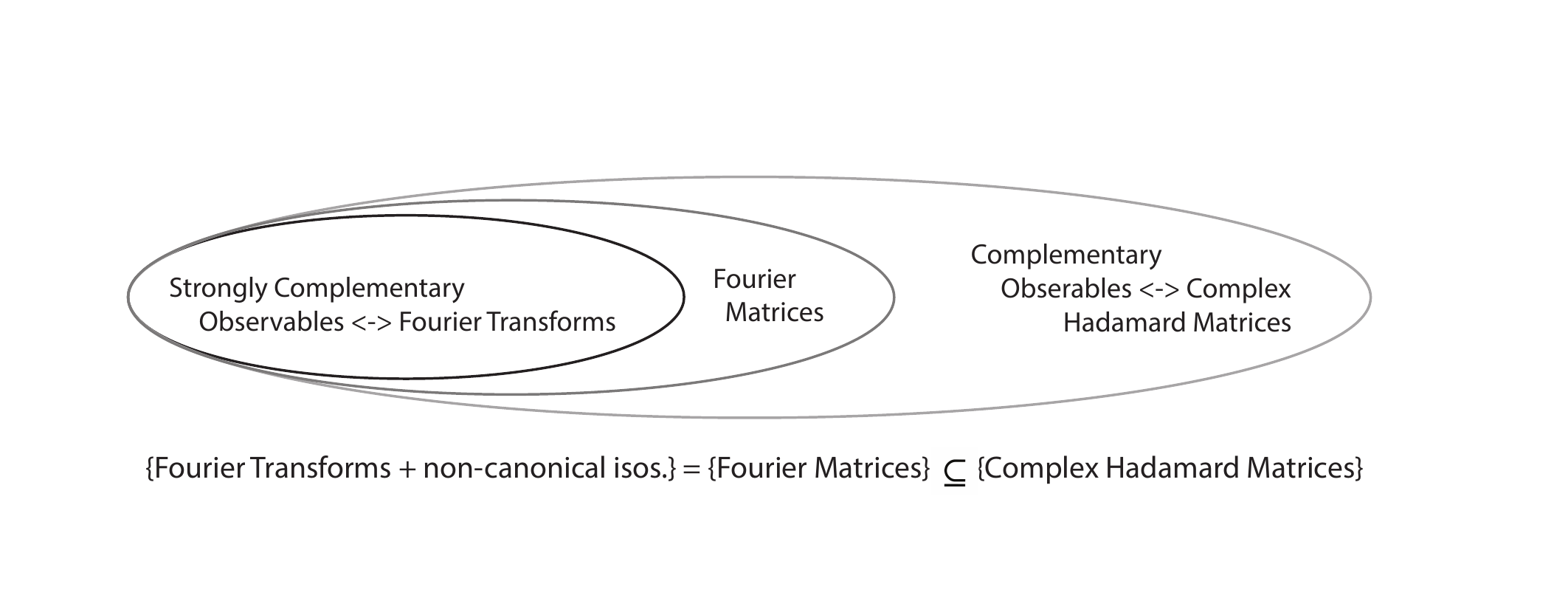}
\caption{Schematic of the relationship between Fourier transforms, Fourier matrices, Hadamard matrices, strongly complementary observables, and complementary observables.}
\label{fig:FTtoHrelationship}
\end{figure}

\noindent We begin with a quick review of Pontryagin duality and the Fourier transform as it relates to quantum computation. A number of different notions related to the Fourier transform on finite abelian groups can be found in mathematics, physics, computer science and quantum computation, so it is useful to clarify them:

\begin{enumerate}
  \item[1.] In mathematics, the Fourier transform is understood through Pontryagin duality.
  \item[2.] In physics and signal processing, the Fourier transform is understood as a transformation of fields/signals from time/space domain to energy\footnote{Or frequency.}/momentum domain.
  \item[3.] In quantum computing, we have Fourier matrices and (complex) Hadamard matrices that correspond to unitary quantum processes.
\end{enumerate}

\noindent This section is ultimately concerned with the first notion, where the Fourier transform is defined on locally compact groups. The other notions are relevant in the context of quantum information. We begin by explaining the relationship of Notion 1 with the others listed above. In what follows, $(G,\cdot,0)$ is a finite abelian group of order $N$, and $\complexs^{\times}=\complexs\setminus\{0\}$ is the multiplicative group of non-zero complex numbers. 

\noindent A \textbf{(multiplicative) character} of $G$ is a group homomorphism $\chi:G\to \complexs^\times$.\footnote{For finite $G$, $\chi$ maps into the subgroup $S^1 \subseteq C^\times$ of unit complex numbers.} If $G = \prod_j \integersMod{n_j}$,\footnote{Which is always true when $G$ is finite, for some family $(n_j)_j$ of positive integers.} the multiplicative characters of $G$ take the following form, for any $h_j \in G$:
\begin{equation}
\label{eqn:FTDefCharacters}
  g_j \mapsto \exp\left[\, i \, \sum_j \frac{2 \pi}{n_j} \left(\modclass{g_jh_j}{n_j}\right) \, \right].
\end{equation}

\noindent The set of characters, with pointwise multiplication defined as $(\chi\cdot\psi)(x):=\chi(x)\psi(x)$, forms a group; this is called the \textbf{Pontryagin dual} (or \textbf{dual group}) of $G$, and is denoted $G^\wedge$. In fact, the Pontryagin construction can be made (contravariantly) functorial on the category \cat{Grp} of groups and group homomorphisms. Define $f^\wedge : G^\wedge \rightarrow H^\wedge$, for any $f: H \rightarrow G$ morphism of abelian groups, as follows:
\begin{equation*}
  f^\wedge( \chi ) = \chi \circ f.
\end{equation*}

\noindent From~\eqref{eqn:FTDefCharacters} it is not hard to see that $G^\wedge \isom G$.\footnote{Note that if $G \isom H$, then there are always exactly as many isomorphisms $G \isom H$ as there are automorphisms $G \isom G$.} However, this isomorphism is \emph{not canonical}. This means that there is no natural way of identifying the multiplicative characters with group elements, and we must keep track of our choice of isomorphism $G^\wedge \isom G$. Remarkably though, there is a canonical isomorphism $G \isom (G^\wedge)^\wedge$ given as follows, making the functor $(-)^{\wedge}:$ \cat{Grp} $\to$ \cat{Grp} is its own (weak) inverse:
\begin{equation*}
  g \mapsto (\chi \mapsto \chi(g)).
\end{equation*}

\newcommand{\FourierTransformSym}[1]{\mathcal{F}_{#1}}
\newcommand{\InverseFourierTransformSym}[1]{\mathcal{F}_{#1}^{-1}}
\newcommand{\FourierTransform}[1]{\mathcal{F}_G[#1]}
\newcommand{\InverseFourierTransform}[1]{\mathcal{F}_G^{-1}[#1]}

\noindent We are now ready to introduce the Fourier transform in the context of Pontryagin duality: this is the most abstract among the notions, and the others are derived from it. Let $\Ltwo{G}$ denote the space of functions $f:G\to \complexs$. These functions are necessarily square-integrable (as $G$ is finite), and thus $\Ltwo{G}$ is an $N$-dimensional complex Hilbert space (and lives in the category $\fdHilbCategory$ of finite-dimensional complex Hilbert spaces and linear maps). 

\begin{definition}
The \textbf{Fourier transform} for a finite abelian group $G$ is a bijection $\mathcal{F}_G:\Ltwo{G}\rightarrow\Ltwo{G^\wedge}$, sending $f:G\to\complexs$ to the $\bar{f}:=\FourierTransform{f}:G^\wedge\to\complexs$ defined as follows:
\begin{equation}\label{eqn:DefTraditionalFT}
  \FourierTransform{f}(\chi) := \frac{1}{N}\sum_{g \in G}\chi^{-1}(g)f(g).
\end{equation}
The \textbf{Inverse Fourier transform} is the inverse bijection $\mathcal{F}_G^{-1}: \Ltwo{G^\wedge} \rightarrow \Ltwo{G}$ and is defined as follows:
\begin{equation}\label{eqn:DefTraditionalInverseFT}
  \InverseFourierTransform{\bar{f}}(g) := \sum_{\chi \in G^\wedge}\chi(g)\bar{f}(\chi).
\end{equation}
\end{definition}

\noindent The Fourier transform is natural.  This means it is invariant under automorphisms of abelian groups (note that the isomorphism $G \isom G^\wedge$ was not). Let $\Psi : G \rightarrow H$ be some isomorphism where $M_\Psi = \Ltwo{G} \rightarrow \Ltwo{H}$ is the corresponding unitary isomorphism that takes $f\mapsto f\circ \Psi$. We then always have:
\begin{equation}\label{eqn:FTcanonicity}
  M_{\Psi^\wedge} \circ \mathcal{F}_H \circ M_\Psi =  \mathcal{F}_G,
\end{equation}

\noindent There are a number of properties of interest for the Fourier transform, some rather straightforward and others more complicated to prove. One of specific interest to this work, because of its wide application and relationship with structures in $\dagger$-SMCs, is the Convolution Theorem. The space $\Ltwo{G}$ comes with a distinguished orthonormal basis, given by the \textbf{delta functions} $(\delta_g)_{g\in G}$ defined as follows.
\begin{equation}
\label{eqn:computationalBasis}
  \delta_g(h):=\begin{cases}
    1, & \text{if $h=g$}.\\
    0, & \text{otherwise}.
  \end{cases}
\end{equation}

\noindent We sometimes refer to this as the \textbf{computational basis}, the name usually given to it in the context of (group-theoretic) quantum algorithms. The computational basis comes with a monoid structure, defined below and with unit $\delta_0$:
\begin{equation}
  \left(\delta_g*\delta_h\right):=\delta_{g+h}.
\end{equation}

\noindent Linearly extended to $\Ltwo{G}$, this structure yields the \textbf{convolution operation} $\left( \Ltwo{G},*,\delta_0 \right)$.
\begin{align}
\label{eqn:convolutionOperation}
\left(f * f'\right) = \left(\sum_{g\in G} f(g) \delta_g \right) * \left( \sum_{g' \in  G} f'(g') \delta_{g'} \right) &= \sum_{g\in G} \sum_{g'\in G'} f(g) f'(g') \delta_{g+g'} \\ 
&= \sum_{h\in G} \left(\sum_{g'\in G} f(h-g') f'(g')\right) \delta_h
\end{align}

\noindent The Fourier transforms of the delta functions yield the following orthogonal basis for $\Ltwo{G^\wedge}$, which we refer to as the \textbf{basis of evaluation functions}:
\begin{align*}
\xi_{g} := \sqrt{N}\tilde{\delta}_{-g} = \left(\chi \mapsto \sum_{h \in G}\chi^{-1}(h)\delta_{-g}(h) \right)= \left(\chi \mapsto \chi(g)\right).
\end{align*}

\noindent The basis of evaluation functions also comes with a monoid structure, with unit $\xi_0: \chi \mapsto 1$:
\begin{equation*}
  \left(\xi_g\cdot\xi_h\right):= \chi \mapsto \xi_g(\chi)\xi_h(\chi) = \chi \mapsto \chi(g)\chi(h).
\end{equation*}

\noindent Functions $F \in \Ltwo{G^\wedge}$ on the dual group have the following expansion in terms of evaluation functions:
\begin{equation*}
  F = \sum_{g\in G} \left( \frac{1}{N}\sum_{\chi \in G^\wedge} F(\chi) \chi^{-1}(g) \right) \xi_g
\end{equation*}

\noindent Linearly extended to $\Ltwo{G^\wedge}$, the monoid structure above yields the \textbf{pointwise multiplication} $\left( \Ltwo{G^\wedge},\cdot,\xi_0 \right)$: 
\begin{align}
\label{eqn:PointwiseMultCharacters}
  \left(F \cdot F' \right) &= \tau \mapsto \sum_{\chi,\kappa \in G^\wedge}  F(\chi)  F'(\kappa) \left(\frac{1}{N} \sum_{g\in G} \chi^{-1}(g)\tau(g)\right) \left(\frac{1}{N} \sum_{g'\in G}\kappa^{-1}(g') \tau(g') \right) \\ &= \tau \mapsto F(\tau) F'(\tau)
\end{align}
We use the (easy to check) fact that, for any $\chi,\tau \in G^\wedge$, the expression $\frac{1}{N} \sum_{g\in G}\chi^{-1}(g) \tau(g)$ yields $1$ if $\tau = \chi$ and $0$ otherwise (this is usually referred to as \textbf{orthogonality of (multiplicative) characters}).

\begin{theorem}[Convolution Theorem]
The Fourier transform is a monoid isomorphism in $\fdHilbCategory$, from the convolution monoid $\left( \Ltwo{G},*,\delta_0 \right)$ to the pointwise multiplication monoid $\left( \Ltwo{G^\wedge},\cdot, \xi_0 \right)$. This statement amounts exactly to the following expression (for every $f\in \Ltwo{G}$), which is the usual formulation of the Convolution Theorem: 
\begin{equation}\label{eqn:ConvolutionTheorem}
  \mathcal{F}_G (f') \cdot \mathcal{F}_G (f) = \mathcal{F}_G (f * f').
\end{equation}
\end{theorem}

\noindent This concludes our presentation of the Fourier transform in the context of Pontryagin duality. A further reference for details of the topics in this presentation is Rudin~\cite{rudin1962fourier}. The Fourier transform finds wide applicability in signal processing, physics, engineering and the applied sciences, but the full formulation based on Pontryagin duality is rarely used, if mentioned at all. In the engineering context, one usually considers periodic real-valued or complex-valued functions on a $D$-dimensional space, discretized in a rectangular $D$-dimensional lattice, and defines the (Discrete) Fourier transform as a transformation on them. Due to the periodicity conditions, complex-valued functions on a rectangular $D$-dimensional lattice can be equivalently seen as living in $\Ltwo{G}$, where $G = \prod_{j=1}^D \integersMod{n_j}$ and $n_j$ is the number of lattice sites along the $j$-th dimension. The  Fourier transform $\mathcal{G} : \Ltwo{G} \rightarrow \Ltwo{G^\wedge}$ defined above sends these functions onto functions on another, isomorphic $D$-dimensional lattice corresponding to $G^\wedge$. In order to obtain functions living back on the original lattice, one \emph{fixes an isomorphism} $\Psi : G \rightarrow G^\wedge$ (traditionally the one from Equation \ref{eqn:FTDefCharacters}), and defines the Discrete Fourier transform as the following transformation on $\Ltwo{G}$:
\begin{equation}\label{eqn:FTDefDFT}
  \mathbf{F} := f \mapsto \mathcal{F}_G(f) \circ \Psi.
\end{equation}

\noindent This definition has the advantage of working with functions on the same lattice, but the disadvantage of implicitly depending on the choice $\Psi$ of isomorphism.\footnote{This is a common issue in signal processing and physics, where it is related to the symmetry group of the underlying space and the choice of units of measure for energy/frequency. We will not discuss this further.} The transformation $\mathbf{F}$ from Equation \ref{eqn:FTDefDFT} is in fact a unitary automorphism of $\Ltwo{G}$. Its matrix $(\mathbf{F}_{hg})_{h,g \in G}$ in the computational basis is:
\begin{equation} \label{eqn:HadamardMatrixDef}
  \mathbf{F}_{hg} = \exp\left[\, i \, \sum_j \frac{2 \pi}{n_j} \left(\modclass{g_jh_j}{n_j}\right) \, \right] 
\end{equation}
and it is called a \textbf{Fourier matrix} in the context of quantum computing. Fourier matrices correspond to a Fourier transform along with a choice of the isomorphism. Thus the Fourier matrices, exactly like the definition of the Discrete Fourier transform above, are non-canonical, and depend on an implicit choice of isomorphism $\Psi$. This contrasts with the Fourier transform, which is itself canonical.

Fourier matrices are a subclass of more general \textbf{complex Hadamard matrices}: orthogonal matrices\footnote{Here an orthogonal matrix $H$ is a square matrix such that $H^TH=HH^T=\mathbbm{1}$.} whose complex entries are unimodular. In particular, (real) \textbf{Hadamard matrices} are orthogonal matrices with entries $\pm1$. Having defined these four different terms (the Fourier transform, Fourier matrices, Hadamard matrices, and complex Hadamard matrices, see Figure~\ref{fig:FTtoHrelationship}) we will clarify a few ways that they appear in quantum computation.

There is a particularly interesting reason the lack of canonicity of Fourier matrices is not usually an issue in quantum computing. Most of the algorithms are traditionally formulated for qubits, and the state-space of a $D$-qubit system is isomorphic to $\Ltwo{G}$ for $G = \prod_{j=1}^D \integersMod{2}$. The group $\integersMod{2}$ has a unique automorphism (the identity), and thus a unique isomorphism $\integersMod{2} \rightarrow \integersMod{2}^\wedge$, resulting in the familiar matrix
\begin{equation}
\label{hmat}
\begin{pmatrix}1 & 1 \\
1 & -1 \\
\end{pmatrix},
\end{equation}
which is both the only Fourier matrix on a two dimensional system and, in fact, a (real) Hadamard matrix. There is then a unique isomorphism $\Psi : \integersMod{2}^N \rightarrow (\integersMod{2}^N)^\wedge$ which can be obtained by local qubit operations only, namely the $N$-fold tensor product of the isomorphism in \ref{hmat}; if multi-qubit operations are allowed, however, the isomorphism is not unique. We stress that for general groups, i.e. for combinations of quantum systems where some have dimensions larger than two, the Fourier transform in terms of Pontryagin duality does not fix a unique Fourier matrix (not even requiring that it is obtained by local operations only). Furthermore, not all complex Hadamard matrices correspond to a Fourier matrix. In this work, we use Fourier transform to refer explicitly to the canonical one defined in terms of Pontryagin duality.\\

There are a number of existing generalizations in the literature of the Fourier transform presented here that we make contact with to varying degrees. 
\begin{enumerate}
\item[1.] Pontryagin theory can be extended from finite abelian groups to arbitrary locally compact abelian groups equipped with the Haar measure: the groups $G$ and $G^\wedge$ are not necessarily isomorphic (e.g. $\reals^\wedge = \reals$ but $\integers^\wedge = S^1$), but the Fourier transform is still a canonical isomorphism between $\Ltwo{G}$ and $\Ltwo{G^\wedge}$, and it's still true that $(G^\wedge)^\wedge = G$. 

\item[2.] The representation theory can be extended from abelian to arbitrary locally compact groups by observing that $\Ltwo{G}$ is always a $C^\star$ algebra, and considering the Gelfand-Naimark representation. In the abelian case, this representation coincides with the Fourier transform.

\item[3.] Tannaka-Krein duality provides a different generalisation from compact abelian groups to arbitrary compact groups: the finite-dimensional linear representations of a compact group $G$ form a symmetric monoidal category $\Pi(G)$, generalising $G^\wedge$, with representations $R:G \rightarrow \Endoms{}{V_R}$ as objects, intertwiners (linear maps $f: V_R \rightarrow V_S$ s.t. $f \circ R(g) = S(g) \circ f$ for all $g\in G$) as morphisms and tensor product of representations as monoidal tensor. The category $\Pi(G)$ comes with a complex conjugation operation on morphisms, and a theorem of Tannaka shows that the set $\Gamma(\Pi(G))$ of all self-conjugate monoidal natural transformations $\idm{\Pi(G)} \rightarrow \idm{\Pi(G)}$ forms (once equipped with composition of natural transformations and an appropriate topology) a compact group isomorphic to $G$. A generalisation of Tannaka-Krein duality to braided monoidal categories appears in the representation theory of Drinfeld-Jimbo quantum groups. In this work, we do not deal with either Tannaka-Krein theory or Drinfeld-Jimbo quantum groups.
For more on quantum groups and their connection to Hopf algebras see, 
e.g.~\cite{cartier2007primer,street2007quantum}.
\end{enumerate}

\section{Overview}
\label{section_Intro}

This work elucidates the connection between the Fourier transform - as described in the previous section - and strongly complementary observables, i.e. Hopf algebras in dagger symmetric monoidal categories. We summarize our outline and main results in the following paragraphs. Throughout we assume familiarity with the standard notions of Frobenius algebras in $\dagger$-symmetric monoidal categories, which can be reviewed in Appendix~\ref{app:basic}.

In Section \ref{section_InternalGroups} we cover the definition of strong complementarity~\cite{coecke2011interacting}, which has been used in the foundations of quantum mechanics to study non-locality~\cite{coecke2012strong, gogioso2015mermin}, quantum secret sharing~\cite{gogioso2015mermin, zamdzhiev2012abstract}, and blackbox quantum algorithms~\cite{vicary-tqa, zeng2014abstract, zeng2015models}. This allows a generalization beyond $\fdHilbCategory$ to strongly complementarity pairs of a quasi-Special $\dagger$-Frobenius Algebra ($\dagger$-qSFA or $\dagger$-qSCFA if commutative) and a $\dagger$-SCFA. We use this generalization to embed finite groups in arbitrary dagger symmetric monoidal categories.

In Section \ref{section_AbelianGroups_FourierTransform} the traditional definition of Fourier Transform from Section \ref{section_FourierTransformIntro} is lifted to general symmetric monoidal categories. We construct the accompanying general definitions for multiplicative characters and the abelian Fourier transform in this setting.  These results allow us to provide categorical versions, with abstract proofs, of the Fourier inversion theorem, the Convolution theorem, and Pontryagin duality that are all based on a strongly complementary pair of observables.

In Section \ref{section_RelFT}, we study $\RelCategory$ as an example setting for our categorical Fourier transform. This example is of particular interest as it often acts as a toy model for quantum theory~\cite{evans2009classifying, cqm-notes, Msc-ClassicalStructuresRel, zeng2015models}.  We find that, while a generalized Fourier matrix is not suitably defined, a Fourier transform can be.

In Section \ref{section_EnrichedPontryaginDuality}, enriched category theory is used to extend Pontryagin duality, i.e. the abelian Fourier transform, from isomorphism of $\LtwoSym$-spaces to isomorphism of free modules of arbitrary semirings.

In Section \ref{section_NonAbelianFourierTransform}, we look further into the non-abelian case of our construction, and provide a fully categorical generalisation of the Gelfand-Naimark theorem from finite-dimensional Hilbert spaces and C*-algebras to algebras on finite-dimensional free modules of arbitrary semirings.

These results both move Fourier theory into a new mathematical setting and capture the structural connection between quantum theory and the Fourier transform.  Though this connection has been much exploited in quantum algorithms, this work is the first abstract presentation that shows its place in the structure of quantum theory, i.e. alongside strongly complementary observables.

Finally, we note that throughout this work familiarity with the diagrammatic presentation of morphisms in symmetric monoidal categories is assumed. See Selinger~\cite{selinger2011survey} for general reference.

\section{Strong Complementarity and Internal Groups}
\label{section_InternalGroups}

\noindent We will eventually show that the Fourier transform is related to a special type of complementarity called strong complementarity. In this section we introduce the strongly complementary notion for quasi-special commutative Frobenius algebras. The use of quasi-special rather than the special Frobenius algebras that are often used by other authors is for convenience and allows us to lump scalars together rather than having to keep track of them at every step.

\begin{definition}\label{def:QuasiSpecial}
A \textbf{quasi-special} $\dagger$-Frobenius algebra \whitefrob{A} is one that satisfies the following equation for some invertible scalar $N$:
\begin{equation}\label{eqn:QuasiSpecialDef}
\begin{aligned}
\begin{tikzpicture}
\draw (0,0.25) to (0,1) to [out=\nwangle, in=down] (-0.5,1.5) to [out=up, in=\swangle] (0,2) to (0,2.75);
\draw (0,1) to [out=\neangle, in=down] (0.5,1.5) to [out=up, in=\seangle] (0,2);
\node at (0,1) [dot, fill = \classicalStructColour] {};
\node at (0,2) [dot, fill = \classicalStructColour] {};
\end{tikzpicture}
\end{aligned}
\quad=\quad
  \begin{aligned}
  \begin{tikzpicture}
  \node [dot, fill = \classicalStructColour] at (-1.5,1.5) {$N$};
  \draw (-0.5,0) to (-0.5,3);
  \end{tikzpicture}
  \end{aligned}
\end{equation}
We will use the shorthand $\dagger$-qSFA, and refer to $N$ as the \textbf{normalisation factor} for the $\dagger$-qSFA.
\end{definition}

\noindent These $\dagger$-qSFA's can be thought of as generalized orthogonal bases that are normalize-able (as long as the square root of the scalar $\sqrt{N}$ is invertible) even if they are not normalized.  While classical states can be defined in any $\dagger$-SMC, the following definition for matching families requires an appropriate \textit{zero scalar}. For this and other reasons, we consider categories enriched over commutative monoids,\footnote{Refer to~\cite{zeng2015abstract} for more detail on enriched process theories.} i.e. where homsets come with a commutative monoid structure $(\Hom{\CategoryC}{A}{B},+,0)$, and we require the appropriate distributivity laws between the tensor product and the monoidal structure:
\begin{align}
(f+g) \tensor h &= (f\tensor h) + (g \tensor h)\\
f \tensor (g+h) &= (f \tensor g) + (f \tensor h)\\
0 \tensor f &= 0\\
f \tensor 0 &= 0
\end{align}

\noindent We will refer to these as \textbf{distributively $\ComMonCategory$-enriched} $\dagger$-SMCs.

\begin{definition}
\label{def:matchables}
Let $\ket{x}_{x \in X}$ be a finite family of states $I \rightarrow \SpaceG$ in a $\dagger$-SMC which is distributively $\ComMonCategory$-enriched. A \textbf{matchable family} $\ket{x}_{x \in X}$ for a monoid $(\SpaceG,\XmultSym, \XunitSym)$ are those for which the following holds for all $x,y \in X$:
\begin{equation}\label{eqn:matchables}
    \XmultSym \circ \left( \ket{x} \tensor \ket{y} \right) = 
    \begin{cases} 
        \ket{x} \text{ if } \ket{x} = \ket{y}\\
        0 \text{ otherwise }
    \end{cases}
\end{equation}
\end{definition}

\noindent We re-emphasise that while Coecke et al. show that $\dagger$-qSCFA's correspond to bases in $\fdHilbCategory$ by \cite[Thm 5.1]{coecke2013new}, the general notion of a (orthogonal) basis is somewhat different.

\begin{definition}\label{def:basis}
A finite family of states $\ket{x}_{x\in X}: I \rightarrow \SpaceH$ is a \textbf{(orthogonal) basis} (for~$\SpaceH$) if it satisfies the following conditions:
\begin{enumerate}
\item[(i)] Orthogonality, i.e. $\langle y|x\rangle = 0$ if $x \neq y$ (where $\bra{y}$ stands for $\ket{y}^\dagger$).
\item[(ii)] Completeness, i.e. for every $f,g : \SpaceH \rightarrow \SpaceH'$ we have that $\forall x:X \, f \ket{x} = g \ket{x}$ implies $f=g$. 
\end{enumerate}
A finite family of co-states $\bra{x}_{x\in X}: \SpaceH \rightarrow I$ is a \textbf{(orthogonal) cobasis} (for $\SpaceH$) if the family of states $\ket{x}_{x\in X}: I \rightarrow \SpaceH$ is a basis.
\end{definition}
\noindent When the classical states for a classical structure form a basis in this manner, the algebra has ``enough classical points". In $\fdHilbCategory$, this is the usual linear-algebraic notion of orthogonal basis. Strong complementarity was originally introduced by Coecke and Duncan in~\cite{coecke2011interacting} as the additional rule that makes classical structures into a Hopf algebra.\footnote{ They are also studied in this form, though in a different framework, as a foundation for graphical linear algebra by Bonchi et al.~\cite{bonchi2014interacting}.}

\begin{definition}\label{def:StrongComplementarity}
A pair of $\dagger$-qSFAs \whitefrob{A} and \blackfrob{A}, henceforth written as \scpair, is \textbf{strongly complementary} if they are coherent (Definition \ref{def:coherence})  and  satisfy the following \textbf{bialgebra equation} (\ref{eqn:bialgebraEqns}):
\begin{equation}
\label{eqn:bialgebraEqns}
\input{modules/pictures/bialgebraEqns.tex}
\end{equation}
\end{definition}
\noindent Though this definition is usually given for classical structures (Appendix~\ref{app:basic}), we generalise to $\dagger$-qSFAs to include non-commutative algebras and, hence, our later construction of a generalized non-abelian Fourier transform. Coecke et al.~\cite[Thm 3.6]{coecke2012strong} also give a quick proof that the name is an apt one, i.e. that strongly complementary structures are also complementary in the sense of mutually unbiased bases.

\begin{definition}\label{def:Antipode} 
Given a strongly complementary pair of $\dagger$-FAs $\scpair$ on some object $\SpaceG$ in a $\dagger$-SMC, the \textbf{antipode} $\hbox{\input{modules/symbols/antipodeSym.tex}}\!:\SpaceG \rightarrow \SpaceG$ is defined to be the following map:
\begin{equation}
    \input{modules/pictures/Antipode.tex}
\end{equation}
\end{definition}

\begin{lemma}\label{lemma_AntipodeInverse}
Given a strongly complementary pair of $\dagger$-FAs \scpair~on some object $\SpaceG$ in a $\dagger$-SMC, the \textbf{antipode inverse} $\hbox{\input{modules/symbols/antipodeSym.tex}}\!^{-1}:\SpaceG \rightarrow \SpaceG$ is the following map:
\begin{equation}
    \input{modules/pictures/AntipodeInverse.tex}
\end{equation}
Furthermore, if at least one of the two $\dagger$-FAs has a finite matchable family that forms a basis, then the antipode is self-adjoint\footnote{Recall that under certain assumptions on the $\dagger$-FAs that are common in process theories, the antipode is self-adjoint~\cite[Lem. 7.2.6]{kissinger2012pictures}, though we will work in the more general setting.} and unitary, i.e. antipode and antipode inverse coincide.
\end{lemma}
\begin{proof}
    The fact that $\hbox{\input{modules/symbols/antipodeSym.tex}}\!^{-1}$ as defined is indeed the inverse of $\hbox{\input{modules/symbols/antipodeSym.tex}}\!$ is an immediate consequence of the Frobenius law (one application per colour). Now suppose without loss of generality that the matchable states $\ket{g}_{g \in G}$ of $\hbox{\input{modules/symbols/XdotSym.tex}}\!$ form a basis, and remember that $\hbox{\input{modules/symbols/ZdotSym.tex}}\!$ acts as some (possibly non-abelian) group $(G,\cdot,1)$ on them. Then $\bra{h}\, \hbox{\input{modules/symbols/antipodeSym.tex}}\! \ket{g} = \braket{h \cdot g}{1}$ and $\bra{h} \; \hbox{\input{modules/symbols/antipodeSym.tex}}\!^{-1} \ket{g} = \braket{1}{h \cdot g}$ and $\bra{h}\; \hbox{\input{modules/symbols/antipodeSym.tex}}\!^\dagger \ket{g} = \braket{1}{g \cdot h}$ and $\bra{h}\; (\hbox{\input{modules/symbols/antipodeSym.tex}}\!^{-1})^\dagger \ket{g} = \braket{g \cdot h}{1}$ coincide for all $g,h \in G$, proving that $\hbox{\input{modules/symbols/antipodeSym.tex}}\! = \hbox{\input{modules/symbols/antipodeSym.tex}}\!^{-1} = \hbox{\input{modules/symbols/antipodeSym.tex}}\!^{\dagger} = \ket{g} \mapsto \ket{g^{-1}}$ for all $g\in G$.
\end{proof}

\noindent Coecke and Duncan \cite{coecke2011interacting} showed that strongly complementary classical structures have a specific relationship between their phase groups and classical states.

\begin{theorem}
Let $\hbox{\input{modules/symbols/ZdotSym.tex}}\!$ and $\hbox{\input{modules/symbols/XdotSym.tex}}\!$ be a pair of strong complementary classical structures with finite numbers of classical states. Then $K_{\hbox{\input{modules/symbols/XdotSym.tex}}\!}\subseteq P_{\hbox{\input{modules/symbols/ZdotSym.tex}}\!}$, i.e. the classical states of the black classical structure form a subgroup of the phase group of the white classical structure. The converse is true when $\XcomultSym$ has enough classical points.
\end{theorem}

\noindent This leads to Kissinger's classification \cite[Cor. 3.10]{coecke2012strong} of strongly complementary classical structures in $\fdHilbCategory$:
\begin{corollary}
\label{col:SCclassification}
Every pair of strongly complementary classical structures in $\fdHilbCategory$ is
of the following form: 
\begin{equation}
\left\{\begin{array}{cl}
\tinycomultblack   & :: \ket{g}\mapsto \ket{g}\otimes \ket{g}\vspace{1mm}\\
\tinycounitblack & :: \ket{g}\mapsto 1
\end{array}\right.
\quad
\left\{\begin{array}{cl}
\ZmultSym   & :: \ket{g}\otimes \ket{h} \mapsto {1\over\sqrt{D}} \ket{g + h}\vspace{1mm} \\
\ZunitSym & :: 1 \mapsto \sqrt{D} \ket{0}
\end{array}\right.
\end{equation}
where $(G =\{g, h, \ldots\}, +, 0)$ is finite abelian. Conversely, each such pair is strongly complementary.  
\end{corollary}

\noindent We use strongly complementary structures to embed groups into an arbitrary $\dagger$-SMC.

\begin{definition}\label{def:AbClassicalGroup} An \textbf{internal group}, denoted by $(\,\mathcal{G} , \hbox{\input{modules/symbols/timemultSym.tex}}\!, \hbox{\input{modules/symbols/timeunitSym.tex}}\!, \tinycomultblack, \tinycounitblack$ ) or  $(\,\mathcal{G}, \hbox{\input{modules/symbols/ZdotSym.tex}}\!, \hbox{\input{modules/symbols/XdotSym.tex}}\!)$ when no confusion should arise), consists of a strongly complementary pair on the same object $\mathcal{G}$ of a $\dagger$-SMC and
\begin{enumerate}
\item A $\dagger$-qSFA, the \textbf{group structure}, denoted by $\hbox{\input{modules/symbols/ZdotSym.tex}}\!$.
\item A $\dagger$-qSCFA, the \textbf{point structure}, denoted by $\hbox{\input{modules/symbols/XdotSym.tex}}\!$.
\end{enumerate}
The multiplication and unit for the group structure are called \textbf{group multiplication} and \textbf{group unit}, and the antipode $\hbox{\input{modules/symbols/antipodeSym.tex}}\!$ for the pair is called the \textbf{group inverse}. An \textbf{abelian internal group} is one where the group structure is commutative.
\end{definition}

\noindent The internal groups in an $\dagger$-SMC form a category $\cat{Grp}[\cat{C}]$, with objects given by the strongly complementary pairs 
$(\,\SpaceG , \hbox{\input{modules/symbols/ZdotSym.tex}}\!, \hbox{\input{modules/symbols/XdotSym.tex}}\!)$, and morphisms 
$(\,\SpaceG , \hbox{\input{modules/symbols/ZdotSym.tex}}\!, \hbox{\input{modules/symbols/XdotSym.tex}}\!) \rightarrow (\,\SpaceG' , \hbox{\input{modules/symbols/ZaltdotSym.tex}}\!, \hbox{\input{modules/symbols/XaltdotSym.tex}}\!)$ given by 
$f: \SpaceG \rightarrow \SpaceG'$ in $\cat{C}$ that are co-monoid homomorphisms 
$f: (\tinycomultblack,\tinycounitblack) \rightarrow (\XaltcomultSym,\XaltcounitSym)$ and monoid homomorphisms 
$f: (\ZmultSym,\ZunitSym) \rightarrow ( \ZaltmultSym,\ZaltunitSym)$;
the abelian internal groups form a full subcategory $\cat{AbGrp}[\cat{C}]$. We will refer to these morphisms as \textbf{internal group homomorphisms}, both when seen as morphisms in $\cat{Grp}[\cat{C}]$ and in $\cat{C}$.

\begin{theorem}\label{thm_InteralGroupsTraditionalGroups} 
        If $(\,\SpaceG , \hbox{\input{modules/symbols/ZdotSym.tex}}\!, \hbox{\input{modules/symbols/XdotSym.tex}}\!)$  is an (abelian) internal group in any $\dagger$-SMC, then $(\hbox{\input{modules/symbols/timemultSym.tex}}\!, \hbox{\input{modules/symbols/timeunitSym.tex}}\!)$ acts as an (abelian) group $G$ on the classical points of $(\;\hbox{\input{modules/symbols/timediagSym.tex}}\!, \hbox{\input{modules/symbols/trivialcharSym.tex}}\!)$, henceforth the \textbf{group elements}. Furthermore, this correspondence yields an equivalence between the the category of (abelian) internal groups in $\fdHilbCategory$ and the category of finite (abelian) groups.
\end{theorem}

\noindent In $\fdHilbCategory$, the point structure $(\hbox{\input{modules/symbols/XdotSym.tex}}\!)$ characterises the group elements $\ket{g}_{g\in G}$ as an orthonormal basis for $\SpaceG$.  This is the basis of delta functions from Equation \ref{eqn:computationalBasis}, with $\ket{g} := \delta_g$. The corresponding isomorphism $\Ltwo{G} \isom \SpaceG$ sends any square-integrable $f: G \rightarrow \complexs$ to the vector $\ket{f} \in \SpaceG$ defined by $\ket{f} = \sum_{g\in G} f(g) \ket{g}$. Also under this isomorphism, the multiplicative fragment $(\, \SpaceG,\hbox{\input{modules/symbols/timemultSym.tex}}\!,\hbox{\input{modules/symbols/timeunitSym.tex}}\!)$ of the internal group structure acts as the convolution operation from Equation \ref{eqn:convolutionOperation}. Simply put, an internal group  $\mathbb{G} = (\,\SpaceG,\hbox{\input{modules/symbols/ZdotSym.tex}}\!,\hbox{\input{modules/symbols/XdotSym.tex}}\!)$ in $\fdHilbCategory$ consists of:
\begin{enumerate}
\item[(i)] a space $\SpaceG$
\item[(ii)] a distinguished orthogonal basis, encoded by the $\dagger$-qSCFA $\hbox{\input{modules/symbols/XdotSym.tex}}\!$
\item[(iii)] a group structure on that basis, encoded by the $\dagger$-qSFA $\hbox{\input{modules/symbols/ZdotSym.tex}}\!$
\end{enumerate}

\noindent From the point of view of the category $\cat{Grp}[\cat{C}]$, $\mathbb{G}$ should be understood as the group $G$ encoded by $\hbox{\input{modules/symbols/ZdotSym.tex}}\!$, while from the point of view of the category $\fdHilbCategory$ it should be considered as endowing $\SpaceG$ with the structure of $\Ltwo{G}$.\footnote{In this correspondence, the orthogonal basis in $\SpaceG$ corresponds to the basis of delta functions in $\Ltwo{G}$, as given in Equation \ref{eqn:computationalBasis}. The groups structure given by $\hbox{\input{modules/symbols/ZdotSym.tex}}\!$ corresponds to the convolution operation from Equation \ref{eqn:convolutionOperation}.} As we abstract away from Hilbert spaces, we will take this conceptual standpoint. Sometimes, when talking about an internal group $\mathbb{G} = (\,\SpaceG,\hbox{\input{modules/symbols/ZdotSym.tex}}\!,\hbox{\input{modules/symbols/XdotSym.tex}}\!)$, we will refer to states $\ket{f} : I \rightarrow \SpaceG$ as \textbf{states of $\mathbb{G}$}, generalising square-integrable functions $f\in \Ltwo{G}$.


\section{Abelian Fourier transform}
\label{section_AbelianGroups_FourierTransform}

The previous section provides us with the basic tools to do group theory in arbitrary $\dagger$-SMCs. This section contains the first of the main results of this work. It defines general Pontryagin duals, introduces the Fourier Transform and connects it to the more traditional theory presented in Section \ref{section_FourierTransformIntro}. We begin by introducing multiplicative characters as co-states, building on ideas in~\cite{vicary-tqa} and~\cite{zeng2014abstract}.

The use of the $\Ltwo{G}$ notation in this section is consistent with the fact that $\LtwoSym$-spaces over finite groups are exactly finite-dimensional Hilbert spaces that come with a canonical choice of basis (the group elements) and a group operation over them. Throughout, we have identified $\Ltwo{\hat{G}}\isom \Ltwo{G}^\star$ as the multiplicative characters are a basis of $\Ltwo{G}^\star$.

\begin{definition}\label{def_MultiplicativeCharacters}
A \textbf{multiplicative character} for a monoid $(\, \SpaceG, \XmultSym, \XunitSym)$ in a $\dagger$-SMC is a monoid homomorphism from $(\,\hbox{\input{modules/symbols/timemultSym.tex}}\!, \hbox{\input{modules/symbols/timeunitSym.tex}}\!)$ to the canonical monoid on the trivial object $\tensorUnit$ induced by the unitors, or equivalently it is a co-state $
		\renewcommand{\tempSymLabel}{}
		\hbox{\input{modules/symbols/multCharacterEffectSym.tex}}
	: \SpaceG \rightarrow \tensorUnit$ satisfying the following equations:
	\begin{equation}\label{eqn:MultCharDef}
		\input{modules/pictures/MultCharDef.tex}
	\end{equation}
\end{definition}

\begin{lemma}\label{thm_AbCopiablesMultiplicativeCharacters}
If $(\,\SpaceG,\hbox{\input{modules/symbols/XdotSym.tex}}\!,\hbox{\input{modules/symbols/ZdotSym.tex}}\!)$ is an internal group in a $\dagger$-SMC, then the classical states of the group structure are exactly the (adjoints of its) multiplicative characters. In the case of $\fdHilbCategory$, the group structure of an abelian internal group thus characterises the (group theoretic) multiplicative characters of $G$ as a co-basis for $\SpaceG$, all characters having norm $N$, the normalisation factor for $\hbox{\input{modules/symbols/XdotSym.tex}}\!$.
\end{lemma}
\begin{proof} 
  The first part is immediate, the second follows from the results in \cite{coecke2013new}.
\end{proof}

\noindent In $\fdHilbCategory$, the multiplicative characters of an internal group $(\,\SpaceG,\hbox{\input{modules/symbols/XdotSym.tex}}\!,\hbox{\input{modules/symbols/ZdotSym.tex}}\!)$ are co-states $\SpaceG \rightarrow \complexs$, while the multiplicative characters defined in Section \ref{section_FourierTransformIntro} are group homomorphisms $G \rightarrow \complexs^\times$. Under the isomorphism $\Ltwo{G} \isom \SpaceG$ given by the point structure, the multiplicative characters of the internal group are exactly the linear extensions to $\Ltwo{G}$ of the multiplicative characters of $G$. If the internal group is abelian, then the multiplicative character are exactly the adjoints of the unique orthogonal basis associated with the $\hbox{\input{modules/symbols/XdotSym.tex}}\!$ structure.

\begin{theorem}\label{thm_PontryaginDualsSMC}
Let $\mathbb{G} = (\,\SpaceG,\hbox{\input{modules/symbols/XdotSym.tex}}\!,\hbox{\input{modules/symbols/ZdotSym.tex}}\!)$ be an abelian internal group in a $\dagger$-SMC $\CategoryC$.
Then $\mathbb{G}^\wedge := (\,\SpaceG,\hbox{\input{modules/symbols/ZdotSym.tex}}\!,\hbox{\input{modules/symbols/XdotSym.tex}}\!)$ is an abelian internal group in the $\dagger$-SMC $\OpCategory{\CategoryC}$, and we shall refer to it as the \textbf{Pontryagin dual} of $(\,\SpaceG,\hbox{\input{modules/symbols/XdotSym.tex}}\!,\hbox{\input{modules/symbols/ZdotSym.tex}}\!)$. The group elements of $(\,\SpaceG,\hbox{\input{modules/symbols/ZdotSym.tex}}\!,\hbox{\input{modules/symbols/XdotSym.tex}}\!)$ are exactly the multiplicative characters of $(\,\SpaceG,\hbox{\input{modules/symbols/XdotSym.tex}}\!,\hbox{\input{modules/symbols/ZdotSym.tex}}\!)$ -- this is to say that $\hbox{\input{modules/symbols/timediagSym.tex}}\!$ acts as a group, the \textbf{pointwise multiplication} group, on the multiplicative characters, with the \textbf{trivial character} $\hbox{\input{modules/symbols/trivialcharSym.tex}}\!$ as unit. The antipode acts again as group inverse. 
\end{theorem}

\noindent It is worth clarifying that the pointwise multiplication of Equation \ref{eqn:PointwiseMultCharacters} is different from the pointwise multiplication of Theorem \ref{thm_PontryaginDualsSMC}: the former is a pointwise product of functions of characters, and would correspond to the co-monoid $(\;\hbox{\input{modules/symbols/timecomultSym.tex}}\!,\hbox{\input{modules/symbols/timecounitSym.tex}}\!)$ (because the group multiplication duplicates multiplicative characters), while the latter is a pointwise product of functions of group elements, and thus corresponds to the co-monoid $(\;\hbox{\input{modules/symbols/timediagSym.tex}}\!,\hbox{\input{modules/symbols/trivialcharSym.tex}}\!)$ (which duplicates group elements). Also, note that $(\mathbb{G}^\wedge)^\wedge = \mathbb{G}$, like in the traditional formulation of Pontryagin duality.

The usual formulation of the Fourier Transform, and of its properties, involves several summations, but a careful analysis shows that they all boil down to appropriate resolutions of the identity, like $\frac{1}{N} \sum_{g\in G} \ket{g}\bra{g} = \id{\Ltwo{G}}$, and to various formulations of orthogonality of characters. The following lemma shows that, from a categorical perspective, the two are equivalent.

\begin{lemma}\label{lemma_BasisResolutionPartition}
Let $\hbox{\input{modules/symbols/ZdotSym.tex}}\!$ be a $\dagger$-qSFA on an object $\SpaceG$ in a $\dagger$-SMC which is distributively $\ComMonCategory$-enriched, and let $\ket{x}_{x\in X}$ be a finite family of classical states for the co-monoid $(\ZcomultSym,\ZcounitSym)$ such that
\begin{enumerate}
\item[(a.)] the family is \textbf{orthogonal}, i.e. $\braket{x'}{x} = 0$ (the zero scalar) for all $x \neq x'$
\item[(b.)] the family is \textbf{normalisable}, i.e. $\braket{x}{x}$ is an invertible scalar for all $x$.
\end{enumerate} 
Then the following are equivalent:
\begin{enumerate}
\item[(i)] The classical states $\ket{x}_{x\in X}$ form a (orthogonal) basis, as per Definition \ref{def:basis}.
\item[(ii)] The classical states $\ket{x}_{x\in X}$ form an  \textbf{(orthogonal) resolution of the identity}:
	\begin{equation}\label{eqn:ResolutionId}
		\hbox{\input{modules/pictures/MultCharResolutionId.tex}}
	\end{equation}
\item[(iii)] The adjoints of the classical states form an \textbf{(orthogonal) partition of the counit}, i.e. they satisfy the following equation:
\begin{equation}\label{eqn:PartitionCounit}
\frac{1}{N}\sum_x \!\!\!\!\text{
		\renewcommand{\tempSymLabel}{x}
		\hbox{\input{modules/symbols/effectSym.tex}}
	}\!\! = \ZcounitSym
\end{equation}
\end{enumerate}
\end{lemma}

\begin{proof}

Since we have assumed that $\braket{x}{x}$ is invertible, $\braket{x}{x} = N$ as observed in Section \ref{section_InternalGroups}.
\begin{itemize}

\item $(i) \implies (ii)$ Suppose that the classical states form a basis, i.e. further to orthogonality of the family suppose that $\forall \chi , f \circ \ket{x} = g \circ \ket{x}$ implies $f=g$ (completeness). Then we get, for all $x'$:
\begin{equation*}
\left(\frac{1}{N} \sum_x \, \ket{x}\bra{x}\right) \circ \ket{x'} = \frac{1}{N} \ket{x'} \braket{x'}{x'} = \ket{x'} = \id{\SpaceG} \circ \ket{x'}
\end{equation*}
where we have used orthogonality. We conclude $(ii)$ by completeness.

\item $(ii) \implies (iii)$ By using the fact that $\ZcounitSym \circ \ket{x} = 1$ for all $x:X$ we immediately get $(iii)$:
\begin{equation*}
\ZcounitSym \circ \left(\frac{1}{N} \sum_x \ket{x} \bra{x} \right) = \frac{1}{N} \sum_x \left( \ZcounitSym \circ \ket{x} \right) \bra{x} = \frac{1}{N} \sum_x \bra{x}
\end{equation*}
 
\item $(ii) \implies (i)$ All we have to prove is completeness, as orthogonality of the family $\ket{x}_{x\in X}$ was assumed as a hypothesis of the lemma. Assume $f \circ \ket{x} = g \circ \ket{x}$ for all $x$, then we get:
\begin{equation*}
\frac{1}{N} \sum_x f \circ \ket{x} \bra{x} = \frac{1}{N} \sum_x g \circ \ket{x} \bra{x}
\end{equation*}
But the LHS is $f \circ \id{\SpaceG}$, i.e. $f$, and the RHS is $g \circ \id{\SpaceG}$, i.e. $g$.

\item $(iii) \implies (ii)$ Assume that $\ZcounitSym = \frac{1}{N} \sum_x \bra{x}$, then we get (using Frobenius law in the first equality):
\begin{equation*}
\id{\SpaceG} = (\ZcounitSym \tensor \id{\SpaceG})\circ (\ZmultSym \tensor \id{\SpaceG}) \circ (\id{\SpaceG} \tensor \ZcomultSym) \circ (\id{\SpaceG} \tensor \ZunitSym) = \frac{1}{N} \sum_x \frac{1}{N} \sum_{x'} \ket{x'} \braket{x'}{x} \bra{x} = \frac{1}{N} \sum_x \ket{x} \bra{x}
\end{equation*}
 
\end{itemize}

\noindent As a final remark, please note that orthogonality, assumed separately, is already included in the definition of basis used in point $(i)$; it is, however, necessary to assume it explicitly in points $(ii)$ and $(iii)$. As for point $(iii)$, a counterexample can be found in $\fRelCategory$, by replacing an orthogonal family with the one obtained by repeating some element $\bra{x}$, and using the fact that $\bra{x} + \bra{x} = \bra{x}$ (since the monoidal operation $\sum$ in $\fRelCategory$ is just the union $\cup$). As for point $(ii)$, one can consider the category of finite-dimensional vector spaces over the field with 2 elements (where $1+1=0$): if $\ket{x}$ is a norm-1 vector in a 1-dimensional space $\SpaceG$, the family $(\ket{x},\ket{x},\ket{x})$ is non-orthogonal, and yet a resolution of the identity as $\ket{x}\bra{x}+\ket{x}\bra{x}+\ket{x}\bra{x} = \ket{x}\bra{x} = \id{\SpaceG}$ (this cannot happen in $\fdHilbCategory$).
\end{proof}

\noindent The formulation in terms of orthogonal partition of the counit is related to the orthogonality of (multiplicative) characters traditionally mentioned in the context of Fourier Transform (e.g. used here in Equation \ref{eqn:PointwiseMultCharacters}), as the following lemma shows. 

\begin{theorem}[Orthogonality of Multiplicative Characters] \label{lemma_OrthogonalityCharacters}
Let $(\,\SpaceG,\hbox{\input{modules/symbols/XdotSym.tex}}\!,\hbox{\input{modules/symbols/ZdotSym.tex}}\!)$ be an internal group in a $\dagger$-SMC $\CategoryC$, and $N$ be the normalisation factor for $\hbox{\input{modules/symbols/XdotSym.tex}}\!$. Assume that the characters are all orthogonal, in the sense that $\braket{\chi}{\chi'} = 0$ for $\chi \neq \chi'$, and that $\braket{1}{1} = N$, where $\bra{1} := \ZcounitSym$ is the trivial character. Then if $\ket{\chi},\ket{\chi'}$ are (not necessarily distinct) multiplicative characters of the internal group, the following notion \textbf{orthogonality of multiplicative characters} holds:
\begin{equation}\label{eqn:orthogonalityMultChars}
	\hbox{\input{modules/pictures/orthogonalityMultChars.tex}}
\end{equation}
Now assume that $\CategoryC$ is distributively $\ComMonCategory$-enriched. If the family $\bra{g}_{g\in G}$ of (adjoints of) group elements is normalisable and forms an orthogonal partition of the counit, then Equation \ref{eqn:orthogonalityMultChars} can be re-written in the following, more familiar form (where we have set $\ket{\chi^{-1}} := \ket{\chi} \circ \hbox{\input{modules/symbols/antipodeSym.tex}}\!$):
\begin{equation}\label{eqn:orthogonalityMultCharsFamiliar}
\frac{1}{N}\sum_{g\in G} \braket{\chi^{-1}}{g}\braket{\chi'}{g} = \delta_{\chi\chi'}
\end{equation} 
\end{theorem}
\begin{proof}
By Theorem \ref{thm_PontryaginDualsSMC}, the comultiplication $\ZcomultSym$ acts as a group on the multiplicative characters, and $\hbox{\input{modules/symbols/antipodeSym.tex}}\!$ as the group inverse. The LHS of Equation \ref{eqn:orthogonalityMultChars} can be re-written as $\braket{\chi^{-1} \cdot \tau}{1}$, and $\cdot$ is the pointwise multiplication: since we assumed that the multiplicative characters are orthogonal, the result follows. In order to obtain Equation \ref{eqn:orthogonalityMultCharsFamiliar} from Equation \ref{eqn:orthogonalityMultChars}, all we have to do is observe that the group elements $\ket{g}_{g:G}$ form an orthogonal partition of the unit $\ZunitSym$ (by taking adjoints), and that they are classical states of $\hbox{\input{modules/symbols/ZdotSym.tex}}\!$.
\end{proof}

\noindent Note that, by Definition \ref{def:Antipode} and Frobenius law for $\hbox{\input{modules/symbols/ZdotSym.tex}}\!$, Equation \ref{eqn:orthogonalityMultChars} can equivalently be written as the following, stating that the multiplicative characters are a matchable family (Definition~\ref{def:matchables}) for $(\XcomultSym,\XcounitSym)$: 
\begin{equation}\label{eqn:orthogonalityMultCharsRed}
	\hbox{\input{modules/pictures/orthogonalityMultCharsRed.tex}}
\end{equation}
Equations \ref{eqn:orthogonalityMultChars} and \ref{eqn:orthogonalityMultCharsRed} provide a summation-free version of the orthogonality of multiplicative characters of Equation \ref{eqn:orthogonalityMultCharsFamiliar} (under appropriate conditions). This leads us to the following summation-free definition of the Fourier Transform, valid for any internal group in any $\dagger$-SMC.

\begin{definition}\label{def_FourierTransform}
	Let $\mathbb{G} = (\,\SpaceG,\hbox{\input{modules/symbols/XdotSym.tex}}\!,\hbox{\input{modules/symbols/ZdotSym.tex}}\!)$ be an internal group in any $\dagger$-SMC. The \textbf{Fourier Transform} is defined to be the following mapping $\FourierTransformSym{\mathbb{G}}$ of states of $\SpaceG$ to co-states of $\SpaceG$:
\begin{equation}\label{eqn:FT}
	\hbox{\input{modules/pictures/FTv3.tex}}
\end{equation} 
The \textbf{inverse Fourier transform} is defined to be the following mapping $\InverseFourierTransformSym{\mathbb{G}}$ of co-states of $\SpaceG$ to states of $\SpaceG$: 
\begin{equation}\label{eqn:InverseFT}
    \hbox{\input{modules/pictures/InverseFTv3.tex}}
 \end{equation} 
\end{definition}

\noindent Under appropriate circumstances, the Fourier Transform of Definition \ref{def_FourierTransform} takes the more familiar form of Equation \ref{eqn:DefTraditionalFT}, as shown by the following lemma and its subsequent application to $\fdHilbCategory$.

\begin{lemma}\label{lemma_FTTraditionalSMC}
Let $\mathbb{G} = (\,\SpaceG,\hbox{\input{modules/symbols/XdotSym.tex}}\!,\hbox{\input{modules/symbols/ZdotSym.tex}}\!)$ be an internal group in a $\dagger$-SMC which is distributively $\ComMonCategory$-enriched. Further assume that the multiplicative characters and the group elements of $\mathbb{G}$ are both finite, normalisable families, which form an orthogonal partition of the counits $\XcounitSym$ and $\ZcounitSym$ respectively. Then the Fourier Transform of Definition \ref{eqn:FT} can be written in the following way:
\begin{equation}\label{eqn:FTSummation}
	\hbox{\input{modules/pictures/FTv2.tex}}
\end{equation} 
Furthermore, the Inverse Fourier Transform of Definition \ref{eqn:InverseFT} can be written in the following way:
\begin{equation}\label{eqn:InverseFTSummation}
    \hbox{\input{modules/pictures/InverseFTv2.tex}}
 \end{equation} 
\end{lemma}
\begin{proof}
For the Fourier Transform, use the fact that the multiplicative characters form an orthogonal partition of the counit $\XcounitSym$, as per Equation \ref{eqn:PartitionCounit}, and that they are classical states of $\XcomultSym$, as per Equation \ref{eqn:MultCharDef}. Similar reasoning is used for the Inverse Fourier Transform.
\end{proof}

\noindent In $\fdHilbCategory$, the conditions of Lemma \ref{lemma_FTTraditionalSMC} hold for abelian inner groups (but not for non-abelian ones, as the multiplicative characters fail to form a basis). The rightmost expression in equation \ref{eqn:FTSummation} can be written as follows, where we have $\ket{\chi^{-1}} = \ket{\chi} \cdot \hbox{\input{modules/symbols/antipodeSym.tex}}\!$ (as in Lemma \ref{lemma_OrthogonalityCharacters}):
\begin{equation*}
\frac{1}{N} \sum_\chi \bra{\chi} \braket{\chi^{-1}}{f}
\end{equation*}
The vector $\ket{f}$ can be expanded by using a resolution of the identity in terms of the group elements, courtesy of Lemma \ref{lemma_BasisResolutionPartition}:
\begin{equation*}
\sum_\chi \bra{\chi} \frac{1}{N} \sum_g \braket{\chi^{-1}}{g} \braket{g}{f}
\end{equation*}
Now we use the isomorphism $\SpaceG \isom \Ltwo{G}$ induced by the point structure, under which $f : \Ltwo{G}$ gets mapped to $\ket{f} := \sum_g \ket{g} f(g)$, to obtain:
\begin{equation*}
\sum_\chi \bra{\chi} \frac{1}{N} \sum_g \braket{\chi^{-1}}{g} f(g)
\end{equation*}
Furthermore, the multiplicative characters of $\mathbb{G}$ are, in $\fdHilbCategory$ and under the isomorphism $\SpaceG \isom \Ltwo{G}$ above, the linear extensions of the multiplicative characters of the $G$, and we can re-write the above as:
\begin{equation*}
 \sum_\chi \bra{\chi} \frac{1}{N} \sum_g \chi^{-1}(g) f(g)
\end{equation*}
Finally, we use the isomorphism $\SpaceG^\star \isom \Ltwo{G^\wedge}$ induced by the group structure,\footnote{Where we have used the fact that $\fdHilbCategory$ can be $\fdHilbCategory$-enriched, and thus that $\Hom{\fdHilbCategory}{\SpaceG}{\complexs}$ can be canonically endowed with the finite-dimensional Hilbert space structure of the space $\SpaceG^\star$ of linear functionals $\SpaceG \rightarrow \complexs$.} under which $\tilde{f}:\Ltwo{G^\wedge}$ gets mapped to $\bra{\tilde{f}} := \sum_\chi \bra{\chi} \tilde{f}(\chi)$, to finally obtain:
\begin{equation*}
\tilde{f}(\chi) = \frac{1}{N} \sum_g \chi^{-1}(g) f(g)
\end{equation*}

\noindent This is exactly the same as Equation \ref{eqn:DefTraditionalFT}, and a similar reasoning shows that in $\fdHilbCategory$ Equation \ref{eqn:InverseFTSummation} coincides with Equation \ref{eqn:DefTraditionalInverseFT}. Therefore Definition \ref{def_FourierTransform} matches the traditional definition in the case of abelian internal groups of $\fdHilbCategory$, but it remains to be seen under which circumstances and in which form its usual properties extend to internal groups in arbitrary $\dagger$-SMCs. Here we will focus on three particularly important results: the Fourier Inversion Theorem, the Convolution Theorem and Pontryagin Duality. In order to clarify their categorical formulation, we summarize the role played by each structure:
\begin{enumerate}
\item[(i)] When states $\complexs \rightarrow \SpaceG$ are identified as functions in $\Ltwo{G}$ via the basis of group elements, the monoid $(\XmultSym,\XunitSym)$ acts as the \emph{convolution} operation on $\Ltwo{G}$:
\begin{equation}
\XmultSym \circ \left( \sum_g  f(g) \,\ket{g} \tensor \sum_{g} f'(g)\, \ket{g} \right) = \sum_g \left( \sum_h f(h)f'(g-h) \right)\, \ket{g} = \sum_g (f\star_G f')(g)\, \ket{g}
\end{equation}

\item[(ii)] When states $\complexs \rightarrow \SpaceG$ are identified with functions in $\Ltwo{G}$ via the basis of group elements, the monoid $(\ZmultSym,\ZunitSym)$ acts as the \emph{pointwise multiplication} operation on $\Ltwo{G}$:
\begin{equation}
\ZmultSym \circ \left( \sum_g  f(g) \,\ket{g} \tensor \sum_{g} f'(g)\, \ket{g} \right) = \sum_g f(g)f'(g)\, \ket{g} 
\end{equation}

\item[(iii)] When co-states $\SpaceG \rightarrow \complexs$ are identified with functions in $\Ltwo{G^\wedge}$ via the co-basis of multiplicative characters, the monoid\footnote{It is a co-monoid in $\fdHilbCategory$, but when acting on co-states it is a monoid.} $(\ZcomultSym,\ZcounitSym)$ acts as the \emph{convolution} operation on $\Ltwo{G^\wedge}$:
\begin{equation}
\ZcomultSym \circ \left( \sum_\chi f(\chi)\, \bra{\chi} \tensor \sum_\chi f'(\chi)\, \bra{\chi} \right) = \sum_\chi \left( \sum_\sigma f(\sigma)f'(\chi-\sigma) \right)\, \bra{\chi} = \sum_\chi (f\star_{G^\wedge} f')(\chi)\, \bra{\chi}
\end{equation}

\item[(iv)] When co-states $\SpaceG \rightarrow \complexs$ are identified with functions in $\Ltwo{G^\wedge}$ via the co-basis of multiplicative characters, the monoid. $(\XcomultSym,\XcounitSym)$ acts as the \emph{pointwise multiplication} operation on $\Ltwo{G^\wedge}$:
\begin{equation}
\XcomultSym \circ \left( \sum_\chi f(\chi)\, \bra{\chi} \tensor \sum_\chi f'(\chi)\, \bra{\chi} \right) = \sum_\chi f(\chi)f'(\chi)\, \bra{\chi}
\end{equation}

\end{enumerate}

\begin{theorem}[Categorical Fourier Inversion Theorem]\label{thm_CategoricalFourierInversion}
Let $\mathbb{G} = (\,\SpaceG,\hbox{\input{modules/symbols/XdotSym.tex}}\!,\hbox{\input{modules/symbols/ZdotSym.tex}}\!)$ be an internal group in any $\dagger$-SMC $\CategoryC$. Then the Fourier Transform and the Inverse Fourier Transform from Definition \ref{def_FourierTransform} are mutually inverse bijections between states and co-states of $\SpaceG$.
\end{theorem}
\begin{proof}
The following shows that $\InverseFourierTransformSym{\mathbb{G}} \circ \FourierTransformSym{\mathbb{G}} = \id{\mathbb{G}}$, by using Definition \ref{def:Antipode} for the antipode and Lemma \ref{lemma_AntipodeInverse}:
\begin{equation}\label{eqn:FTInversionThm}
    \hbox{\input{modules/pictures/FTInversionThm.tex}}
\end{equation} 
A similar proof holds for $\FourierTransformSym{\mathbb{G}} \cdot \InverseFourierTransformSym{\mathbb{G}} = \id{\mathbb{G}}$, by expanding the antipode in terms of its definition and using Frobenius law (once per colour, exactly like in the proof of Lemma  \ref{lemma_AntipodeInverse}).
\end{proof}
  
\begin{theorem}[Categorical Convolution Theorem] \label{thm_categoricalConvolutionTheorem}
Let $\mathbb{G} = (\,\SpaceG,\hbox{\input{modules/symbols/XdotSym.tex}}\!,\hbox{\input{modules/symbols/ZdotSym.tex}}\!)$ be an internal group in any $\dagger$-SMC $\CategoryC$. Then the Fourier Transform is a monoid homomorphism:
\begin{equation}
\FourierTransformSym{\mathbb{G}} : (\,\SpaceG,\XmultSym,\XunitSym) \longrightarrow  (\,\SpaceG,\XcomultSym,\XcounitSym)
\end{equation}
\end{theorem}
\begin{proof}
That $\mathcal{F}_{\mathbb{G}}$ preserves the unit $\hbox{\input{modules/symbols/XdotSym.tex}}\!$ is obvious by unitality of $\XmultSym$. Using associativity, followed by one application of Frobenius Law, we get the desired:
\begin{equation}\label{eqn:ConvolutionThmProof}
    \hbox{\input{modules/pictures/ConvolutionThmProof.tex}}
\end{equation} 
\end{proof}

\noindent In $\fdHilbCategory$, Theorem \ref{thm_categoricalConvolutionTheorem} can be seen to reduce to Equation \ref{eqn:ConvolutionTheorem}. The monoid $(\,\SpaceG,\XmultSym,\XunitSym)$ acts on states $\ket{f} = \sum_g \ket{g} f(g)$ as the following convolution operation:
\begin{equation}
\XmultSym \circ (\ket{f} \tensor \ket{f'}) = \sum_g \sum_{g'} f(g)f'(g') \XmultSym \circ (\ket{g} \tensor \ket{g'}) = \sum_h \left(\sum_{g'} f(h-g')f'(g')\right) \ket{h}
\end{equation}

\begin{theorem}[Categorical Pontryagin Duality]\label{thm_CategoricalPontryaginDuality}
Let $\mathbb{G} = (\,\SpaceG,\hbox{\input{modules/symbols/XdotSym.tex}}\!,\hbox{\input{modules/symbols/ZdotSym.tex}}\!)$ be an internal group in a $\dagger$-SMC $\CategoryC$. Then the Fourier Transform $\FourierTransformSym{\mathbb{G}}$ is a bijection between states of $\mathbb{G}$ and states of $\mathbb{G}^\wedge$, which is furthermore canonical in the sense that:
\begin{equation}
^{(\varphi^\wedge)}M \circ \FourierTransformSym{\mathbb{H}} \circ _{\varphi\!}M = \FourierTransformSym{\mathbb{G}}
\end{equation}
where $\varphi: \mathbb{G} \rightarrow \mathbb{H}$ is any unitary isomorphim of internal groups in $\CategoryC$, for  $\mathbb{H} = (\,\SpaceH,\hbox{\input{modules/symbols/XaltdotSym.tex}}\!,\hbox{\input{modules/symbols/ZaltdotSym.tex}}\!)$ any other internal group of $\CategoryC$. We have defined the following:
\begin{enumerate}
\item[(i)] $\varphi^\wedge := \varphi$ is an isomorphism of internal groups $\mathbb{H}^\wedge \rightarrow \mathbb{G}^\wedge$ in $\CategoryC^{op}$
\item[(ii)] $_\varphi M$ as the map sending state $\ket{f} : \tensorUnit \rightarrow \SpaceG$ to state $\varphi \circ \ket{f}: \tensorUnit \rightarrow \SpaceH$
\item[(iii)] $^{(\varphi^\wedge)} M$ as the map sending co-state $\bra{f} : \SpaceH \rightarrow \tensorUnit$ (a state in $\CategoryC^{op}$) to co-state $\bra{f}\circ \varphi : \SpaceG \rightarrow \tensorUnit$
\end{enumerate}
\end{theorem}
\begin{proof}
The bijection is proven by Theorem \ref{thm_CategoricalFourierInversion}, so all we have to show is canonicity:
\begin{equation}\label{eqn:PontDualityCanonProof}
    \hbox{\input{modules/pictures/PontDualityCanonProof.tex}}
\end{equation} 
The first equality follows from the fact that $\varphi$ is a morphism of internal groups. The second equality follows from the (easy to check) fact that, if $\varphi$ is a unitary isomorphism $\varphi: \mathbb{G} \rightarrow \mathbb{H}$, then $\varphi^\dagger$ is a unitary isomorphism $\varphi^\dagger: \mathbb{H} \rightarrow \mathbb{G}$, and hence we have:
\begin{equation}
\XaltcounitSym \circ \varphi = \left(\varphi^\dagger \circ \XaltunitSym \right)^\dagger = \left(  \XunitSym\right)^\dagger = \XcounitSym
\end{equation}
\end{proof}

\noindent In $\fdHilbCategory$ (with abelian internal groups), Equation \ref{eqn:PontDualityCanonProof} takes the form of Equation \ref{eqn:FTcanonicity}. To conclude, we provide a categorical definition of Fourier matrices, which helps to frame the difference between them and the Fourier Transform.

\begin{definition}
\label{def_CategoricalHadamard}
Let $\mathbb{G} = (\,\SpaceG,\hbox{\input{modules/symbols/XdotSym.tex}}\!,\hbox{\input{modules/symbols/ZdotSym.tex}}\!)$ be an internal group in a $\dagger$-SMC $\CategoryC$. A \textbf{Fourier matrix} is defined to be a co-monoid isomorphism $h:(\;\ZcomultSym,\ZcounitSym) \rightarrow (\;\XcomultSym,\XcounitSym)$ which is furthermore a monoid isomorphism $h:(\;\hbox{\input{modules/symbols/timemultSym.tex}}\!,\hbox{\input{modules/symbols/timeunitSym.tex}}\!) \rightarrow (\;\hbox{\input{modules/symbols/timematchSym.tex}}\!,\hbox{\input{modules/symbols/timematchunitSym.tex}}\!)$
\end{definition}

\noindent Definition \ref{def_CategoricalHadamard} may look cryptic at first, and we clarify it in the following remark.
\begin{remark}\label{rmrk_CategoricalHadamard}
A co-monoid isomorphims $h:(\;\ZcomultSym,\ZcounitSym) \rightarrow (\;\XcomultSym,\XcounitSym)$ is an isomorphism $h: \SpaceG \rightarrow \SpaceG$ which satisfies:
\begin{align}
    \XcomultSym \circ h &= \left( h \tensor h \right) \circ \ZcomultSym \\
    \XcounitSym \circ h &= \ZcounitSym
\end{align}
and in particular it maps $\hbox{\input{modules/symbols/ZdotSym.tex}}\!$-classical states to $\hbox{\input{modules/symbols/XdotSym.tex}}\!$-classical states (since it is an isomorphism, it is a bijection between the classical states of the two comonoids). The requirement that $h$ is furthermore a monoid isomorphism $h:(\;\hbox{\input{modules/symbols/timemultSym.tex}}\!,\hbox{\input{modules/symbols/timeunitSym.tex}}\!) \rightarrow (\;\hbox{\input{modules/symbols/timematchSym.tex}}\!,\hbox{\input{modules/symbols/timematchunitSym.tex}}\!)$ amounts to the requirement that:
\begin{align}
    h \circ \XmultSym &= \ZmultSym \circ \left( h \tensor h \right)\\
    h \circ \XunitSym &= \ZunitSym
\end{align}
and in particular, as a bijection of classical states, $h$ is a group isomorphism from the group given by $(\XmultSym,\XunitSym)$ acting on $\hbox{\input{modules/symbols/ZdotSym.tex}}\!$-classical states to the group  given by $(\ZmultSym,\ZunitSym)$ acting on $\hbox{\input{modules/symbols/XdotSym.tex}}\!$-classical states.\\
\end{remark}

\noindent In $\fdHilbCategory$ (with abelian internal groups), Definition \ref{def_CategoricalHadamard} yields the usual Fourier matrices. Indeed $h$ corresponds to an isomorphism $\Psi: G \rightarrow G^\wedge$ (by considering its action on classical states): the linear isomorphism $h$ is itself the Discrete Fourier Transform of Equation \ref{eqn:FTDefDFT} (where $\SpaceG$ is identified with $\Ltwo{G}$ and $\Ltwo{G^\wedge}$ using $\hbox{\input{modules/symbols/ZdotSym.tex}}\!$ and $\hbox{\input{modules/symbols/XdotSym.tex}}\!$ respectively), and the matrix of $h$ in the basis defined by $\hbox{\input{modules/symbols/ZdotSym.tex}}\!$ is the Fourier matrix of Equation \ref{eqn:HadamardMatrixDef}.

\section{The Fourier transform in the Category $\fRelCategory$}
\label{section_RelFT}

The abstract correspondence between strongly complementary observables and Fourier transforms in Section \ref{section_AbelianGroups_FourierTransform} means that a characterization of strongly complementary observables in any $\dagger$-SMC allows for the definition of a Fourier transform in that category.  In this section, we apply this idea to $\fRelCategory$, the category of finite sets and relations (for an introduction, see Evans et al.\cite{Msc-ClassicalStructuresRel} and references therein). In this setting, the relevant classical structures are classified by abelian groupoids, in a sense made clear by Theorem \ref{thm_classicalStructuresRel} below. Please note that in $\fRelCategory$ the scalars are the booleans $\{\bot,\top\}$, and $\top = \id{\tensorUnit}: \tensorUnit \rightarrow \tensorUnit$ is the only invertible scalar. As a consequence, all $\dagger$-qSFAs are automatically $\dagger$-SFAs (and thus all $\dagger$-qSCFAs are in fact classical structures). 

\begin{definition}
A \textbf{(finite) abelian groupoid} on some finite set $A$ is any finite family $(G_\lambda,+_\lambda,0_\lambda)_{\lambda \in \Lambda}$ of finite abelian groups such that $(G_\lambda)_{\lambda \in \Lambda}$ is a finite partition of $A$ into disjoint subsets. We denote one such groupoid by $\oplus_{\lambda \in \Lambda} G_\lambda$, leaving the groups' structures understood.
\end{definition}

\begin{theorem}\label{thm_classicalStructuresRel}
Let $\hbox{\input{modules/symbols/XdotSym.tex}}\!$ be a classical structure in $\fRelCategory$ on a finite set $X$. Then there exists a (unique) abelian groupoid $\oplus_{\lambda \in \Lambda} G_\lambda$ on $A$ such that, for all $a,b \in A$:
\begin{equation}
    \XmultSym \circ \left( \{a\} \times \{b\} \right) = 
    \begin{cases}
        \{ a +_\lambda b\} \text{ if for some } \lambda \in \Lambda \text{ we have } a,b \in G_\lambda \\
        \emptyset \text{ otherwise}
    \end{cases}
\end{equation}
Furthermore, each abelian groupoid defines a unique $\dagger$-SCFA in this way. The classical states of the co-monoid fragment $(\XcomultSym,\XcounitSym)$ are the family $(G_\lambda)_{\lambda \in \Lambda}$, which is also a matching family for the monoid fragment $(\XmultSym,\XunitSym)$.\footnote{Recall that the states $1 \rightarrow A$ of a finite set $A$ in $\fRelCategory$ are exactly the subsets of $A$.}
\end{theorem}
\begin{proof}
Proven by Pavlovic in~\cite{Msc-ClassicalStructuresRel}, and extended to the case of non-commutative $\dagger$-SFAs / non-abelian groupoids by Heunen et al.~\cite{Msc-RelativeFrob}.
\end{proof}

\noindent The groupoids corresponding to complementary / strongly complementary classical structures ($\dagger$-SCFA's)\footnote{See Appendix~\ref{app:basic} for more on classical structures.} take a particularly nice form.

\begin{theorem}\label{thm_StrongComplementarityRel}
Let $\oplus_{\gamma \in \Gamma} H_\gamma$ and $\oplus_{\lambda \in \Lambda} G_\lambda$ be abelian groupoids on some finite set $A$, and let $\hbox{\input{modules/symbols/ZdotSym.tex}}\!$ and $\hbox{\input{modules/symbols/XdotSym.tex}}\!$ be the corresponding classical structures on $A$ in $\fRelCategory$. Then $\hbox{\input{modules/symbols/ZdotSym.tex}}\!$ and $\hbox{\input{modules/symbols/XdotSym.tex}}\!$ are strongly complementary if and only if: 
\begin{enumerate}
\item[(i)] there is a finite abelian group $H$ such that for all $\gamma \in \Gamma$ we have $H_\gamma \isom H$ as groups.
\item[(ii)] there is a finite abelian group $G$ such that for all $\lambda \in \Lambda$ we have $G_\lambda \isom G$ as groups.
\item[(iii)] for each $(\lambda,\gamma) \in \Lambda \times \Gamma$, the intersection $G_\lambda \cap H_\gamma$ is a singleton. 
\end{enumerate}
\end{theorem}
\begin{proof}
Proven by Evans et al.~\cite{evans2009classifying}.
\end{proof}

\noindent In particular, this means that $\Lambda \isom H$ and $\Gamma \isom G$ as sets: as a consequence, we will write abelian groupoids corresponding to strongly complementary classical structures as $\oplus^{|G|}H$ and $\oplus^{|H|}G$, leaving the indexing of the partitions as understood. We will also implicitly label the elements of the underlying set $A$ as:
\begin{equation}
A \isom \suchthat{(h,g)}{h \in H \text{ and } g \in G}
\end{equation}

\noindent In $\fdHilbCategory$, strongly complementary pairs of classical structures have the same number of classical states, and their monoid fragments act on each other's classical states as isomorphic groups $G$ and $G^\wedge$. As a consequence, it is possible to fix isomorphisms between the two groups and construct Hadamard transforms as per Definition \ref{def_CategoricalHadamard}. In $\fRelCategory$, on the other hand, Hadamard transforms only exist in very special cases.

\begin{theorem} Let $\mathbb{G} = (\,\SpaceG,\hbox{\input{modules/symbols/XdotSym.tex}}\!,\hbox{\input{modules/symbols/ZdotSym.tex}}\!)$ be an abelian internal group in $\fRelCategory$, and let $Z = \oplus^{|G|}H$ and $X = \oplus^{|H|}G$ be the groupoids corresponding to the $\hbox{\input{modules/symbols/ZdotSym.tex}}\!$ and $\hbox{\input{modules/symbols/XdotSym.tex}}\!$ classical structures respectively. Then $(\XmultSym,\XunitSym)$ acts on the $\hbox{\input{modules/symbols/ZdotSym.tex}}\!$-classical states $(H_g)_{g \in G}$ as the finite abelian group $G$, and $(\ZmultSym,\ZunitSym)$ acts on the $\hbox{\input{modules/symbols/XdotSym.tex}}\!$-classical states $(G_h)_{h \in H}$ as the finite abelian group $H$.
\end{theorem}
\begin{proof}
It is sufficient to prove this for $(\XmultSym,\XunitSym)$ acting on the $\hbox{\input{modules/symbols/ZdotSym.tex}}\!$-classical states. Indeed we have that $\XunitSym = H_0$ is a $\hbox{\input{modules/symbols/ZdotSym.tex}}\!$-classical state, and that:
\begin{equation}
\XmultSym \circ \left( \ket{H_g} \times \ket{H_{g'}} \right) = \bigcup_{h,h' \in H} \XmultSym \circ \left( \{(h,g)\} \times \{(h',g')\} \right) = \bigcup_{h \in H} \{(h,g+g')\} = \ket{H_{g+g'}}
\end{equation}
\end{proof}

\begin{example}
The groupoids $Z = \mathbb{Z}_2 \oplus \mathbb{Z}_2 \oplus \mathbb{Z}_2$ and $X = \mathbb{Z}_3 \oplus \mathbb{Z}_3$ correspond to strongly complementary structures, call them $\hbox{\input{modules/symbols/ZdotSym.tex}}\!$ and $\hbox{\input{modules/symbols/XdotSym.tex}}\!$ respectively, on a 6-element set $A$. We can label the elements of $A$ as:
\begin{equation}
A \isom \suchthat{(h,g)}{h \in \integersMod{2} \text{ and } g \in \integersMod{3}}
\end{equation}
The classical states of $\hbox{\input{modules/symbols/ZdotSym.tex}}\!$ are the 3 subsets $(\integersMod{2},g)$ for $g \in \integersMod{3}$, while the classical states of $\hbox{\input{modules/symbols/XdotSym.tex}}\!$ are the 2 subsets $(h,\integersMod{3})$ for $h \in \integersMod{2}$. The monoid $(\XmultSym,\XunitSym)$ acts on the $\hbox{\input{modules/symbols/ZdotSym.tex}}\!$-classical states as the group $\integersMod{3}$, while the monoid $(\ZmultSym,\ZunitSym)$ acts on the $\hbox{\input{modules/symbols/XdotSym.tex}}\!$-classical states as the group $\integersMod{2}$.
\end{example}

\begin{corollary}
Let $\mathbb{G} = (\,A,\hbox{\input{modules/symbols/XdotSym.tex}}\!,\hbox{\input{modules/symbols/ZdotSym.tex}}\!)$ be an abelian internal group in $\fRelCategory$. Hadamard transforms for $\mathbb{G}$ exist if and only if the two groups $G$ and $H$ are isomorphic.
\end{corollary}
\begin{proof}
By Remark \ref{rmrk_CategoricalHadamard}, a Hadamard transform gives a group isomorphism between the groups given by the action of the two monoid fragments on each other's classical states: in this case, the existence of a Hadamard transform forces $G$ and $H$ to be isomorphic. On the other hand, if $\Psi: G \rightarrow H$ is a group isomorphism, one could define a map $M_\Psi: A \rightarrow A$ as follows:
\begin{equation}
M_\Psi := \bigcup_{g \in G} \ket{G_{\Psi(g)}}\bra{H_g}  
\end{equation}
This satisfies the monoid and comonoid homomorphism requirements from Definition \ref{def_CategoricalHadamard}, but is not an isomorphism in $\fRelCategory$. To get an isomorphism, and prove the existence of a relevant Hadamard transform in $\fRelCategory$, we consider a new map $t: A \rightarrow A$, which we give in the form of a relation $t \subseteq A \times A$:
\begin{equation}
t := \suchthat{\left((h,g),(\Psi g,\Psi^{-1}h)\right)}{h \in H \text{ and } g \in G}
\end{equation}
\end{proof}

\noindent The Fourier Transform as given by Definition \ref{def_FourierTransform} is valid in any $\dagger$-SMC with a pair of strongly complementary classical structures, and in particular it holds in $\fRelCategory$. The more traditional formulation given by Lemma \ref{lemma_FTTraditionalSMC}, which allows one to see the Fourier Transform as a canonical isomorphism $\Ltwo{G} \isom \Ltwo{G^\wedge}$, is based on the assumption that the group elements of an abelian internal group $\mathbb{G} = (\,\SpaceG,\hbox{\input{modules/symbols/XdotSym.tex}}\!,\hbox{\input{modules/symbols/ZdotSym.tex}}\!)$ in $\fdHilbCategory$ form a basis, giving $\Hom{\fdHilbCategory}{\complexs}{\SpaceG}$ the Hilbert space structure of $\Ltwo{G}$ in a natural way, and that the multiplicative characters form a co-basis, giving $\Hom{\fdHilbCategory}{\SpaceG}{\complexs}$ the Hilbert space structure of $\Ltwo{G^\wedge}$ in a natural way; more details about this can be found in Section \ref{section_EnrichedPontryaginDuality} below.

\noindent In $\fRelCategory$, however, the assumptions of Lemma \ref{lemma_FTTraditionalSMC} only hold in one, trivial case. The classical states of a classical structure in $\fRelCategory$ are always a finite, orthogonal and normalisable family, and $\fRelCategory$ is distributively $\ComMonCategory$-enriched as required. However, as Lemma \ref{thm_partitionIdentityRel} below notes, there is a unique classical structure on any finite set $A$ with classical states forming the resolution of the identity, and the only abelian internal group in $\fRelCategory$ satisfying the assumptions of Lemma \ref{lemma_FTTraditionalSMC} is the one on the tensor unit $\{\star\}$ of $\fRelCategory$.

\begin{lemma}
\label{thm_partitionIdentityRel}
Let $\hbox{\input{modules/symbols/ZdotSym.tex}}\!$ be a classical structure on a finite set $A$ in $\fRelCategory$, and let $Z$ be the associated abelian groupoid. The classical states of $\hbox{\input{modules/symbols/ZdotSym.tex}}\!$ form an orthogonal resolution of the identity if and only if the abelian groupoid is discrete, i.e.  $Z=\bigoplus^{A}\mathbb{Z}_1$. Furthermore, if $(\hbox{\input{modules/symbols/ZdotSym.tex}}\!,\hbox{\input{modules/symbols/XdotSym.tex}}\!)$ is a strongly complementary pair on $A$, with $\hbox{\input{modules/symbols/ZdotSym.tex}}\!$ associated to an abelian groupoid $Z=\bigoplus^{A}\mathbb{Z}_1$, then the abelian groupoid associated with $\hbox{\input{modules/symbols/XdotSym.tex}}\!$ is in fact a group, in the form $X = G$ for some finite abelian group $G = (A,+,0)$ on the element set $A$. As a consequence, the only abelian internal group in $\fRelCategory$ satisfying the assumptions of Lemma \ref{lemma_FTTraditionalSMC} is the (unique) abelian internal group on the tensor unit $\{\star\}$.
\end{lemma}
\begin{proof} 
In $\RelCategory$ the scalars are $0$ or $1$ and summation is given by set union.  Thus a resolution of the identity must satisfy the following equation, where each $\chi$ is a classical state:
  \begin{equation}\label{eqn:MultCharResolutionIdRel}
    \hbox{\input{modules/pictures/MultCharResolutionID_Rel.tex}}
  \end{equation}
In the specific case of $Z=\bigoplus^{A}\mathbb{Z}_1$, we have that classical states are in the form $\chi = \{(\star,a)\}$ and~\eqref{eqn:MultCharResolutionIdRel} reads
\begin{align}
\bigcup_{a \in A} \{(\star, a)\}\circ\{(a,\star)\} 
= \bigcup_{a \in A} \{(a, a)\} = \id{A}
\end{align}

\noindent When $Z$ is of the generic form $Z = \bigoplus_{\lambda \in \Lambda} G_\lambda$, the classical states are in the form $\chi = G_\lambda$, and~\eqref{eqn:MultCharResolutionIdRel} reads
\begin{align}
\bigcup_{\lambda \in \Lambda} \{(\star, g')|g': G_\lambda\}\circ\{(g,\star)|g: G_\lambda\} =  \bigcup_{\lambda \in \Lambda} \{(g, g')|g,g': G_\lambda\}
\end{align}
which cannot be the identity if $|G_h|>1$. Therefore the unique classical structures $\hbox{\input{modules/symbols/ZdotSym.tex}}\!$ with classical states forming an orthogonal resolution of the identity are the ones associated to discrete abelian groupoids, in the form $Z=\bigoplus^{A}\mathbb{Z}_1$. If $\hbox{\input{modules/symbols/ZdotSym.tex}}\!$ is one such classical structure, and $\hbox{\input{modules/symbols/XdotSym.tex}}\!$ is strongly complementary to it, then it follows immediately from Theorem \ref{thm_StrongComplementarityRel}, and subsequent remarks, that the abelian groupoid $X$ associated to $\hbox{\input{modules/symbols/XdotSym.tex}}\!$ must be in fact an abelian group, in the form $X = \oplus^{\integersMod{1}} G$, where $G$ is some finite abelian group with element set $A$. But to satisfy the assumptions of Lemma \ref{lemma_FTTraditionalSMC}, we must have that $X$ is also a discrete groupoid, which forces $G \isom \integersMod{1}$ and $A \isom \{\star\}$.
\end{proof}

\noindent So there is no direct parallel in $\fRelCategory$ of the $\fdHilbCategory$ view that the Fourier Transform is a canonical isomorphism $\Ltwo{G} \isom \Ltwo{G^\wedge}$, or of its traditional formulation from Equation \ref{eqn:DefTraditionalFT}. At first sight, this seems to be because the classical states of the two structures of a generic abelian internal group need not form a basis. The question then naturally arises whether restricting our attention to some subclass of states (e.g. those which are linear combinations of the classical states) would lead to a canonical isomorphism similar to the $\fdHilbCategory$ one. Theorem \ref{thm_NoCanonicalIsomRel} below answers this question negatively.

Define the \textbf{span} $\langle s_j \rangle_{j \in J}$ of a family of states $s_j \subseteq A$ to be the set of all states $r \subseteq A$ which can be obtained as the union $r = \cup_{j \in I} s_j$ of some subfamily $(s_j)_{j \in I\subseteq J}$. If the $(s_j)_{j \in J}$ are pairwise disjoint, then the states $r \in \langle s_j \rangle_{j \in J}$ are exactly those that descend to boolean functions over the set $\suchthat{s_j}{j \in J}$:
\begin{align}\label{eqn:spanBooleanFunctions}
    r \in \langle s_j \rangle_{j \in J} \;\; \mapsto \;\; \left(j \mapsto \braket{s_j}{r}\right) \in \mathbb{B}[J]
\end{align}
where we denoted by $\ket{s} : \tensorUnit \rightarrow A$ the state corresponding to subset $s \subseteq A$. Now consider an abelian internal group $\mathbb{G} = (A,\hbox{\input{modules/symbols/ZdotSym.tex}}\!,\hbox{\input{modules/symbols/XdotSym.tex}}\!)$ in $\fRelCategory$, and let $Z = \bigoplus^{|G|}H$ and $X = \bigoplus^{|H|}G$ be the abelian groupoids corresponding to $\hbox{\input{modules/symbols/ZdotSym.tex}}\!$ and $\hbox{\input{modules/symbols/XdotSym.tex}}\!$ respectively. Then $\langle H_g \rangle_{g:G}$, seen as the set $\{0,1\}^G$ of boolean functions on $\suchthat{H_g}{g:G}$, plays the role in $\RelCategory$ that $\Ltwo{G}$ played in $\fdHilbCategory$, and similarly $\langle G_h \rangle_{h:H}$, seen as $\{0,1\}^H$, is the analogue of $\Ltwo{H}$ (and takes the place $\Ltwo{G^\wedge}$ had in $\fdHilbCategory$). 

\begin{theorem}\label{thm_NoCanonicalIsomRel}
Let $\mathbb{G} = (A,\hbox{\input{modules/symbols/ZdotSym.tex}}\!,\hbox{\input{modules/symbols/XdotSym.tex}}\!)$ be an abelian internal group in $\fRelCategory$, and let $Z = \bigoplus^{|G|}H$ and $X = \bigoplus^{|H|}G$ be the abelian groupoids corresponding to $\hbox{\input{modules/symbols/ZdotSym.tex}}\!$ and $\hbox{\input{modules/symbols/XdotSym.tex}}\!$ respectively. Then $\langle H_g \rangle_{g \in G} \cap \langle G_h \rangle_{h \in H} = \{\emptyset, A\}$. In particular, under the correspondence of~\eqref{eqn:spanBooleanFunctions},  the Fourier Transform of $\fRelCategory$ does not restrict to an isomorphism $\{\bot,\top\}^G \isom \{\bot,\top\}^H$, nor can be restricted to an isomorphism $S_G \isom S_H$ for any $S_G \subseteq \{\bot,\top\}^G$ and $S_H \subseteq\{\bot,\top\}^H$ containing non-constant functions. 
\end{theorem}
\begin{proof}
Let $r \in \langle H_g \rangle_{g\in G} \cap \langle G_h \rangle_{h \in H}$, then $r$ is either the empty set or it contains some $(h',g') \in A$. But if $(h',g') \in r$, then for all $g \in G$ we have $(h',g) \in r$ (because $r \in \langle G_h \rangle_{h:H}$), and then for any $g \in G$ we have that for all $h \in H$ $(h,g) \in r$ (because  $r \in \langle H_g \rangle_{g \in G}$). Thus for all $g \in G$ and $h \in H$ we have $(h,g) \in r$, i.e. $r = A$. The state $r = \emptyset$ in $\langle H_g \rangle_{g\in G} \cap \langle G_h \rangle_{h \in H}$ corresponds the constant $\bot$ function in $\{\bot,\top\}^H$ and in $\{\bot,\top\}^G$ (under Equation \ref{eqn:spanBooleanFunctions}), while the state $r = A$ corresponds to the constant $\top$ function. The Fourier Transform maps the constant $\bot$ function of $\{\bot,\top\}^G$ to the constant $\bot$ function of $\{\bot,\top\}^H$, and similarly with the constant $\top$ functions. However, if $r \in \langle H_g \rangle_{g\in G}$ is neither empty nor the whole of $A$, i.e. if it corresponds to a non-constant function in $\{\bot,\top\}^G$, then its Fourier Transform does not lie in the span $\langle G_h \rangle_{h \in H}$, and thus cannot be seen as a function in $\{\bot,\top\}^H$.
\end{proof}

\noindent As a consequence of Theorem \ref{thm_NoCanonicalIsomRel}, the best analogue in $\fRelCategory$ of the statement that the Fourier Transform is a canonical isomorphism $\Ltwo{G} \isom \Ltwo{G^\wedge}$ in $\fdHilbCategory$ is the trivial statement that the Fourier Transform in $\fRelCategory$ is an isomorphism $\{\bot_G,\top_G\} \isom \{\bot_H,\top_H\}$ between the constant functions of $\{\bot,\top\}^G$ and of $\{\bot,\top\}^H$.\\

\noindent To summarise, in $\fdHilbCategory$ we have the following views of the Fourier transform for abelian internal groups:
\begin{enumerate}
\item[1.] As quantum Fourier Transform, implemented by application of a Hadamard transform (subject to a non-canonical choice of isomorphism $G \isom G^\wedge$) followed by operations in the computational basis.
\item[2.] In the sense of Pontryagin duality, as a canonical isomorphism $\Ltwo{G} \isom \Ltwo{G^\wedge}$.
\item[3.] Again as quantum Fourier Transform, but implemented without Hadamard transform by operating in the strongly complementary basis of multiplicative characters.
\end{enumerate}

\noindent In $\RelCategory$, on the other hand, things are very different:
\begin{enumerate}
\item[1.] Except in the case where $Z = \bigoplus^{|G|}G$ and $X = \bigoplus^{|G|}G$ (isomorphic, but different), no Hadamard transform can exist in $\RelCategory$.
\item[2.] The Fourier transform does not give, in $\RelCategory$, an isomorphism $\{0,1\}^G \isom \{0,1\}^H$ between the spaces of boolean-valued functions on the group elements / multiplicative characters.
\item[3.] The operational definition based on strong complementarity, however, is still valid in $\fRelCategory$. 
\end{enumerate}




\noindent It is part of the process theoretic programme that the operational features of quantum theory should be modelled categorically, and that any category sharing features with $\fdHilbCategory$ should be considered, at least in principle, as a potential (toy?) model of quantum mechanics. The category $\fRelCategory$ is an example of one such model. We have shown that Hadamard matrices do not generalise well outside $\fdHilbCategory$, and certainly fail to implement a Quantum Fourier Transform in $\RelCategory$, but that our treatment of Fourier theory based on strong complementarity goes through unharmed. As a consequence, categorical quantum algorithms where the Quantum Fourier Transform is formulated using strong complementarity will straightforwardly generalise to $\RelCategory$ and other categories in the CQM programme. We take this to be an indication that our perspective on the Quantum Fourier Transform is conceptually sound and operationally advantageous.

\section{Enriched Pontryagin Duality}
\label{section_EnrichedPontryaginDuality}

\noindent While our Definition \ref{def_FourierTransform} of the Fourier Transform holds for any internal group in any $\dagger$-SMC, we have seen that the it need not correspond to the traditional formulation from Equation \ref{eqn:DefTraditionalFT}, even in the case of distributively $\ComMonCategory$-enriched categories where the summation is well defined. We have seen that one of the reasons this may fail is that the classical states of $\dagger$-qSFAs (and even of $\dagger$-qSCFAs, as $\fRelCategory$ shows) need not form a basis, and hence need not give the resolution of the identity / partition of the counit required to obtain the formulation in Equation \ref{eqn:DefTraditionalFT} via Lemma \ref{lemma_FTTraditionalSMC}.

Now consider an abelian internal group  $\mathbb{G} = (\,\SpaceG,\hbox{\input{modules/symbols/ZdotSym.tex}}\!,\hbox{\input{modules/symbols/XdotSym.tex}}\!)$ in $\fdHilbCategory$, with $\hbox{\input{modules/symbols/XdotSym.tex}}\!$ acting on the $\hbox{\input{modules/symbols/ZdotSym.tex}}\!$-classical states as the finite abelian group $G$: the fact that in $\fdHilbCategory$ the classical states of $\dagger$-qSCFAs form orthogonal bases is in turn related to the identification of $\Hom{\fdHilbCategory}{\complexs}{\SpaceG}$ with $\Ltwo{G}$ and of $\Hom{\fdHilbCategory}{\SpaceG}{\complexs}$ with $\Ltwo{G^\wedge}$. The categorical Pontryagin duality of Theorem \ref{thm_CategoricalPontryaginDuality} states that the Fourier Transform is a bijection between a set of states and a set of co-states: in order to finally turn this into a statement about $\LtwoSym$-spaces (or an appropriate generalisation), we now look a little further into enrichment.\\

\newcommand{\HomsetFunctor}[1]{\Hom{#1}{\emptyArg}{\emptyArg}}
\newcommand{\EnrichedHomsetFunctor}[2]{\Hom{#1}{\emptyArg}{\emptyArg}^{(#2)}}

\noindent If $\CategoryC$ is any category, then the statement that the morphisms between any two objects $\SpaceH,\SpaceH'$ form a set can be made categorically significant by considering the Hom functor:
\begin{equation}\label{eqn:Homfunctor}
    \hbox{\input{modules/pictures/Homfunctor.tex}}
\end{equation}

If $\CategoryC$ is a $\dagger$-SMC the Hom functor comes with natural transformations relating it to the additional structure:
\begin{enumerate}

\item[1.] There is a functor $\dagger: \CategoryC \rightarrow \CategoryC^{op}$ acting as the identity on objects and as the dagger on morphisms. The relation between the dagger and the Hom functor is encoded by a natural transformation $\eta^{[\dagger]}: \HomsetFunctor{\CategoryC} \rightarrow \HomsetFunctor{\CategoryC} \cdot swap \circ (\dagger \times \dagger)$, where $swap$ is the functor $\CategoryC \times \CategoryC^{op} \rightarrow \CategoryC^{op} \times \CategoryC$ swapping the pair: 
\begin{equation}
\eta^{[\dagger]}_{(\SpaceH,\SpaceH')} = f \mapsto f^\dagger
\end{equation}

\item[2.] For each morphism $h : \SpaceG \rightarrow \SpaceG'$ in $\CategoryC$, there is a functor $_h F: \CategoryC^{op} \times \CategoryC \rightarrow \CategoryC^{op} \times \CategoryC$ defined as:
\begin{align}
(\SpaceH,\SpaceH') &\mapsto (\SpaceG \tensor \SpaceH, \SpaceG' \tensor \SpaceH') \text{ on objects}\\
(f,f') & \mapsto (\id{\SpaceG} \tensor f, \id{\SpaceG'} \tensor f') \text{ on morphisms}
\end{align}
Similarly, there is a functor $F_h: \CategoryC^{op} \times \CategoryC \rightarrow \CategoryC^{op} \times \CategoryC$ that applies the tensor product on the right. The relation between $\tensor$ and the Hom functor is encoded by a natural transformation $\eta^{[h \tensor \emptyArg]}: \HomsetFunctor{\CategoryC} \rightarrow \HomsetFunctor{\CategoryC} \cdot \; _h F$  given by:
\begin{equation}
\eta^{[h \tensor \emptyArg]}_{(\SpaceH,\SpaceH')} = f \mapsto h \tensor f
\end{equation}
Similarly, there is a natural transformation $\eta^{[\emptyArg \tensor h]}: \HomsetFunctor{\CategoryC} \rightarrow \HomsetFunctor{\CategoryC} \circ F_h$.

\item[3.] There are a number of additional conditions, e.g. the consistency condition on composition and tensoring $\eta^{\left[(a \circ b)\tensor \emptyArg\right]} = (\eta^{[a \tensor \emptyArg]} \; _b F) \circ \eta^{[b \tensor \emptyArg]}$.

\end{enumerate}

\noindent If the category $\CategoryC$ is enriched over some other category $\CategoryD$, then the Hom functor factorises as $\HomsetFunctor{\CategoryC} = U \circ \EnrichedHomsetFunctor{\CategoryC}{\CategoryD}$, where $\EnrichedHomsetFunctor{\CategoryC}{\CategoryD}: \CategoryC^{op}\times \CategoryC \rightarrow \CategoryD$ is the \textbf{enriched Hom functor} and $U : \CategoryD \rightarrow \SetCategory$ if a faithful functor making $\CategoryD$ into a concrete category. If $\CategoryC$ is a $\dagger$-SMC, on also requires the existence of appropriate natural transformations relating the enriched Hom functor to the dagger and tensor structures.\\

\noindent Our current statement of categorical Pontryagin duality, as formalised in Theorem \ref{thm_CategoricalPontryaginDuality}, goes as follows: the Fourier Transform for an internal group $\mathbb{G} = (\,\SpaceG,\hbox{\input{modules/symbols/XdotSym.tex}}\!,\hbox{\input{modules/symbols/ZdotSym.tex}}\!)$ in a $\dagger$-SMC $\CategoryC$ is a bijection between the set of states $\Hom{\CategoryC}{\tensorUnit}{\SpaceG}$ and the set of co-states $\Hom{\CategoryC}{\SpaceG}{\tensorUnit}$, which is furthermore canonical in the sense that: 
\begin{equation}
^{\tiny{(\varphi^\wedge)}}M \circ \FourierTransformSym{\mathbb{H}} \circ _{\tiny{\varphi}\!}M = \FourierTransformSym{\mathbb{G}}
\end{equation}
where $\varphi: \mathbb{G} \rightarrow \mathbb{H}$ is any isomorphims of internal groups in $\CategoryC$, for  $\mathbb{H} = (\,\SpaceH,\hbox{\input{modules/symbols/XaltdotSym.tex}}\!,\hbox{\input{modules/symbols/ZaltdotSym.tex}}\!)$ any other internal group of $\CategoryC$. The morphisms of sets $_{\varphi\!}M$ and $^{(\varphi^\wedge)}M$ can now be recognised to be $\Hom{\CategoryC}{\id{\tensorUnit}}{\varphi}$ and $\Hom{\CategoryC}{\varphi}{\id{\tensorUnit}}$ respectively. We are now ready to present the enriched version of Theorem \ref{thm_CategoricalPontryaginDuality}.

\begin{theorem}[Enriched Pontryagin Duality]\label{thm_EnrichedPontryaginDuality}
Let $\mathbb{G} = (\,\SpaceG,\hbox{\input{modules/symbols/XdotSym.tex}}\!,\hbox{\input{modules/symbols/ZdotSym.tex}}\!)$ be an internal group in a $\dagger$-SMC $\CategoryC$ which is enriched over some category $\CategoryD$, made concrete by the faithful functor $U: \CategoryD \rightarrow \SetCategory$. Then the Fourier Transform $\FourierTransformSym{\mathbb{G}}$ is an isomorphism in $\CategoryD$ between states of $\mathbb{G}$ (the object $\Hom{\CategoryC}{\tensorUnit}{\SpaceG}^{(\CategoryD)}$) and states of $\mathbb{G}^\wedge$ (the object $\Hom{\CategoryC}{\SpaceG}{\tensorUnit}^{(\CategoryD)}$), which is furthermore canonical in the sense that:
\begin{equation}
\Hom{\CategoryC}{\varphi}{\id{\tensorUnit}}^{(\CategoryD)} \circ \FourierTransformSym{\mathbb{H}} \circ \Hom{\CategoryC}{\id{\tensorUnit}}{\varphi}^{(\CategoryD)} = \FourierTransformSym{\mathbb{G}}
\end{equation}
where $\varphi: \mathbb{G} \rightarrow \mathbb{H}$ is any isomorphims of internal groups in $\CategoryC$,  $\mathbb{H} = (\,\SpaceH,\hbox{\input{modules/symbols/XaltdotSym.tex}}\!,\hbox{\input{modules/symbols/ZaltdotSym.tex}}\!)$ any other internal group of $\CategoryC$.
\end{theorem}
\begin{proof}
\newcommand{\EnrichedFourierTransform}{\mathcal{F}_{\mathbb{G}}}
\newcommand{\EnrichedInverseFourierTransform}{\mathcal{F}^{-1}_{\mathbb{G}}}
The morphisms $\EnrichedFourierTransform: \Hom{\CategoryC}{\tensorUnit}{\SpaceG}^{(\CategoryD)} \rightarrow \Hom{\CategoryC}{\SpaceG}{\tensorUnit}^{(\CategoryD)}$ and $\EnrichedInverseFourierTransform: \Hom{\CategoryC}{\SpaceG}{\tensorUnit}^{(\CategoryD)} \rightarrow \Hom{\CategoryC}{\tensorUnit}{\SpaceG}^{(\CategoryD)}$ in $\CategoryD$, corresponding to the Fourier Transform and the Inverse Fourier Transform, are given by Diagrams \ref{eqn:FTEnrichedDuality} and \ref{eqn:FTinvEnrichedDuality} below, where we defined:
\begin{align}
    FT := \XcounitSym \circ \XmultSym \circ \left( \id{\SpaceG} \tensor \hbox{\input{modules/symbols/antipodeSym.tex}}\! \right) &: \SpaceG \tensor \SpaceG \rightarrow \tensorUnit\\
    FT^{-1} := \ZcomultSym \circ \ZunitSym &: \tensorUnit \rightarrow \SpaceG \tensor \SpaceG\\
    \rho_\SpaceG \text{ to be the right unitor } &: \SpaceG \tensor \tensorUnit \rightarrow \SpaceG 
\end{align}

\begin{equation}\label{eqn:FTEnrichedDuality}
    \hbox{\input{modules/pictures/FTEnrichedDuality.tex}}
\end{equation}

\begin{equation}\label{eqn:FTinvEnrichedDuality}
    \hbox{\input{modules/pictures/FTinvEnrichedDuality.tex}}
\end{equation}

\noindent The Categorical Fourier Inversion Theorem \ref{thm_CategoricalFourierInversion} implies that the composites $\EnrichedFourierTransform \circ \EnrichedInverseFourierTransform$ and $\EnrichedInverseFourierTransform \circ \EnrichedFourierTransform$ get mapped by $U$ to identities in $\SetCategory$: but $U$ is faithful, and thus the composites themselves must be identities in $\CategoryD$, proving $\EnrichedFourierTransform$ and $\EnrichedInverseFourierTransform$ to be isomorphisms. Similarly, canonicity can be proven through equalities in $\SetCategory$ as done in Theorem \ref{thm_CategoricalFourierInversion}, and the equalities can be lifted by faithfulness of $U$. 
\end{proof}

\noindent Now consider a $\dagger$-SMC $\CategoryC$ which is distributively $\ComMonCategory$-enriched, i.e. enriched over the category of finite commutative monoids. The scalars of $\CategoryC$ form a semiring $(R,+,0,\tensor,1)$, where:
\begin{enumerate}
\item[a.] $(R,+,0)$ is the commutative monoid structure on $\Hom{\CategoryC}{\tensorUnit}{\tensorUnit}$ coming from the $\ComMonCategory$-enrichment.
\item[b.]  $(R,\tensor,1)$ is the monoid structure on $\Hom{\CategoryC}{\tensorUnit}{\tensorUnit}$ coming from the symmetric monoidal structure (with $1 = \id{\tensorUnit}$).
\end{enumerate}

\noindent The category $\CategoryC$ can then be enriched over the category of $R$-modules by considering scalar multiplication, in addition to the already existing commutative monoid structure on Homsets given by the $\ComMonCategory$-enrichment. Scalar multiplication is defined as sending scalar $a \in R$ and morphism $f \in \Hom{\CategoryC}{\SpaceG}{\SpaceH}$ to morphism $af  \in \Hom{\CategoryC}{\SpaceG}{\SpaceH}$ defined as follows, where $\lambda_\SpaceG,\Lambda_\SpaceH$ are the left unitors:
\begin{equation}
af := \lambda_\SpaceH \circ \left( a \tensor f \right) \circ \lambda_\SpaceG^{-1} 
\end{equation}

\begin{theorem}[Fourier Transform and $R$-modules]\label{thm_FTRmodules}
Let $\mathbb{G} = (\,\SpaceG,\hbox{\input{modules/symbols/XdotSym.tex}}\!,\hbox{\input{modules/symbols/ZdotSym.tex}}\!)$ be an internal group in a distributively $\ComMonCategory$-enriched $\dagger$-SMC $\CategoryC$. Assume that the families of group elements and multiplicative characters, call them $\ket{g}_{g \in G_{\hbox{\input{modules/symbols/ZdotSym.tex}}\!}}$ and $\ket{\chi}_{\chi \in G_{\hbox{\input{modules/symbols/XdotSym.tex}}\!}}$ respectively, are both finite, normalisable and form orthogonal resolutions of the identity. Let $(G_{\hbox{\input{modules/symbols/ZdotSym.tex}}\!},\XmultSym,\XunitSym)$ and $(G_{\hbox{\input{modules/symbols/XdotSym.tex}}\!},\ZcomultSym,\ZcounitSym)$ be the finite (necessarily abelian) groups given, respectively, by $(\XmultSym,\XunitSym)$ acting on the group elements and $(\ZcomultSym,\ZcounitSym)$ acting on the multiplicative characters. Then:
\begin{enumerate}
\item[(i)] The $\hbox{\input{modules/symbols/ZdotSym.tex}}\!$ structure (i.e. the basis of group elements) induces a canonical isomorphism of $R$-modules between $\Hom{\CategoryC}{\tensorUnit}{\SpaceG}$ and the free $R$-module $R[G_{\hbox{\input{modules/symbols/ZdotSym.tex}}\!}]$.
\item[(ii)] The $\hbox{\input{modules/symbols/XdotSym.tex}}\!$ structure (i.e. the basis of multiplicative characters) induces a canonical isomorphism of $R$-modules between $\Hom{\CategoryC}{\SpaceG}{\tensorUnit}$ and the free $R$-module $R[G_{\hbox{\input{modules/symbols/XdotSym.tex}}\!}]$.
\end{enumerate}
Under these premises, the Fourier Transform can be seen as the following canonical isomorphism:
\begin{equation}
\mathcal{F}_{\mathbb{G}}: R[G_{\hbox{\input{modules/symbols/ZdotSym.tex}}\!}] \stackrel{\isom}{\longrightarrow} R[G_{\hbox{\input{modules/symbols/XdotSym.tex}}\!}]
\end{equation}
\end{theorem}
\begin{proof}
Enrichment over $R$-modules goes as per previous discussion. The classical states of $\hbox{\input{modules/symbols/ZdotSym.tex}}\!$, the group elements, form an orthogonal resolution of the identity, which can be used to canonically endow $\Hom{\CategoryC}{\tensorUnit}{\SpaceG}$ with the $R$-module structure of $R[G_{\hbox{\input{modules/symbols/ZdotSym.tex}}\!}]$; similarly for $\hbox{\input{modules/symbols/XdotSym.tex}}\!$, multiplicative characters and $\Hom{\CategoryC}{\SpaceG}{\tensorUnit}$. Both Theorem \ref{thm_EnrichedPontryaginDuality} gives the canonical isomorphism of $R$-modules between $\Hom{\CategoryC}{\tensorUnit}{\SpaceG}$ and $\Hom{\CategoryC}{\SpaceG}{\tensorUnit}$, and composing with the two isomorphisms above yields the desired result.
\end{proof}

\noindent In the case of $\fdHilbCategory$, the enrichment is over the category of $\complexs$-modules, and in fact it falls within the sub-category of finite-dimensional complex Hilbert spaces. Furthermore, if $\mathbb{G}$ is an abelian internal group, then the families of classical states satisfy the required conditions for Theorem~\ref{thm_FTRmodules}; also, we have $G_{\hbox{\input{modules/symbols/ZdotSym.tex}}\!} = G^\wedge$, where we defined $G := G_{\hbox{\input{modules/symbols/ZdotSym.tex}}\!}$. This (finally) gives us back the usual statement of Pontryagin duality in terms of $\LtwoSym$-spaces:
\begin{equation}
\mathcal{F}_{\mathbb{G}}: \Ltwo{G} \stackrel{\isom}{\longrightarrow} \Ltwo{G^\wedge}
\end{equation}

\section{Non-abelian Fourier Transform}
\label{section_NonAbelianFourierTransform}

\noindent In $\fdHilbCategory$, abelian internal groups satisfy the assumptions of Lemma \ref{lemma_FTTraditionalSMC}, and the corresponding Fourier Transform can be seen, via enrichment, as a canonical isomorphism of $\LtwoSym$-spaces $\Ltwo{G} \isom \Ltwo{G^\wedge}$. Non-abelian internal groups in $\fdHilbCategory$, however, fail those assumptions, as we will see that the classical states of the group structure never form a basis. However, it is a consequence of the Peter-Weyl theorem that the irreducible representations can be used to obtain a resolution of the identity. In this section, we generalise our treatment of internal groups from multiplicative characters to full-blown representation theory, concluding with a formulation of non-abelian Fourier Transform connected with the Gelfand--Naimark--Segal construction.

First, let's see why non-abelian internal groups in $\fdHilbCategory$ cannot satisfy the assumptions of Lemma \ref{lemma_FTTraditionalSMC}. The classical states of the point structure form an orthogonal basis, and are the elements of some non-abelian group $G$. The classical states of the group structures are the multiplicative characters of $G$: they are always orthogonal and normalisable, but as long as we show that there are less of them than the number of elements of $G$, we can conclude that they won't form a co-basis as would be required by Lemma \ref{lemma_FTTraditionalSMC}. But the multiplicative characters of a group $G$ are the same as the multiplicative character of its abelianization $G'$, which is always strictly smaller than $G$ (at least by a factor of 2). Therefore there are always strictly less multiplicative characters than group elements for a non-abelian internal group in $\fdHilbCategory$, and hence the assumptions of Lemma \ref{lemma_FTTraditionalSMC} always fail. It turns out that this is not restricted to $\fdHilbCategory$:

\begin{theorem}
Let $\mathbb{G} = (\,\SpaceG,\hbox{\input{modules/symbols/XdotSym.tex}}\!,\hbox{\input{modules/symbols/ZdotSym.tex}}\!)$ be an internal group in a $\dagger$-SMC $\CategoryC$. If the classical states of $\hbox{\input{modules/symbols/XdotSym.tex}}\!$ form a basis, then $\mathbb{G}$ is necessarily an abelian internal group.
\end{theorem}
\begin{proof}
All we have to show is that $\XmultSym = \XmultSym \circ s_{\SpaceG\SpaceG}$, where $s_{\SpaceG\SpaceG}$ is the braiding operator. Equivalently, it is enough to show that $\XcomultSym = s_{\SpaceG\SpaceG} \circ \XcomultSym$. But indeed for any classical state $\ket{\chi}$ of $\hbox{\input{modules/symbols/XdotSym.tex}}\!$ we have $\XcomultSym \circ \ket{\chi} = \ket{\chi} \tensor \ket{\chi} = s_{\SpaceG\SpaceG} \circ \ket{\chi} \tensor \ket{\chi} =   s_{\SpaceG\SpaceG} \circ \XcomultSym \circ \ket{\chi}$. As we have assumed that the classical states form a basis, this completes the proof.
\end{proof}

\noindent Now that we have established that non-abelian internal groups will never satisfy the requirements of Lemma \ref{lemma_FTTraditionalSMC}, we can start paving the way towards a more general, representation-theoretic formulation.  We will work with compact closed $\dagger$-SMCs, using arrows to indicate which end of the cups and caps is the dual object:
\begin{equation}\label{eqn:cupcap}
	\hbox{\input{modules/pictures/cupcap.tex}}
\end{equation}
To avoid misunderstanding, we give an example of the typing for adjoints and conjugates:
\begin{equation}\label{eqn:adjconj}
	\hbox{\input{modules/pictures/adjconj.tex}}
\end{equation}

\noindent So if $f:\SpaceG \rightarrow \SpaceH$, we have that $f^\dagger: \SpaceH \rightarrow \SpaceG$ and that $f^\star : \SpaceG^\star \rightarrow \SpaceH^\star$. When drawing diagrams, we will usually drop the $\;^\dagger$ and $\;^\star$, as the graphical orientation of the boxes and the direction of the wires are sufficient to disambiguate. The cups and caps on a space $\SpaceH$ induce a $\dagger$-Frobenius algebra on $\SpaceH \tensor \SpaceH^\star$, given as follows:
\begin{equation}\label{eqn:MatrixMultFrob}
	\hbox{\input{modules/pictures/MatrixMultFrob.tex}}
\end{equation}

\noindent This $\dagger$-Frobenius algebra is quasi-special if the composition $\cap_{\SpaceH} \circ \cap_\SpaceH^\dagger$ of the counit with the unit is an invertible scalar. In $\fdHilbCategory$, it corresponds to matrix multiplication (where endomorphisms $\SpaceH \rightarrow \SpaceH$ are seen as tensors $\SpaceH \tensor \SpaceH^\star \rightarrow \tensorUnit$ by compact closure). The representations of a monoid $(\, \SpaceG , \XmultSym,\XunitSym)$ in a compact-closed $\dagger$-SMC are defined to be the morphisms of monoids from $(\, \SpaceG , \XmultSym,\XunitSym)$ to the monoid given in \eqref{eqn:MatrixMultFrob}: the multiplicative characters given in Definition \ref{def_MultiplicativeCharacters} are then the special case $\SpaceH = \tensorUnit$, and this will form the starting point of our generalisation to representation theory.

\begin{definition}\label{def_Reps}
	The \textbf{representations} of a monoid $(\, \SpaceG , \XmultSym,\XunitSym)$ in a compact-closed $\dagger$-SMC are the morphisms $\hbox{\input{modules/symbols/repSym.tex}}\!\! : \SpaceG \rightarrow \SpaceH \tensor \SpaceH^\star$ satisfying the first two equations in \eqref{eqn:RepsDef}. The representations of an internal groups $(\, \SpaceG , \hbox{\input{modules/symbols/XdotSym.tex}}\!,\hbox{\input{modules/symbols/ZdotSym.tex}}\!)$ are the representations of the monoid $(\, \SpaceG , \XmultSym,\XunitSym)$. A representation of an internal group is \textbf{unitary} if it satisfies the third as well. A representation is \textbf{isometric} if it is an isometry.
	\begin{equation}\label{eqn:RepsDef}
		\hbox{\input{modules/pictures/RepsDef1.tex}}
	\end{equation}
\end{definition}

\noindent In \cite{StefanoGogioso-CategoricalSemanticsSchrodingersEqn}, (unitary) representations appear in a slightly different flavour, as certain Eilenberg-Moore algebras $\SpaceH \tensor \SpaceG \rightarrow \SpaceH$, without any need for compact closure. The (unitary) representations from Definition~\ref{def_Reps} are exactly the morphisms obtained from the representations defined in \cite{StefanoGogioso-CategoricalSemanticsSchrodingersEqn} by using compact-closure to bend the input $\SpaceH$ wire of the latter into the output $\SpaceH^\star$ wire of the former. However, in this work compact-closure is explicitly needed to define the characters associated with the representations of internal groups.

\begin{definition}\label{def_Characters}
	The \textbf{character} associated with a representation $\rho: \SpaceG \rightarrow \SpaceH \tensor \SpaceH^\star$ of a monoid / internal group in a compact-closed $\dagger$-SMC is the morphism $\chi_\rho : \SpaceG \rightarrow \tensorUnit$ defined by Equation \ref{eqn:CharacterDef}. In the case of internal groups, a character $\chi_\rho$ is \textbf{unitary} if the representation $\rho$ is. 
	\begin{equation}\label{eqn:CharacterDef}
		\hbox{\input{modules/pictures/CharacterDef.tex}}
	\end{equation}
\end{definition} 

\noindent In particular, the multiplicative characters given in Definition \ref{def_MultiplicativeCharacters} are both unitary representations and the corresponding characters. As we will shortly see, representations will generalise multiplicative characters in their role as classical states, and in giving a resolution of the identity; characters, on the other hand, will generalise multiplicative characters in their role as matching family, and in giving a partition of the counit.

\begin{definition}
Let $(\,\SpaceG , \XmultSym,\XunitSym)$ be a monoid in a compact-closed $\dagger$-SMC which is distributively $\ComMonCategory$-enriched, and let $(\rho: \SpaceG \rightarrow \SpaceH_\rho \tensor \SpaceH^\star_\rho)_{\rho \in \mathcal{R}}$ be a finite family of representations of $(\,\SpaceG , \XmultSym,\XunitSym)$. We say that the family is \textbf{orthogonal} if:
\begin{equation}
\sigma \circ \rho^\dagger = 0 \text{, for all } \sigma, \rho \in \mathcal{R} \text{ s.t. }\sigma \neq \rho
\end{equation}
where $0 : \SpaceH_\rho \tensor \SpaceH^\star_\rho \rightarrow \SpaceH_\sigma \tensor \SpaceH^\star_\sigma$ is the zero morphism given by $\ComMonCategory$-enrichment. We say that the family is \textbf{normalisable} with normalisation factors $(N_\rho)_{\rho \in \mathcal{R}}$ if:
\begin{equation}
\rho \circ \rho^\dagger = N_\rho \; \id{\SpaceH_\rho \tensor \SpaceH_\rho^\star} \text{, for some invertible scalar } N_\rho
\end{equation}
\end{definition}
 
\noindent In a picture, a finite, normalisable, orthonormal family of representations $(\rho)_{\rho \in \mathcal{R}}$ is one that satisfies the following family of equations, where $(N_\rho)_{\rho \in \mathcal{R}}$ is a family of invertible scalars:\footnote{Note that the $\delta_{\rho \sigma}$ symbol allows for the equation to be correct in type. If $\rho = \sigma$, then the identities are well-typed; if $\rho \neq \sigma$, then the RHS is the unique, well defined zero morphism $0 : \SpaceH_\rho \tensor \SpaceH^\star_\rho \rightarrow \SpaceH_\sigma \tensor \SpaceH^\star_\sigma$, also well-typed.}
	\begin{equation}\label{eqn:RepsOrthonormal}
		\hbox{\input{modules/pictures/RepsOrthonormal.tex}}
	\end{equation}

\noindent The following Lemma \ref{lemma_BasisResolutionPartition2} generalises Lemma \ref{lemma_BasisResolutionPartition} from classical states to representations.
 
\begin{lemma}\label{lemma_BasisResolutionPartition2}
Let $\hbox{\input{modules/symbols/XdotSym.tex}}\!$ be a $\dagger$-qSFA on an object $\SpaceG$ in a compact-closed $\dagger$-SMC which is distributively $\ComMonCategory$-enriched, and let $(\rho: \SpaceG \rightarrow \SpaceH_\rho \tensor \SpaceH^\star_\rho)_{\rho \in \mathcal{R}}$ be a finite, normalisable (with normalisation factors $(N_\rho)_{\rho \in \mathcal{R}}$), orthogonal family of representations of $(\,\SpaceG , \XmultSym,\XunitSym)$. Then the following are equivalent:
\begin{enumerate}
\item[(i)] the representations $(\rho)_{\rho \in \mathcal{R}}$ form an \textbf{(orthogonal) resolution of the identity}:
	\begin{equation}\label{eqn:RepsResolutionId}
		\hbox{\input{modules/pictures/RepsResolutionId.tex}}
	\end{equation}
\item[(ii)] the characters $(\chi_\rho)_{\rho \in \mathcal{R}}$ form an \textbf{(orthogonal) partition of the counit}:
	\begin{equation}\label{eqn:CharPartitionCounit}
		\hbox{\input{modules/pictures/CharPartitionCounit.tex}}
	\end{equation}
\end{enumerate}
\end{lemma}
\begin{proof}
The proof goes essentially like that of parts $(ii) \implies (iii)$ and $(iii) \implies (ii)$ from Lemma \ref{lemma_BasisResolutionPartition}, with representations and characters substituted to multiplicative characters in parts $(ii)$ and $(iii)$ respectively, and Definitions \ref{def_Reps} and \ref{def_Characters} in place of Definition \ref{def_MultiplicativeCharacters}.
\end{proof}

\noindent We saw at the beginning of this section that only abelian internal groups can have multiplicative characters forming an orthogonal resolution of the identity / partition of the counit: in $\fdHilbCategory$ this is the case for all abelian internal groups, while in $\fRelCategory$ this only occurs in  selected special cases. The following Theorem shows that the generalisation of Lemma \ref{lemma_BasisResolutionPartition} to Lemma \ref{lemma_BasisResolutionPartition2} will be sufficient to cover all internal groups in $\fdHilbCategory$, both abelian and non-abelian.

\begin{theorem}\label{thm_PeterWeyl} \textbf{(Peter-Weyl)}\\
	If $\mathbb{G} = (\, \SpaceG , \hbox{\input{modules/symbols/XdotSym.tex}}\!,\hbox{\input{modules/symbols/ZdotSym.tex}}\!)$ is an internal group in $\fdHilbCategory$ (with $\hbox{\input{modules/symbols/ZdotSym.tex}}\!$ a $\dagger$-SCFA), then there is a finite, normalisable, orthogonal family $\Irreps{\mathbb{G}}$ of unitary representations $\rho: \SpaceG \rightarrow \SpaceH_\rho \tensor \SpaceH_\rho^\star$, which form a resolution of the identity. The family can be obtained by taking exactly one representative for each equivalence class of irreducible representations of the finite group $G$, given by $(\XmultSym,\XunitSym)$ acting on the classical states of $\hbox{\input{modules/symbols/ZdotSym.tex}}\!$, and extending it linearly from $G$ to $\SpaceG$, under the isomorphism $\SpaceG \isom \Ltwo{G}$ induced by $\hbox{\input{modules/symbols/ZdotSym.tex}}\!$. We shall refer to one such family $\Irreps{\mathbb{G}}$ as a family of \textbf{irreps} (or irreducible representations) for the internal group $\mathbb{G}$. The normalisation factors are given by $N_\rho = N / d_\rho$, where $N$ is the dimensionality of $\SpaceG$ (the normalisation factor for $\hbox{\input{modules/symbols/XdotSym.tex}}\!$) and $d_\rho$ is the dimensionality of $\SpaceH_\rho$.
\end{theorem}
\begin{proof} This is the finite-dimensional case of the Peter-Weyl Theorem: see, for example, James and Liebeck~\cite{Msc-RepCharactersGroups} and Knapp~\cite{CQM-PeterWeylTheorem}.
\end{proof}

\noindent In $\fdHilbCategory$, the resolution of the identity provided by representations takes the form of a complete family of projectors, one $d_\rho$-dimensional projector for each irrep $\rho$. By fixing a basis for $\SpaceH$, this family can be turned into a family of $1$-dimensional projectors, using matrix elements.

\begin{definition}
    Let $(\,\SpaceG , \XmultSym,\XunitSym)$ be a monoid in a compact-closed $\dagger$-SMC and let $\rho: \SpaceG \rightarrow \SpaceH \tensor \SpaceH^\star$ be a representation. Let $\hbox{\input{modules/symbols/ZaltdotSym.tex}}\!$ be a $\dagger$-qSCFA on $\SpaceH$ with classical states $\ket{x}_{x \in X}$ forming a finite, normalisable, orthogonal basis for $\SpaceG$. The \textbf{matrix elements} of $\rho$ with respect to $\hbox{\input{modules/symbols/ZaltdotSym.tex}}\!$ are the following states of $\SpaceG$:
\begin{equation}
    \rho^x_y := \rho^\dagger \tensor \left( \ket{x} \tensor \ket{y}^\star \right) \text{, for } x,y \in X   
\end{equation}
\end{definition}

\begin{lemma}
Let $\hbox{\input{modules/symbols/XdotSym.tex}}\!$ be a $\dagger$-qSFA on an object $\SpaceG$ in a compact-closed $\dagger$-SMC which is distributively $\ComMonCategory$-enriched, and let $(\rho: \SpaceG \rightarrow \SpaceH_\rho \tensor \SpaceH^\star_\rho)_{\rho \in \mathcal{R}}$ be a finite, normalisable (with normalisation factors $(N_\rho)_{\rho \in \mathcal{R}}$), orthogonal family of representations of $(\,\SpaceG , \XmultSym,\XunitSym)$ which forms an orthogonal resolution of the identity. For each $\rho \in \mathcal{R}$, let $\hbox{\input{modules/symbols/ZaltdotSym.tex}}\!_\rho$ be a $\dagger$-qSCFA on $\SpaceH_\rho$ with classical states $\ket{x}_{x \in X_\rho}$ forming a finite, normalisable, orthogonal basis for $\SpaceG$. Then the family of all matrix elements for all the representations in $R$ forms a finite, normalisable, orthogonal basis for $\SpaceG$. 
\end{lemma}
\begin{proof}
By inserting the (separable) resolution of the identity provided by the basis $\left(\ket{x} \tensor \ket{y}^\star \right)_{x,y \in X}$ of $\SpaceH_\rho \tensor \SpaceH^\star_\rho$ in the middle $\SpaceH_\rho \tensor \SpaceH^\star_\rho$ wires of the resolution of the identity provided by the representations.
\end{proof}

\noindent Another useful result is the following, which generalises Theorem \ref{lemma_OrthogonalityCharacters}, or more precisely the equivalent formulation given by Equation \ref{eqn:orthogonalityMultCharsRed}.

\begin{theorem}[Orthogonality of Characters]\label{lemma_OrthogonalityCharacters2}
Let $\hbox{\input{modules/symbols/XdotSym.tex}}\!$ be a $\dagger$-qSFA on an object $\SpaceG$ in a compact-closed $\dagger$-SMC which is distributively $\ComMonCategory$-enriched, and let $(\rho)_{\rho \in \mathcal{R}}$ be a finite, normalisable, orthogonal family of representations of $(\,\SpaceG , \XmultSym,\XunitSym)$. If $(\rho)_{\rho \in \mathcal{R}}$ forms a resolution of the identity, then the \textbf{normalised characters} $(\frac{1}{N_\rho}\chi_\rho)_{\rho \in \mathcal{R}}$ are a matching family for $(\,\SpaceG , \XcomultSym,\XcounitSym)$:
    \begin{equation}\label{eqn:orthogonalityChars}
		\hbox{\input{modules/pictures/orthogonalityChars.tex}}
	\end{equation}
\end{theorem}
\begin{proof}
The proof begins by using the resolution of the identity from Equation \ref{eqn:RepsResolutionId}, and then Equation \ref{eqn:RepsDef}:
    \begin{equation}\label{eqn:CharOrthProof1}
		\hbox{\input{modules/pictures/CharOrthProof1.tex}}
	\end{equation}
The proof concludes by using orthogonality of representations from Equation \ref{eqn:RepsOrthonormal}:
    \begin{equation}\label{eqn:CharOrthProof2}
		\hbox{\input{modules/pictures/CharOrthProof2.tex}}
	\end{equation}
\end{proof}

\noindent Finally, Lemma \ref{lemma_FTTraditionalSMC2} provides the non-abelian analogue to the traditional formulation of the Fourier Transform given in Lemma \ref{lemma_FTTraditionalSMC}. It only deals with the Fourier Transform, i.e. with representations of $\hbox{\input{modules/symbols/XdotSym.tex}}\!$: this is because $\hbox{\input{modules/symbols/ZdotSym.tex}}\!$ is commutative, and hence the Inverse Fourier Transform is usually covered by Lemma \ref{lemma_FTTraditionalSMC}, but the result can be straightforwardly extended to the Inverse Fourier Transform

\begin{lemma}\label{lemma_FTTraditionalSMC2}
Let $\mathbb{G} = (\,\SpaceG,\hbox{\input{modules/symbols/XdotSym.tex}}\!,\hbox{\input{modules/symbols/ZdotSym.tex}}\!)$ be an internal group in a compact-closed $\dagger$-SMC which is distributively $\ComMonCategory$-enriched. Further assume that $(\rho)_{\rho \in \mathcal{R}}$ is a finite, normalisable family of representations of $\mathbb{G}$ (with normalisation factors $(N_\rho)_{\rho \in \mathcal{R}}$), which forms an orthogonal resolution of the identity. Then the Fourier Transform of Definition \ref{eqn:FT} can be written in the following way:
\begin{equation}\label{eqn:FTv2nonabelian}
	\hbox{\input{modules/pictures/FTv2nonabelian.tex}}
\end{equation} 
\end{lemma}
\begin{proof}
Goes on the same lines as the proof of Lemma \ref{lemma_FTTraditionalSMC}, using Lemma \ref{lemma_BasisResolutionPartition2} in place of Lemma \ref{lemma_BasisResolutionPartition}.
\end{proof}

\noindent Using Lemma \ref{lemma_FTTraditionalSMC2} and Theorem \ref{thm_PeterWeyl}, we recover a formulation of the Fourier Transform in $\fdHilbCategory$ which is reminiscent of the traditional abelian one, yet valid for all internal groups: the only thing missing is an analogue to the interpretation of $\mathcal{F}_{G}[f]$ as a function in $\Ltwo{G^\wedge}$. In the abelian case, this went as follows: 
\begin{enumerate}
\item[1.] The Fourier Transform for an abelian internal group  $\mathbb{G} = (\,\SpaceG,\hbox{\input{modules/symbols/XdotSym.tex}}\!,\hbox{\input{modules/symbols/ZdotSym.tex}}\!)$ in $\fdHilbCategory$ is given by Definition \ref{def_FourierTransform}.
\item[2.] By Theorem \ref{thm_CategoricalPontryaginDuality}, it is a canonical isomorphism of sets between states and co-states of $\SpaceG$.
\item[3.] By Theorem \ref{thm_EnrichedPontryaginDuality}, it is furthermore an isomorphism of finite-dimensional Hilbert spaces.
\item[4.] The group elements form a basis of $\SpaceG$, inducing a canonical isomorphism $\Hom{\fdHilbCategory}{\complexs}{\SpaceG} \isom \Ltwo{G}$ (see Section \ref{section_EnrichedPontryaginDuality}).
\item[5.] The multiplicative characters form a co-basis of $\SpaceG$, inducing a canonical isomorphism $\Hom{\fdHilbCategory}{\SpaceG}{\complexs} \isom \Ltwo{G^\wedge}$  (see Section \ref{section_EnrichedPontryaginDuality}).
\item[6.] Lemma \ref{lemma_FTTraditionalSMC} uses points 4 and 5 to turn the Fourier Transform into a canonical isomorphism $\Ltwo{G} \isom \Ltwo{G^\wedge}$, sending $\ket{f} = \sum_{g \in G} f(g) \ket{g}$ to $\bra{\tilde{f}} = \sum_{\chi \in G^\wedge} \tilde{f}(\chi) \bra{\chi}$.
\end{enumerate}

\noindent In Section \ref{section_EnrichedPontryaginDuality}, we have seen how, if $\CategoryC$ is a $\dagger$-SMC which is distributively $\ComMonCategory$-enriched, then $\CategoryC$ can be enriched over the category of $R$-modules, where $R$ is its semiring of scalars. If $\CategoryC$ is furthermore compact-closed with self-dual objects (i.e. $\SpaceH^\star = \SpaceH$), then the requirement of compatibility between $\star: \SpaceH \rightarrow \SpaceH$ and the $R$-module structure should be added to the enrichment.

In Section \ref{section_InternalGroups} we have defined the category $\IntGrpCategory{\CategoryC}$ of internal groups and internal group homomorphisms for any $\dagger$-SMC $\CategoryC$. Now we define a new category, $\IntGrpStructCategory{\CategoryC}$, with objects the internal groups and the same morphisms of $\CategoryC$:
\begin{equation}
\Hom{\IntGrpStructCategory{\CategoryC}}{(\SpaceG,\hbox{\input{modules/symbols/XdotSym.tex}}\!,\hbox{\input{modules/symbols/ZdotSym.tex}}\!)}{(\SpaceH,\hbox{\input{modules/symbols/XaltdotSym.tex}}\!,\hbox{\input{modules/symbols/ZaltdotSym.tex}}\!)} = \Hom{\CategoryC}{\SpaceG}{\SpaceH}
\end{equation}

\noindent Then there is a faithful functor $\IntGrpStructCategory{\CategoryC} \monom \IntGrpStructCategory{\CategoryC}$ which is bijective on objects, and a full and faithful functor $\IntGrpStructCategory{\CategoryC} \rightarrow \CategoryC$, which is furthermore surjective on objects (and hence an equivalence of categories) if and only if each object of $\CategoryC$ has at least one internal group structure on it.

Since $\IntGrpStructCategory{\CategoryC}$ keeps track of internal group structure, it can be enriched more than $\CategoryC$: if $\CategoryC$ is distributively $\ComMonCategory$-enriched, then the category $\IntGrpStructCategory{\CategoryC}$ can be enriched in unital $R$-algebras, where the unital ring structure on $\Hom{\IntGrpStructCategory{\CategoryC}}{(\SpaceG,\hbox{\input{modules/symbols/XdotSym.tex}}\!,\hbox{\input{modules/symbols/ZdotSym.tex}}\!)}{(\SpaceH,\hbox{\input{modules/symbols/XaltdotSym.tex}}\!,\hbox{\input{modules/symbols/ZaltdotSym.tex}}\!)}$ is given by the following operation and unit:
\begin{equation}\label{eqn:GrpAlgebraStructure}
	\hbox{\input{modules/pictures/GrpAlgebraStructure.tex}}
\end{equation}

\noindent In the case of $\fdHilbCategory$, take $F,F' : \Ltwo{G} \rightarrow \Ltwo{H}$, where we already identified $\SpaceG \isom \Ltwo{G}$ via $\hbox{\input{modules/symbols/ZdotSym.tex}}\!$ and $\SpaceH \isom \Ltwo{H}$ via $\hbox{\input{modules/symbols/ZaltdotSym.tex}}\!$. Then the ring operation and unit take the following explicit form:
\begin{align}
F \bullet F' &= f \mapsto \sum_{g \in G} f(g) \, F(\delta_g) \star F'(\delta_g) \\
1_\bullet &= f \mapsto \delta_0
\end{align}
where $\star$ is the convolution operation on $\Ltwo{H}$, and $\delta_0 \in \Ltwo{H}$ is its unit. In particula, the $\complexs$-algebra structure on $\Hom{\IntGrpStructCategory{\fdHilbCategory}}{\complexs}{(\SpaceG,\hbox{\input{modules/symbols/XdotSym.tex}}\!,\hbox{\input{modules/symbols/ZdotSym.tex}}\!)}$ is the convolution algebra on $\Ltwo{G}$, and the $\complexs$-algebra structure on $\Hom{\IntGrpStructCategory{\fdHilbCategory}}{(\SpaceG,\hbox{\input{modules/symbols/XdotSym.tex}}\!,\hbox{\input{modules/symbols/ZdotSym.tex}}\!)}{\complexs}$ is the pointwise multiplication algebra\footnote{Sending $f,f' \in \Ltwo{G^\wedge}$ to $f \bullet f' = \chi \mapsto f(\chi)f'(\chi)$.} on $\Ltwo{G^\wedge}$.

The key point is that the $R$-algebra structures on states and co-states of $\mathbb{G} = (\SpaceG,\hbox{\input{modules/symbols/XdotSym.tex}}\!,\hbox{\input{modules/symbols/ZdotSym.tex}}\!)$ in any $\IntGrpStructCategory{\CategoryC}$ are enough to reconstruct $\hbox{\input{modules/symbols/XdotSym.tex}}\!$ and $\hbox{\input{modules/symbols/ZdotSym.tex}}\!$ respectively: the algebraic generalisation of Pontryagin duality from the abelian case to the non-abelian case goes through this shift of perspective from $\LtwoSym$-spaces to algebras. Remember that, if the group elements of an internal group $\mathbb{G} = (\SpaceG,\hbox{\input{modules/symbols/XdotSym.tex}}\!,\hbox{\input{modules/symbols/ZdotSym.tex}}\!)$ (in a distributively $\ComMonCategory$-enriched category $\CategoryC$) form a finite, normalisable orthogonal basis, then $\Hom{\CategoryC}{\tensorUnit}{\SpaceG}^{(R\text{-modules})} \isom R[G_{\hbox{\input{modules/symbols/ZdotSym.tex}}\!}]$ as $R$-modules, where $R[G_{\hbox{\input{modules/symbols/ZdotSym.tex}}\!}]$ is the free $R$-module on the group elements (see Section \ref{section_EnrichedPontryaginDuality}). This generalised the $\LtwoSym$-spaces from $\fdHilbCategory$. 
The corresponding statement in $\IntGrpStructCategory{\CategoryC}$ is that $\Hom{\IntGrpStructCategory{\CategoryC}}{\tensorUnit}{\mathbb{G}}^{(R\text{-algebras})}$ is isomorphic, as $R$-algebras, to the algebra on $R[G_{\hbox{\input{modules/symbols/ZdotSym.tex}}\!}]$ induced by the isomorphism of $R$-modules above; we shall refer to this algebra over the $R$-module $R[G_{\hbox{\input{modules/symbols/ZdotSym.tex}}\!}]$ as $R[\mathbb{G}]$, as it involves both the $\hbox{\input{modules/symbols/ZdotSym.tex}}\!$ and the $\hbox{\input{modules/symbols/XdotSym.tex}}\!$ structures. This generalises $\LtwoSym$-spaces with their convolution algebra structure. 

\begin{lemma}\label{lemma_CategoricalGNRep}
Let $\mathbb{G} = (\,\SpaceG,\hbox{\input{modules/symbols/XdotSym.tex}}\!,\hbox{\input{modules/symbols/ZdotSym.tex}}\!)$ be an internal group in a compact-closed $\dagger$-SMC which is distributively $\ComMonCategory$-enriched. Assume that $(\rho:\SpaceG \rightarrow \SpaceH_\rho \tensor \SpaceH_\rho^\star)_{\rho \in \mathcal{R}}$ is a finite, normalisable family of representations of $\mathbb{G}$ (with normalisation factors $(N_\rho)_{\rho \in \mathcal{R}}$), which forms an orthogonal resolution of the identity. If $(i_\rho : \SpaceH_\rho \rightarrow V)_{\rho \in \mathcal{R}}$ is a family of orthogonal isometries (i.e. $i_\rho^\dagger \circ i_\sigma = 0$ whenever $\sigma \neq \rho$), then the following is an isometric representation of $\mathbb{G}$:
	\begin{equation}\label{eqn:GNSRep}
		\hbox{\input{modules/pictures/GNSRep.tex}}
	\end{equation}	
\end{lemma}
\begin{proof}
A series of straightforward checks. The fact that it is a representation follows from the fact that all $\rho$ are representations. The missing $i_\rho,i_\rho^\star$ morphisms are obtained by expanding the caps as:
\begin{equation}
\cap_{\SpaceH_\rho}^\star = \cap_V^\star \circ \left( i_\rho^\star \tensor i_\rho \right)
\end{equation}
The fact that it is an isometry follows from the fact that the $i_\rho$ (and hence the $i_\rho^\star$) are an orthogonal family of isometries, and that the representations form a resolution of the identity.
\end{proof}

\noindent In $\fdHilbCategory$, the algebra $\complexs[\mathbb{G}]$ is a C*-algebra, and the representation of Equation \ref{eqn:GNSRep} is a *-representation of C*-algebras. The typical example of one such representation involves taking $\mathcal{R}$ to be a complete set of representatives for the equivalence classes of irreducible representations of a finite group $G$, and letting $V = \oplus_{\rho \in \mathcal{R}} \SpaceH_\rho$ together with the isometric subspace injections $i_\rho : \SpaceH_\rho \monom V$ given by the biproduct structure. The resulting representation $\oplus_{\rho \in \mathcal{R}} \rho$ is reminiscent of the Gelfand-Naimark representation: the latter takes the same form, but with the entire set of pure states for the C*-algebra as $\mathcal{R}$, and the respective irreducible representations given by the GNS construction. For the purposes of proving the Gelfand-Naimark theorem in $\fdHilbCategory$, however, both representations are equally valid (they are both isometric).

\begin{theorem}[Categorical Gelfand-Naimark Theorem]\label{thm_CategoricalGNTheorem}.\\
Consider the situation where:
\begin{enumerate}
\item[(i)] $\mathbb{G} = (\,\SpaceG,\hbox{\input{modules/symbols/XdotSym.tex}}\!,\hbox{\input{modules/symbols/ZdotSym.tex}}\!)$ is an internal group in a compact-closed $\dagger$-SMC which is distributively $\ComMonCategory$-enriched.
\item[(ii)] $(\rho:\SpaceG \rightarrow \SpaceH_\rho \tensor \SpaceH_\rho^\star)_{\rho \in \mathcal{R}}$ is a finite, normalisable family of representations of $\mathbb{G}$ (with normalisation factors $(N_\rho)_{\rho \in \mathcal{R}}$), which forms an orthogonal resolution of the identity
\item[(iii)] The group elements of $\mathbb{G}$ form a finite, normalisable orthogonal basis, so that the $R$-algebra $R[\mathbb{G}]$ is well defined and isomorphic to $\Hom{\IntGrpStructCategory{\CategoryC}}{\tensorUnit}{\mathbb{G}}^{(R\text{-algebras})}$.
\item[(iv)] $(i_\rho : \SpaceH_\rho \rightarrow V)_{\rho \in \mathcal{R}}$ is a family of orthogonal isometries (i.e. $i_\rho^\dagger \circ i_\sigma = 0$ whenever $\sigma \neq \rho$).
\end{enumerate}
Then the $R$-algebra $R[\mathbb{G}]$ is isometrically isomorphic to a sub-algebra of the operator algebra given on $(V \tensor V^\star)[\mathcal{R}]$ by pointwise composition, where by $V \tensor V^\star$ we mean the $R$-module $\Hom{\CategoryC}{V \tensor V^\star}{\tensorUnit}^{(R\text{-modules})}$.
\end{theorem}
\begin{proof}
First of all, the $R$-module $V \tensor V^\star$ can be given a ring structure by the usual tensor composition:
\begin{equation}
M \cdot N := \left( M \tensor N \right) \cdot \left( \id{V} \tensor \cup_V \tensor \id{V^\star} \right)
\end{equation}
where $M,N \in \Hom{\CategoryC}{V \tensor V^\star}{\tensorUnit}$; the identity tesor is $\cap_V$. The free $(V \tensor V^\star)$-module $(V \tensor V^\star)[\mathcal{R}]$ is in particular an $R$-module (but not a free one), and admits the unital algebra structure of pointwise composition:
\begin{equation}
(M_\rho)_{\rho \in \mathcal{R}} \cdot (N_\rho)_{\rho \in \mathcal{R}} := (M_\rho \cdot N_\rho)_{\rho \in \mathcal{R}}
\end{equation}
where $(M_\rho)_\rho, (N_\rho)_\rho$ are $\mathcal{R}$-indexed families of co-states $V \tensor V^\star \rightarrow \tensorUnit$ of $\CategoryC$. Since the representations satisfy the hypotheses of Lemma \ref{lemma_FTTraditionalSMC2}, we can assume the Fourier Transform to take the form of Equation \ref{eqn:FTv2nonabelian}. Then using isometry of $i_\rho$ (and hence of $i_\rho^\star$), we obtain the following expression for the Fourier Transform $\tilde{f}$:
\begin{equation}\label{eqn:GNTheoremProof1}
	\hbox{\input{modules/pictures/GNTheoremProof1.tex}}
\end{equation}	
where we have defined states (which we will refer to as \textit{tensors}) $\bra{\tilde{f}(\rho)}: V \tensor V^\star \rightarrow \tensorUnit$ by:
\begin{equation}
\bra{\tilde{f}(\rho)} := \left( (i_\rho \tensor i_\rho^\star) \circ \rho \circ \hbox{\input{modules/symbols/antipodeSym.tex}}\! \circ \ket{f} \right)^T
\end{equation}
Exactly as in Equation \ref{eqn:FTSummation} one could recover each scalar $\tilde{f}(\chi)$ as the inner product $\braket{\tilde{f}}{\chi}$, i.e. as the composition $\bra{\tilde{f}} \circ \left(\bra{\chi}\right)^\dagger$, so one can recover the tensor $\bra{\tilde{f}(\rho)}$ in Equation \ref{eqn:GNTheoremProof1} as the composition $\bra{\tilde{f}} \circ \rho^\dagger$. Thus, like Equation \ref{eqn:FTSummation} turned the Fourier Transform into an injection of sets $\mathcal{F}_{\mathbb{G}}: R[G_{\hbox{\input{modules/symbols/ZdotSym.tex}}\!}] \inject R[G_{\hbox{\input{modules/symbols/XdotSym.tex}}\!}]$ (in fact, a bijection), so Equation \ref{eqn:GNTheoremProof1} turns the Fourier Transform into an injection of sets $\mathcal{F}_{\mathbb{G}}: R[\mathbb{G}] \inject (V\tensor V^\star)[\mathcal{R}]$ (but not a bijection unless the representation of Equation \ref{eqn:GNSRep} is a unitary morphism of $\CategoryC$):
\begin{align}
f \in R[\mathbb{G}] \mapsto \tilde{f} \in (V\tensor V^\star)[\mathcal{R}]\text{, where we defined } \tilde{f} = \rho \mapsto \bra{\tilde{f}(\rho)} 
\end{align}
Furthermore, the injection is necessarily $R$-linear, and isometric by Theorem \ref{lemma_CategoricalGNRep}, making it an isometry of $R$-modules. Finally, by the Convolution Theorem \ref{thm_categoricalConvolutionTheorem} it is a morphism of algebras from $R[\mathbb{G}]$ (which is defined as the $R$-algebra induced by $(\XmultSym,\XunitSym)$ on $R[G_{\hbox{\input{modules/symbols/ZdotSym.tex}}\!}]$) to the $R$-algebra on $(V \tensor V^\star)[\mathcal{R}]$ induced by $(\XcomultSym,\XcounitSym)$. All we need to show is that the latter is the algebra of pointwise composition of tensors.
\begin{equation}\label{eqn:GNTheoremProof2}
	\hbox{\input{modules/pictures/GNTheoremProof2.tex}}
\end{equation}	
\begin{equation}\label{eqn:GNTheoremProof3}
	\hbox{\input{modules/pictures/GNTheoremProof3.tex}}
\end{equation}
We have used the fact that all $i_\rho$ are isometries, and the fact that the representations form a resolution of the identity; then we have used the definition of representation to eliminate $\XcomultSym$, followed by orthogonality of representations and, finally, isometry of $i_\rho$ again. The leftmost map in Equation \ref{eqn:GNTheoremProof2} is the product of $(M_\rho)_{\rho\in R}$ and $(N_\rho)_{\rho\in R}$ in the algebra induced on $(V \tensor V^\star)[\mathcal{R}]$ by $(\XcomultSym,\XunitSym)$, while the rightmost map in Equation \ref{eqn:GNTheoremProof3} is the product of $(M_\rho)_{\rho\in R}$ and $(N_\rho)_{\rho \in R}$ in the pointwise composition algebra on $(V \tensor V^\star)[\mathcal{R}]$. Furthermore, by definition of representation and isometry of $i_\rho$, it is immediate to see that the algebra unit $\XcounitSym$ acts exactly as the algebra unit $\cup_V$. Thus $\mathcal{F}_{\mathbb{G}}$ is indeed an isometric $R$-algebra homomorphism from $R[\mathbb{G}]$ to $(V \tensor V^\star)[\mathcal{R}]$ with pointwise composition. $\mathcal{F}_{\mathbb{G}}$ restricted to its image (which is necessarily an $R$-algebra) is the desired isometric $R$-algebra isomorphism between $R[\mathbb{G}]$ and a subalgebra of $(V \tensor V^\star)[\mathcal{R}]$.
\end{proof}

\noindent Theorem \ref{thm_CategoricalGNTheorem} finally provides us with an appropriate generalisation of Pontryagin duality to the case of non-abelian internal groups in $\fdHilbCategory$. The first part goes pretty much as in the abelian case:
\begin{enumerate}
\item[1.] The Fourier Transform for an internal group $\mathbb{G} = (\,\SpaceG,\hbox{\input{modules/symbols/XdotSym.tex}}\!,\hbox{\input{modules/symbols/ZdotSym.tex}}\!)$ in $\fdHilbCategory$ is given by Definition \ref{def_FourierTransform}.
\item[2.] By Theorem \ref{thm_CategoricalPontryaginDuality}, it is a canonical isomorphism of sets between states and co-states of $\SpaceG$.
\item[3.] By Theorem \ref{thm_EnrichedPontryaginDuality}, it is furthermore an isomorphism of finite-dimensional Hilbert spaces  
\item[4.] The group elements form a basis of $\SpaceG$, inducing a canonical isomorphism $\Hom{\fdHilbCategory}{\complexs}{\SpaceG} \isom \Ltwo{G}$.
\item[4b.] This is in fact an isomorphism of C*-algebras between $(\SpaceG,\XcomultSym,\XunitSym)$ and the convolution algebra on $\Ltwo{G}$
\end{enumerate}

When things come to the $\hbox{\input{modules/symbols/XdotSym.tex}}\!$ structure, the generalisation to representations and $\complexs$-algebras kicks in:
\begin{enumerate}
\item[5.] By Theorem \ref{thm_PeterWeyl}, letting $(\rho: \SpaceG \rightarrow \SpaceH_\rho \tensor \SpaceH_\rho^\star)_{\rho \in \mathcal{R}}$ be the irreducible representations\footnote{To be more precise, a complete set of representatives from the equivalence classes of irreducible representations.} of $G$ yields a finite, normalisable orthogonal family of (necessarily unitary) representations.\footnote{With normalisation factors $d_\rho / N$, where $d_\rho$ is the dimensionality of irrep $\rho$ and $N$ is the dimensionality of $\SpaceG$, i.e. the number of elements in $G$.}
\item[6.] Taking $V = \oplus_{\rho \in \mathcal{R}} \SpaceH_\rho$ and $i_\rho: \SpaceH_\rho \monom V$ the isometric subspace injections given by the biproduct $\oplus$ of Hilbert spaces, yields an isometric representation $\oplus_{\rho \in \mathcal{R}} \rho$ as per Lemma \ref{lemma_CategoricalGNRep} (which is in fact a representation of C*-algebras).
\item[7.] A C*-algebra $(V \tensor V^\star)[\mathcal{R}]$ can be defined on the Hilbert space of $(V \tensor V^\star)$-valued functions on $\mathcal{R}$ by considering pointwise composition:
\begin{equation}
(M_\rho)_{\rho \in \mathcal{R}} \bullet (N_\rho)_{\rho \in \mathcal{R}} := (M_\rho \cdot N_\rho)_{\rho \in \mathcal{R}} 
\end{equation}
\item[7.] Theorem \ref{thm_CategoricalGNTheorem} (which uses Enriched Pontryagin Duality Theorem \ref{thm_EnrichedPontryaginDuality}, the Convolution Theorem \ref{thm_categoricalConvolutionTheorem}, Lemma \ref{lemma_FTTraditionalSMC2} and the C*-algebra representation given by Lemma \ref{lemma_CategoricalGNRep}) turns the Fourier Transform into an isomorphism of algebras between the convoluton algebra on $\Ltwo{G}$ and a subalgebra of the pointwise composition algebra on $(V \tensor V^\star)[\mathcal{R}]$. This is in fact an isomorphism of C*-algebras.
\end{enumerate}

\noindent To summarise, the Fourier Transform for an internal group $\mathbb{G} = (\,\SpaceG,\hbox{\input{modules/symbols/XdotSym.tex}}\!,\hbox{\input{modules/symbols/ZdotSym.tex}}\!)$ in $\fdHilbCategory$ is always given by Definition \ref{def_FourierTransform}, and the different forms it takes, as an isomorphism of $\LtwoSym$-spaces or of C*-algebras, all depend on which concrete decomposition is available for the abstract form. In $\fdHilbCategory$, the group elements always form a finite, normalisable orthogonal resolution of the identity, endowing the states of $\SpaceG$ with the structure of the convolution algebra on $\Ltwo{G}$, irregardles of whether $\mathbb{G}$ is abelian or not. All the internal groups in $\fdHilbCategory$ now fall within one of two categories: 
\begin{enumerate}
\item[(a)] If $\mathbb{G}$ is abelian, then the multiplicative characters $(\chi: \SpaceG \rightarrow \complexs)_{\chi \in G^\wedge}$ also form a finite, normalisable orthogonal resolution of the identity, and the Fourier Transform is a canonical isomorphism of C*-algebras between the convolution algebra on $\Ltwo{G}$ and the pointwise multiplication\footnote{Multiplication in $\complexs$, which is commutative.} algebra on $\Ltwo{G^\wedge}$:
\begin{align}
\mathcal{F}_{\mathbb{G}}[f] &:= \chi \mapsto \sum_{g \in G} f(g)\,\chi^\star(g)\\
\mathcal{F}_{\mathbb{G}}[f\star f'] &:= \chi \mapsto \mathcal{F}_{\mathbb{G}}[f](\chi)\mathcal{F}_{\mathbb{G}}[f'](\chi)
\end{align}
\item[(b)] If $\mathbb{G}$ is not abelian then the multiplicative characters cannot form an orthogonal basis. However, Theorem~\ref{thm_PeterWeyl} gives a finite, normalisable family $(\rho:\SpaceG \rightarrow \SpaceH_\rho \tensor \SpaceH_\rho^\star)_{\rho \in \mathcal{R}}$ of irreducible unitary representations of $G$ forming an orthogonal resolution of the identity. Letting $V = \oplus_{\rho \in \mathcal{R}}$, and $i_\rho: \SpaceH_\rho \monom V$ be the subspace injections, the Fourier Transform is a canonical isomorphism of C*-algebras, this time between the convolution algebra on $\Ltwo{G}$ and a subalgebra of the pointwise composition\footnote{Composition in $V \tensor V^\star$, which is non-commutative.} algebra on $(V \tensor V^\star)[\mathcal{R}]$:
\begin{align}
\mathcal{F}_{\mathbb{G}}[f] &:= \rho \mapsto \sum_{g \in G} f(f)\, \rho^\star(g)\\
\mathcal{F}_{\mathbb{G}}[f \star f'] &:= \rho \mapsto \mathcal{F}_{\mathbb{G}}[f](\rho)\circ\mathcal{F}_{\mathbb{G}}[f](\rho)
\end{align}
\end{enumerate}

\section{Conclusions}
\noindent In this work, we have provided an abstract characterisation of Fourier transform in terms of strongly complementary observables (more specifically, internal groups) in arbitrary $\dagger$-SMCs. We have then shown how, under appropriate conditions on the (at least distributively $\ComMonCategory$-enriched) category and observables involved, this definition specialises to a variety of different generalisations of the traditional Fourier transform:
\begin{enumerate}
\item[(a)] if the multiplicative characters form a resolution of the identity, we obtain a generalisation of the abelian Fourier transform from Hilbert spaces to modules over arbitrary semirings. In this case, the internal groups involved are always necessarily abelian, and in $\fdHilbCategory$ this recovers the usual Fourier transform.
\item[(b)] if the category is furthermore enriched over $R$-modules, which it can always be if we take $R$ to be its semiring of scalars, then we obtain a generalisation of Pontryagin duality from $\LtwoSym$-spaces to free modules over arbitrary semirings. In $\fdHilbCategory$, this recovers the usual Pontryagin duality.
\item[(c)] if the multiplicative characters fail to form a resolution of the identity, but some family of unitary representations (which can be defined in compact-closed categories) does, then we obtain a generalisation of the Gelfand-Naimark theorem from C*-algebras to algebras of free modules over arbitrary semirings. This is what happens for non-abelian groups in $\fdHilbCategory$, where this result reduces to the usual Gelfand-Naimark theorem.
\end{enumerate}

\noindent Aside from providing an interesting categorical generalisation of the Fourier transform, this work contributes to quantum information and computation by pin-pointing the operational features required by the quantum Fourier transform, one of the fundamental building blocks of quantum algorithms. Recent work by Aaronson and Ambainis~\cite{aaronson2014forrelation} provides a simple BQP-hard problem (FORRELATION) based entirely around the Fourier transform: we hope that the abstract tools provided in this work will provide valuable insight into the connection between strongly complementary observables and quantum speedup.\footnote{Or, more generally, the computational complexity class of the process theory in which they are embedded.}

\subparagraph*{Acknowledgements}
The authors would like to thank Bob Coecke, Aleks Kissinger, Chris Heunen, Amar Hadzihasanovic, Sukrita Chatterji and Nicol\`o Chiappori for comments, suggestions, useful discussions and support. Funding from EPSRC and Trinity College for the first author and The Rhodes Trust for the second is gratefully acknowledged. Both authors contributed equally to this work.

\bibliography{bibliography/CategoryTheory,bibliography/CategoricalQM,bibliography/NonLocalityContextuality,bibliography/QuantumComputing,bibliography/ClassicalMechanics,bibliography/LogicComputation,bibliography/Gravitation,bibliography/QFT,bibliography/StatisticalPhysics,bibliography/Misc,bibliography/StefanoGogioso,bibliography/thesis}

\appendix

\section{Background on Structures in Symmetric Monoidal Categories}
\label{app:basic}

\newcommand{\whitemonoid}[1]{\ensuremath{(#1, \ZmultSym, \ZunitSym)}}
\newcommand{\whitecomonoid}[1]{\ensuremath{(#1, \ZcomultSym, \ZcounitSym)}}
\newcommand{\blackmonoid}[1]{\ensuremath{(#1, \XmultSym, \XunitSym)}}
\newcommand{\blackcomonoid}[1]{\ensuremath{(#1, \XcomultSym, \XcounitSym)}}

In a dagger symmetric monoidal and unit object $I$, the \textbf{states} are maps $\ket{x}:I\to A$ for any object $A$.  The \textbf{co-states} are maps $(\ket{x})^{\dagger}=\bra{x}= A\to I$. The co-states are called effects by some authors. Scalars in this category are maps $s:I\to I$. This terminology is motivated by the interpretation of these concepts in $\fdHilbCategory$. See~\cite{abramsky2008categorical} as a main reference for these scalars and states in general monoidal categories. \\

\noindent We can equip any object $A$ in the category with a monoid \whitemonoid{A} and a comonoid \whitecomonoid{A}.  We can then define Frobenius algebras on that object when the monoid and comonoid interact in a particular way. A current reference for these structures is~\cite{coecke2015generalised}.

\begin{definition}
\label{def:frobenius}
In a dagger symmetric monoidal category the pair of a monoid \whitemonoid{A} and comonoid \whitecomonoid{A} form a \textbf{dagger-Frobenius algebra} ($\dagger$-FA) when the following equation holds:
\begin{equation}
\label{eq:frobenius}
\input{modules/pictures/1_Frobenius.tikz}
\end{equation}
We will often denote the $\dagger$-FA by its color $\hbox{\input{modules/symbols/ZdotSym.tex}}\!$.
\end{definition}

\noindent When $\ZmultSym$ is commutative the $\dagger$-FA is commutative (is a $\dagger$-CFA). The co-commutativity of $\ZcomultSym$ for a $\dagger$-CFA follows~\cite[Thm 3.2.8]{kissinger2012pictures}.

\begin{definition}
A $\dagger$-FA $(A, \ZmultSym,\ZunitSym,\ZcomultSym,\ZcounitSym)$ is special when
\begin{equation}
\ZmultSym\circ\ZcomultSym = \idm{A}.
\end{equation}
\end{definition}

\noindent A \textbf{classical structure} is a special commutative dagger Frobenius algebra. \\

\noindent In $\fdHilbCategory$, (special) classical structures correspond exactly to choices of (orthonormal) orthogonal bases~\cite{coecke2013new}. In particular the basis elements correspond to certain classical states of the classical structure.
\begin{definition}
\label{def:copyables}
The set of \textbf{classical states} $K_{\hbox{\input{modules/symbols/ZdotSym.tex}}\!}$ for a $\dagger$-FA $\hbox{\input{modules/symbols/ZdotSym.tex}}\!$ are all states $j:I\to A$ such that:
\begin{equation}
\label{eq:copy}
\begin{aligned}
\begin{tikzpicture}[yscale=-1]
\node [kpoint] (s) at (1.25,3) {$j$};
\draw (2,0) to [out=up, in=\seangle] (1.25,1.5);
\draw (0.5,0) to [out=up, in=\swangle] (1.25, 1.5);
\draw (1.25,1.5) to (s.north);
\node [dot, fill=\Zcolour] at (1.25,1.5) {};
\end{tikzpicture}
\end{aligned}
\quad=\quad
\begin{aligned}
\begin{tikzpicture}[yscale=-1]
\node (a) [kpoint] at (2.5,0) {$j$};
\node (b) [kpoint] at (0,0) {$j$};
\draw (a) to (2.5,2);
\draw (b) to (0,2);
\end{tikzpicture}
\end{aligned}
\end{equation}
\end{definition}

\begin{definition}\label{def:coherence}
In an dagger symmetric monoidal category, two $\dagger$-qSFA's (\, $\hbox{\input{modules/symbols/ZdotSym.tex}}\!$) and (\,$\hbox{\input{modules/symbols/XdotSym.tex}}\!$) are \textbf{coherent} when the following equations hold:
\begin{equation}
\label{eq:coherence}
    \input{modules/pictures/2_coher.tikz}
\end{equation}
\end{definition}

\noindent The usual notion of complementary (unbiased) bases in $\fdHilbCategory$ can be lifted to $\dagger$-SFA's, where it means that they obey the Hopf law:
\begin{definition}[Complementarity]
\label{def:complementarity}
In a dagger symmetric monoidal category, a $\dagger$-SFA (\,$\hbox{\input{modules/symbols/ZdotSym.tex}}\!$) and a $\dagger$-SCFA (\,$\hbox{\input{modules/symbols/XdotSym.tex}}\!$) on the same object are \textbf{complementary} when the following equation holds:
\begin{equation}
\label{eq:complementarity}
    \input{modules/pictures/2_antipodehopf.tikz}
\end{equation}
where $\hbox{\input{modules/symbols/antipodeSym.tex}}\!$ is the antipode from Definiton~\ref{def:Antipode}.
\end{definition}

\noindent The correspondence between this definition and that of mutually unbiased bases in $\fdHilbCategory$ is shown by \cite{coecke2015generalised}.

\end{document}

%% file: modules/packages.tex
\usepackage{graphicx}
\usepackage{amsfonts}
\usepackage{amsmath}  
\usepackage{amsthm}
\usepackage{amssymb}
\usepackage{relsize}
\usepackage{mathtools}
\usepackage{hyperref}
\usepackage[nocompress]{cite}
\usepackage[usenames,dvipsnames]{xcolor}
\usepackage{tikz}
\usepackage{tikzfig}
\usepackage{fullpage}

\usetikzlibrary{arrows,shapes,decorations,backgrounds,positioning,circuits.ee.IEC}


%% file: modules/macros.tex
	 \newcounter{theorem_c} 
	 \numberwithin{theorem_c}{section} 
	 \numberwithin{equation}{section} 

	\theoremstyle{plain} 
 	\newtheorem{theorem}[theorem_c]{Theorem}
	
 	\newtheorem{lemma}[theorem_c]{Lemma}
 	\newtheorem{corollary}[theorem_c]{Corollary}

	\theoremstyle{definition} 
 	\newtheorem{definition}[theorem_c]{Definition}
 	\newtheorem{example}[theorem_c]{Example}
 	\newtheorem{remark}[theorem_c]{Remark}

	\newcommand{\emptyArg}{\,\underline{\hspace{6px}}\,} 

	\newcommand{\integers}{\mathbb{Z}} 
	\newcommand{\reals}{\mathbb{R}} 
	\newcommand{\complexs}{\mathbb{C}} 
	\newcommand{\integersMod}[1]{\mathbb{Z}_{#1}} 
	\newcommand{\modclass}[2]{#1 \; (\text{mod } #2)} 
	\newcommand{\inject}{\hookrightarrow} 
	\newcommand{\Irreps}[1]{\operatorname{Irr}[#1]} 

	\newcommand{\eqdef}{\stackrel{def}{=}} 
	\newcommand{\suchthat}[2]{\left\{#1 \: \text{ s.t. } \: #2\right\}} 

		\newcommand{\ket}[1]{\vert #1 \rangle} 
		\newcommand{\bra}[1]{\langle #1 \vert} 
		\newcommand{\braket}[2]{\langle #1 \vert #2 \rangle} 

		\newcommand{\LtwoSym}{\operatorname{L}^2} 
		\newcommand{\Ltwo}[1]{\LtwoSym[#1]} 

		\newcommand{\SpaceH}{\mathcal{H}}

		\newcommand{\SpaceG}{\mathcal{G}}

		\newcommand{\isom}{\cong} 
		\newcommand{\monom}{\rightarrowtail} 
		\newcommand{\id}[1]{id_{#1}} 

		\newcommand{\tensor}{\otimes} 
		\newcommand{\tensorUnit}{I} 

		\newcommand{\Hom}[3]{\operatorname{Hom}_{\,#1}\left[#2,#3\right]} 
		\newcommand{\Endoms}[2]{\operatorname{End}_{\,#1}\left[#2\right]} 

		\newcommand{\SetCategory}{\operatorname{Set}} 
		\newcommand{\fdHilbCategory}{\operatorname{fdHilb}} 
		\newcommand{\ComMonCategory}{\operatorname{CMon}} 
		\newcommand{\RelCategory}{\operatorname{Rel}} 
		\newcommand{\fRelCategory}{\operatorname{fRel}} 

		\newcommand{\CategoryC}{\mathcal{C}}
		\newcommand{\CategoryD}{\mathcal{D}}

		\newcommand{\OpCategory}[1]{#1^{\operatorname{op}}} 
		\newcommand{\IntGrpCategory}[1]{\operatorname{Grp}[#1]} 
		\newcommand{\IntGrpStructCategory}[1]{\operatorname{GrpStruct}[#1]} 
		
	\newcommand{\unit}[1]{\eta_{#1}} 
	\newcommand{\mult}[1]{\mu_{#1}} 

	\newcommand{\timeunit}{\unit{}} 
	\newcommand{\timemult}{\mult{}} 
	\newcommand{\timeinversion}[1]{i_{#1}} 
	\newcommand{\timeinverse}{\timeinversion{}} 

	\newcommand{\Xcolour}{Red}
	\newcommand{\groupStructColour}{\Xcolour}
	\newcommand{\XmultSym}{\hbox{\input{modules/symbols/timemultSym.tex}}\!}
	\newcommand{\XcomultSym}{\hbox{\input{modules/symbols/timecomultSym.tex}}\!}
	\newcommand{\XunitSym}{\hbox{\input{modules/symbols/timeunitSym.tex}}\!}
	\newcommand{\XcounitSym}{\hbox{\input{modules/symbols/timecounitSym.tex}}\!}

	\newcommand{\Zcolour}{YellowGreen}
	\newcommand{\classicalStructColour}{\Zcolour}
	\newcommand{\ZmultSym}{\hbox{\input{modules/symbols/timematchSym.tex}}\!}
	\newcommand{\ZcomultSym}{\hbox{\input{modules/symbols/timediagSym.tex}}\!}
	\newcommand{\ZunitSym}{\hbox{\input{modules/symbols/timematchunitSym.tex}}\!}
	\newcommand{\ZcounitSym}{\hbox{\input{modules/symbols/trivialcharSym.tex}}\!}

	\newcommand{\Xaltcolour}{Purple}
	\newcommand{\internalgroupStructColour}{\Xaltcolour}
	
	\newcommand{\XaltcomultSym}{\hbox{\input{modules/symbols/internaltimecomultSym.tex}}\!}
	\newcommand{\XaltunitSym}{\hbox{\input{modules/symbols/internaltimeunitSym.tex}}\!}
	\newcommand{\XaltcounitSym}{\hbox{\input{modules/symbols/internaltimecounitSym.tex}}\!}
	
	\newcommand{\Zaltcolour}{Cyan}
	\newcommand{\internalclassicalStructColour}{\Zaltcolour}
	
	\newcommand{\ZaltmultSym}{\hbox{\input{modules/symbols/internaltimematchSym.tex}}\!}
	
	\newcommand{\ZaltunitSym}{\hbox{\input{modules/symbols/internaltimematchunitSym.tex}}\!}

	\newcommand{\tempSymLabel}{}

	\tikzset{->-/.style={decoration={markings,mark=at position #1 with {\arrow{>}}},postaction={decorate}}}
	\tikzset{-<-/.style={decoration={markings,mark=at position #1 with {\arrow{<}}},postaction={decorate}}}

%% file: modules/tikzstyles.tex
\def\swangle{-145}
\def\seangle{-35}
\def\nwangle{145}
\def\neangle{35}

\tikzstyle{env}=[copoint,regular polygon rotate=0,minimum width=0.2cm, fill=black]

\tikzstyle{probs}=[shape=semicircle,fill=white,draw=black,shape border rotate=180,minimum width=1.2cm]

\tikzstyle{every picture}=[baseline=-0.25em,scale=0.5]
\tikzstyle{dotpic}=[] 
\tikzstyle{diredges}=[every to/.style={diredge}]
\tikzstyle{math matrix}=[matrix of math nodes,left delimiter=(,right delimiter=),inner sep=2pt,column sep=1em,row sep=0.5em,nodes={inner sep=0pt},text height=1.5ex, text depth=0.25ex]

\tikzstyle{inline text}=[text height=1.5ex, text depth=0.25ex,yshift=0.5mm]
\tikzstyle{label}=[font=\footnotesize,text height=1.5ex, text depth=0.25ex,yshift=0.5mm]
\tikzstyle{left label}=[label,anchor=east,xshift=1.5mm]
\tikzstyle{right label}=[label,anchor=west,xshift=-1.5mm]

\tikzstyle{braceedge}=[decorate,decoration={brace,amplitude=2mm,raise=-1mm}]
\tikzstyle{small braceedge}=[decorate,decoration={brace,amplitude=1mm,raise=-1mm}]

\tikzstyle{doubled}=[line width=1.6pt] 
\tikzstyle{boldedge}=[doubled,shorten <=-0.17mm,shorten >=-0.17mm]
\tikzstyle{boldedgegray}=[doubled,gray,shorten <=-0.17mm,shorten >=-0.17mm]

\tikzstyle{semidoubled}=[line width=1.4pt] 
\tikzstyle{semiboldedgegray}=[semidoubled,gray,shorten <=-0.17mm,shorten >=-0.17mm]

\tikzstyle{boldedgedashed}=[very thick,dashed,shorten <=-0.17mm,shorten >=-0.17mm]
\tikzstyle{vboldedgedashed}=[doubled,dashed,shorten <=-0.17mm,shorten >=-0.17mm]
\tikzstyle{left hook arrow}=[left hook-latex]
\tikzstyle{right hook arrow}=[right hook-latex]
\tikzstyle{sembracket}=[line width=0.5pt,shorten <=-0.07mm,shorten >=-0.07mm]

\tikzstyle{causal edge}=[->,thick,gray]
\tikzstyle{causal nondir}=[thick,gray]
\tikzstyle{timeline}=[thick,gray, dashed]

\tikzstyle{cedge}=[<->,thick,gray!70!white]

\tikzstyle{empty diagram}=[draw=gray!40!white,dashed,shape=rectangle,minimum width=1cm,minimum height=1cm]
\tikzstyle{empty diagram small}=[draw=gray!50!white,dashed,shape=rectangle,minimum width=0.6cm,minimum height=0.5cm]

\tikzstyle{dot}=[inner sep=0mm,minimum width=3mm,minimum height=3mm,draw,shape=circle,text depth=-0.1mm]
\tikzstyle{ddot}=[inner sep=0mm, doubled, minimum width=3.5mm,minimum height=3.5mm,draw,shape=circle]

\tikzstyle{black dot}=[dot,fill=black]
\tikzstyle{white dot}=[dot,fill=white,,text depth=-0.2mm]
\tikzstyle{green dot}=[white dot] 
\tikzstyle{gray dot}=[dot,fill=gray!40!white,,text depth=-0.2mm]
\tikzstyle{red dot}=[gray dot] 

\tikzstyle{black ddot}=[ddot,fill=black]
\tikzstyle{white ddot}=[ddot,fill=white]
\tikzstyle{gray ddot}=[ddot,fill=gray!40!white]

\tikzstyle{gray edge}=[gray!40!white]

\tikzstyle{small dot}=[inner sep=0.5mm,minimum width=0pt,minimum height=0pt,draw,shape=circle]

\tikzstyle{small black dot}=[small dot,fill=black]
\tikzstyle{small white dot}=[small dot,fill=white]
\tikzstyle{small gray dot}=[small dot,fill=gray!40!white]

\tikzstyle{causal dot}=[inner sep=0.4mm,minimum width=0pt,minimum height=0pt,draw=white,shape=circle,fill=gray!40!white]

\tikzstyle{phase dimensions}=[minimum size=5mm,font=\footnotesize,rectangle,rounded corners=2.5mm,inner sep=0.2mm,outer sep=-2mm,text height=1ex, text depth=0.25ex, yshift=0.5mm]
\tikzstyle{dphase dimensions}=[phase dimensions]

\tikzstyle{phase dot}=[dot,phase dimensions]

\tikzstyle{white phase dot}=[dot,fill=white,phase dimensions]
\tikzstyle{white phase ddot}=[ddot,fill=white,dphase dimensions]

\tikzstyle{white rect ddot}=[draw=black,fill=white,doubled,minimum size=5mm,font=\footnotesize,rectangle,rounded corners=2.5mm,inner sep=0.2mm]
\tikzstyle{gray rect ddot}=[draw=black,fill=gray!40!white,doubled,minimum size=6mm,font=\footnotesize,rectangle,rounded corners=3mm]

\tikzstyle{gray phase dot}=[dot,fill=gray!40!white,phase dimensions]
\tikzstyle{gray phase ddot}=[ddot,fill=gray!40!white,dphase dimensions]
\tikzstyle{grey phase dot}=[gray phase dot]
\tikzstyle{grey phase ddot}=[gray phase ddot]

\tikzstyle{cnot}=[fill=white,shape=circle,inner sep=-1.4pt]
\tikzstyle{hadamard}=[square box,inner sep=0 pt,font=\footnotesize,minimum height=4mm,minimum width=4mm]
\tikzstyle{dhadamard}=[hadamard,doubled]
\tikzstyle{antipode}=[white dot,inner sep=0.3mm,font=\footnotesize]

\tikzstyle{scalar}=[diamond,draw,inner sep=0.5pt,font=\small]
\tikzstyle{dscalar}=[diamond,doubled, draw,inner sep=0.5pt,font=\small]

\tikzstyle{small box}=[rectangle,inline text,fill=white,draw,minimum height=5mm,yshift=-0.5mm,minimum width=5mm,font=\small]
\tikzstyle{small gray box}=[small box,fill=gray!30]
\tikzstyle{medium box}=[rectangle,inline text,fill=white,draw,minimum height=5mm,yshift=-0.5mm,minimum width=10mm,font=\small]
\tikzstyle{square box}=[small box] 
\tikzstyle{medium gray box}=[small box,fill=gray!30]
\tikzstyle{semilarge box}=[rectangle,inline text,fill=white,draw,minimum height=5mm,yshift=-0.5mm,minimum width=12.5mm,font=\small]
\tikzstyle{large box}=[rectangle,inline text,fill=white,draw,minimum height=5mm,yshift=-0.5mm,minimum width=15mm,font=\small]
\tikzstyle{large gray box}=[small box,fill=gray!30]

\tikzstyle{gray square point}=[small box,fill=gray!50]

\tikzstyle{dphase box white}=[dbox]
\tikzstyle{dphase box gray}=[dbox,fill=gray!50!white]

\tikzstyle{point}=[regular polygon,regular polygon sides=3,draw,scale=0.75,inner sep=-0.5pt,minimum width=9mm,fill=white,regular polygon rotate=180]
\tikzstyle{copoint}=[regular polygon,regular polygon sides=3,draw,scale=0.75,inner sep=-0.5pt,minimum width=9mm,fill=white]
\tikzstyle{dpoint}=[point,doubled]
\tikzstyle{dcopoint}=[copoint,doubled]

\tikzstyle{wide copoint}=[fill=white,draw,shape=isosceles triangle,shape border rotate=90,isosceles triangle stretches=true,inner sep=0pt,minimum width=1.5cm,minimum height=6.12mm]
\tikzstyle{wide point}=[fill=white,draw,shape=isosceles triangle,shape border rotate=-90,isosceles triangle stretches=true,inner sep=0pt,minimum width=1.5cm,minimum height=6.12mm,yshift=-0.0mm]
\tikzstyle{wide point plus}=[fill=white,draw,shape=isosceles triangle,shape border rotate=-90,isosceles triangle stretches=true,inner sep=0pt,minimum width=1.74cm,minimum height=7mm,yshift=-0.0mm]

\tikzstyle{wide dpoint}=[fill=white,doubled,draw,shape=isosceles triangle,shape border rotate=-90,isosceles triangle stretches=true,inner sep=0pt,minimum width=1.5cm,minimum height=6.12mm,yshift=-0.0mm]

\tikzstyle{tinypoint}=[regular polygon,regular polygon sides=3,draw,scale=0.55,inner sep=-0.15pt,minimum width=6mm,fill=white,regular polygon rotate=180] 

\tikzstyle{white point}=[point]
\tikzstyle{white dpoint}=[dpoint]
\tikzstyle{green point}=[white point] 
\tikzstyle{white copoint}=[copoint]
\tikzstyle{gray point}=[point,fill=gray!40!white]
\tikzstyle{gray dpoint}=[gray point,doubled]
\tikzstyle{red point}=[gray point] 
\tikzstyle{gray copoint}=[copoint,fill=gray!40!white]
\tikzstyle{gray dcopoint}=[gray copoint,doubled]

\tikzstyle{black point}=[point,fill=black]
\tikzstyle{black copoint}=[copoint,fill=black]

\tikzstyle{tiny gray point}=[tinypoint,fill=gray!40!white]

\tikzstyle{diredge}=[->]
\tikzstyle{rdiredge}=[<-]
\tikzstyle{thickdiredge}=[->, very thick]
\tikzstyle{pointer edge}=[->,very thick,gray]
\tikzstyle{pointer edge part}=[very thick,gray]
\tikzstyle{dashed edge}=[dashed]
\tikzstyle{thick dashed edge}=[very thick,dashed]
\tikzstyle{thick gray dashed edge}=[thick dashed edge,gray!40]
\tikzstyle{thick map edge}=[very thick,|->]

\makeatletter
\newcommand{\boxshape}[3]{%
\pgfdeclareshape{#1}{
\inheritsavedanchors[from=rectangle] 
\inheritanchorborder[from=rectangle]
\inheritanchor[from=rectangle]{center}
\inheritanchor[from=rectangle]{north}
\inheritanchor[from=rectangle]{south}
\inheritanchor[from=rectangle]{west}
\inheritanchor[from=rectangle]{east}
\backgroundpath{
\southwest \pgf@xa=\pgf@x \pgf@ya=\pgf@y
\northeast \pgf@xb=\pgf@x \pgf@yb=\pgf@y

\@tempdima=#2
\@tempdimb=#3

\pgfpathmoveto{\pgfpoint{\pgf@xa - 5pt + \@tempdima}{\pgf@ya}}
\pgfpathlineto{\pgfpoint{\pgf@xa - 5pt - \@tempdima}{\pgf@yb}}
\pgfpathlineto{\pgfpoint{\pgf@xb + 5pt + \@tempdimb}{\pgf@yb}}
\pgfpathlineto{\pgfpoint{\pgf@xb + 5pt - \@tempdimb}{\pgf@ya}}
\pgfpathlineto{\pgfpoint{\pgf@xa - 5pt + \@tempdima}{\pgf@ya}}
\pgfpathclose
}
}}

\boxshape{NEbox}{0pt}{5pt}
\boxshape{SEbox}{0pt}{-5pt}
\boxshape{NWbox}{5pt}{0pt}
\boxshape{SWbox}{-5pt}{0pt}
\boxshape{EBox}{-3pt}{3pt}
\boxshape{WBox}{3pt}{-3pt}
\makeatother

\tikzstyle{cloud}=[shape=cloud,draw,minimum width=1.5cm,minimum height=1.5cm]

\tikzstyle{map}=[draw,shape=NEbox,inner sep=2pt,minimum height=6mm,fill=white]
\tikzstyle{dashedmap}=[draw,dashed,shape=NEbox,inner sep=2pt,minimum height=6mm,fill=white]
\tikzstyle{mapdag}=[draw,shape=SEbox,inner sep=2pt,minimum height=6mm,fill=white]
\tikzstyle{mapadj}=[draw,shape=SEbox,inner sep=2pt,minimum height=6mm,fill=white]
\tikzstyle{maptrans}=[draw,shape=SWbox,inner sep=2pt,minimum height=6mm,fill=white]
\tikzstyle{mapconj}=[draw,shape=NWbox,inner sep=2pt,minimum height=6mm,fill=white]

\tikzstyle{medium map}=[draw,shape=NEbox,inner sep=2pt,minimum height=6mm,fill=white,minimum width=7mm]
\tikzstyle{medium map dag}=[draw,shape=SEbox,inner sep=2pt,minimum height=6mm,fill=white,minimum width=7mm]
\tikzstyle{medium map adj}=[draw,shape=SEbox,inner sep=2pt,minimum height=6mm,fill=white,minimum width=7mm]
\tikzstyle{medium map trans}=[draw,shape=SWbox,inner sep=2pt,minimum height=6mm,fill=white,minimum width=7mm]
\tikzstyle{medium map conj}=[draw,shape=NWbox,inner sep=2pt,minimum height=6mm,fill=white,minimum width=7mm]
\tikzstyle{semilarge map}=[draw,shape=NEbox,inner sep=2pt,minimum height=6mm,fill=white,minimum width=9.5mm]
\tikzstyle{semilarge map trans}=[draw,shape=SWbox,inner sep=2pt,minimum height=6mm,fill=white,minimum width=9.5mm]
\tikzstyle{semilarge map adj}=[draw,shape=SEbox,inner sep=2pt,minimum height=6mm,fill=white,minimum width=9.5mm]
\tikzstyle{semilarge map dag}=[draw,shape=SEbox,inner sep=2pt,minimum height=6mm,fill=white,minimum width=9.5mm]
\tikzstyle{semilarge map conj}=[draw,shape=NWbox,inner sep=2pt,minimum height=6mm,fill=white,minimum width=9.5mm]
\tikzstyle{large map}=[draw,shape=NEbox,inner sep=2pt,minimum height=6mm,fill=white,minimum width=12mm]
\tikzstyle{very large map}=[draw,shape=NEbox,inner sep=2pt,minimum height=6mm,fill=white,minimum width=17mm]

\tikzstyle{medium dmap}=[draw,doubled,shape=NEbox,inner sep=2pt,minimum height=6mm,fill=white,minimum width=7mm]
\tikzstyle{medium dmap dag}=[draw,doubled,shape=SEbox,inner sep=2pt,minimum height=6mm,fill=white,minimum width=7mm]
\tikzstyle{medium dmap adj}=[draw,doubled,shape=SEbox,inner sep=2pt,minimum height=6mm,fill=white,minimum width=7mm]
\tikzstyle{medium dmap trans}=[draw,doubled,shape=SWbox,inner sep=2pt,minimum height=6mm,fill=white,minimum width=7mm]
\tikzstyle{medium dmap conj}=[draw,doubled,shape=NWbox,inner sep=2pt,minimum height=6mm,fill=white,minimum width=7mm]
\tikzstyle{semilarge dmap}=[draw,doubled,shape=NEbox,inner sep=2pt,minimum height=6mm,fill=white,minimum width=9.5mm]
\tikzstyle{semilarge dmap trans}=[draw,doubled,shape=SWbox,inner sep=2pt,minimum height=6mm,fill=white,minimum width=9.5mm]
\tikzstyle{semilarge dmap adj}=[draw,doubled,shape=SEbox,inner sep=2pt,minimum height=6mm,fill=white,minimum width=9.5mm]
\tikzstyle{semilarge dmap dag}=[draw,doubled,shape=SEbox,inner sep=2pt,minimum height=6mm,fill=white,minimum width=9.5mm]
\tikzstyle{semilarge dmap conj}=[draw,doubled,shape=NWbox,inner sep=2pt,minimum height=6mm,fill=white,minimum width=9.5mm]
\tikzstyle{large dmap}=[draw,doubled,shape=NEbox,inner sep=2pt,minimum height=6mm,fill=white,minimum width=12mm]
\tikzstyle{large dmap conj}=[draw,doubled,shape=NWbox,inner sep=2pt,minimum height=6mm,fill=white,minimum width=12mm]
\tikzstyle{large dmap trans}=[draw,doubled,shape=SWbox,inner sep=2pt,minimum height=6mm,fill=white,minimum width=12mm]
\tikzstyle{very large dmap}=[draw,doubled,shape=NEbox,inner sep=2pt,minimum height=6mm,fill=white,minimum width=19.5mm]

\tikzstyle{muxbox}=[draw,shape=rectangle,minimum height=3mm,minimum width=3mm,fill=white]
\tikzstyle{dmuxbox}=[muxbox,doubled]

\tikzstyle{box}=[draw,shape=rectangle,inner sep=2pt,minimum height=6mm,minimum width=6mm,fill=white]
\tikzstyle{dbox}=[draw,doubled,shape=rectangle,inner sep=2pt,minimum height=6mm,minimum width=6mm,fill=white]
\tikzstyle{dmap}=[draw,doubled,shape=NEbox,inner sep=2pt,minimum height=6mm,fill=white]
\tikzstyle{dmapdag}=[draw,doubled,shape=SEbox,inner sep=2pt,minimum height=6mm,fill=white]
\tikzstyle{dmapadj}=[draw,doubled,shape=SEbox,inner sep=2pt,minimum height=6mm,fill=white]
\tikzstyle{dmaptrans}=[draw,doubled,shape=SWbox,inner sep=2pt,minimum height=6mm,fill=white]
\tikzstyle{dmapconj}=[draw,doubled,shape=NWbox,inner sep=2pt,minimum height=6mm,fill=white]

\tikzstyle{ddmap}=[draw,doubled,dashed,shape=NEbox,inner sep=2pt,minimum height=6mm,fill=white]
\tikzstyle{ddmapdag}=[draw,doubled,dashed,shape=SEbox,inner sep=2pt,minimum height=6mm,fill=white]
\tikzstyle{ddmapadj}=[draw,doubled,dashed,shape=SEbox,inner sep=2pt,minimum height=6mm,fill=white]
\tikzstyle{ddmaptrans}=[draw,doubled,dashed,shape=SWbox,inner sep=2pt,minimum height=6mm,fill=white]
\tikzstyle{ddmapconj}=[draw,doubled,dashed,shape=NWbox,inner sep=2pt,minimum height=6mm,fill=white]

\boxshape{sNEbox}{0pt}{3pt}
\boxshape{sSEbox}{0pt}{-3pt}
\boxshape{sNWbox}{3pt}{0pt}
\boxshape{sSWbox}{-3pt}{0pt}
\tikzstyle{smap}=[draw,shape=sNEbox,fill=white]
\tikzstyle{smapdag}=[draw,shape=sSEbox,fill=white]
\tikzstyle{smapadj}=[draw,shape=sSEbox,fill=white]
\tikzstyle{smaptrans}=[draw,shape=sSWbox,fill=white]
\tikzstyle{smapconj}=[draw,shape=sNWbox,fill=white]

\tikzstyle{dsmap}=[draw,dashed,shape=sNEbox,fill=white]
\tikzstyle{dsmapdag}=[draw,dashed,shape=sSEbox,fill=white]
\tikzstyle{dsmaptrans}=[draw,dashed,shape=sSWbox,fill=white]
\tikzstyle{dsmapconj}=[draw,dashed,shape=sNWbox,fill=white]

\boxshape{mNEbox}{0pt}{10pt}
\boxshape{mSEbox}{0pt}{-10pt}
\boxshape{mNWbox}{10pt}{0pt}
\boxshape{mSWbox}{-10pt}{0pt}
\tikzstyle{mmap}=[draw,shape=mNEbox]
\tikzstyle{mmapdag}=[draw,shape=mSEbox]
\tikzstyle{mmaptrans}=[draw,shape=mSWbox]
\tikzstyle{mmapconj}=[draw,shape=mNWbox]

\tikzstyle{mmapgray}=[draw,fill=gray!40!white,shape=mNEbox]
\tikzstyle{smapgray}=[draw,fill=gray!40!white,shape=sNEbox]

\makeatletter
\pgfdeclareshape{cornerpoint}{
\inheritsavedanchors[from=rectangle] 
\inheritanchorborder[from=rectangle]
\inheritanchor[from=rectangle]{center}
\inheritanchor[from=rectangle]{north}
\inheritanchor[from=rectangle]{south}
\inheritanchor[from=rectangle]{west}
\inheritanchor[from=rectangle]{east}
\backgroundpath{
\southwest \pgf@xa=\pgf@x \pgf@ya=\pgf@y
\northeast \pgf@xb=\pgf@x \pgf@yb=\pgf@y

\pgfmathsetmacro{\pgf@shorten@left}{\pgfkeysvalueof{/tikz/shorten left}}
\pgfmathsetmacro{\pgf@shorten@right}{\pgfkeysvalueof{/tikz/shorten right}}

\pgfpathmoveto{\pgfpoint{0.5 * (\pgf@xa + \pgf@xb)}{\pgf@ya - 5pt}}
\pgfpathlineto{\pgfpoint{\pgf@xa - 8pt + \pgf@shorten@left}{\pgf@yb - 1.5 * \pgf@shorten@left}}
\pgfpathlineto{\pgfpoint{\pgf@xa - 8pt + \pgf@shorten@left}{\pgf@yb}}
\pgfpathlineto{\pgfpoint{\pgf@xb + 8pt - \pgf@shorten@right}{\pgf@yb}}
\pgfpathlineto{\pgfpoint{\pgf@xb + 8pt - \pgf@shorten@right}{\pgf@yb - 1.5 * \pgf@shorten@right}}
\pgfpathclose
}
}

\pgfdeclareshape{cornercopoint}{
\inheritsavedanchors[from=rectangle] 
\inheritanchorborder[from=rectangle]
\inheritanchor[from=rectangle]{center}
\inheritanchor[from=rectangle]{north}
\inheritanchor[from=rectangle]{south}
\inheritanchor[from=rectangle]{west}
\inheritanchor[from=rectangle]{east}
\backgroundpath{
\southwest \pgf@xa=\pgf@x \pgf@ya=\pgf@y
\northeast \pgf@xb=\pgf@x \pgf@yb=\pgf@y

\pgfmathsetmacro{\pgf@shorten@left}{\pgfkeysvalueof{/tikz/shorten left}}
\pgfmathsetmacro{\pgf@shorten@right}{\pgfkeysvalueof{/tikz/shorten right}}

\pgfpathmoveto{\pgfpoint{0.5 * (\pgf@xa + \pgf@xb)}{\pgf@yb + 5pt}}
\pgfpathlineto{\pgfpoint{\pgf@xa - 8pt + \pgf@shorten@left}{\pgf@ya + 1.5 * \pgf@shorten@left}}
\pgfpathlineto{\pgfpoint{\pgf@xa - 8pt + \pgf@shorten@left}{\pgf@ya}}
\pgfpathlineto{\pgfpoint{\pgf@xb + 8pt - \pgf@shorten@right}{\pgf@ya}}
\pgfpathlineto{\pgfpoint{\pgf@xb + 8pt - \pgf@shorten@right}{\pgf@ya + 1.5 * \pgf@shorten@right}}
\pgfpathclose
}
}

\makeatother

\pgfkeyssetvalue{/tikz/shorten left}{0pt}
\pgfkeyssetvalue{/tikz/shorten right}{0pt}

\tikzstyle{kpoint common}=[draw,fill=white,inner sep=1pt,minimum height=3mm]
\tikzstyle{kpoint}=[shape=cornerpoint,shorten left=5pt,kpoint common]
\tikzstyle{kpoint adjoint}=[shape=cornercopoint,shorten left=5pt,kpoint common]
\tikzstyle{kpoint conjugate}=[shape=cornerpoint,shorten right=5pt,kpoint common]
\tikzstyle{kpoint transpose}=[shape=cornercopoint,shorten right=5pt,kpoint common]
\tikzstyle{kpoint symm}=[shape=cornerpoint,shorten left=5pt,shorten right=5pt,kpoint common]

\tikzstyle{black kpoint}=[shape=cornerpoint,shorten left=5pt,kpoint common,fill=black]
\tikzstyle{black kpoint adjoint}=[shape=cornercopoint,shorten left=5pt,kpoint common,fill=black]

\tikzstyle{kpointdag}=[kpoint adjoint]
\tikzstyle{kpointadj}=[kpoint adjoint]
\tikzstyle{kpointconj}=[kpoint conjugate]
\tikzstyle{kpointtrans}=[kpoint transpose]

\tikzstyle{big kpoint}=[kpoint, minimum width=1.2 cm, minimum height=8mm, inner sep=4pt, text depth=3mm]

\tikzstyle{wide kpoint}=[kpoint, minimum width=1 cm, inner sep=2pt, text depth=-0.7 mm]
\tikzstyle{wide kpointdag}=[kpointdag, minimum width=1 cm, inner sep=2pt, text depth=0.7 mm]
\tikzstyle{wide kpointconj}=[kpointconj, minimum width=1 cm, inner sep=2pt, text depth=-0.7 mm]
\tikzstyle{wide kpointtrans}=[kpointtrans, minimum width=1 cm, inner sep=2pt, text depth=0.7 mm]

\tikzstyle{gray kpoint}=[kpoint,fill=gray!50!white]
\tikzstyle{gray kpointdag}=[kpointdag,fill=gray!50!white]
\tikzstyle{gray kpointadj}=[kpointadj,fill=gray!50!white]
\tikzstyle{gray kpointconj}=[kpointconj,fill=gray!50!white]
\tikzstyle{gray kpointtrans}=[kpointtrans,fill=gray!50!white]

\tikzstyle{gray dkpoint}=[kpoint,fill=gray!50!white,doubled]
\tikzstyle{gray dkpointdag}=[kpointdag,fill=gray!50!white,doubled]
\tikzstyle{gray dkpointadj}=[kpointadj,fill=gray!50!white,doubled]
\tikzstyle{gray dkpointconj}=[kpointconj,fill=gray!50!white,doubled]
\tikzstyle{gray dkpointtrans}=[kpointtrans,fill=gray!50!white,doubled]

\tikzstyle{white label}=[draw,fill=white,rectangle,inner sep=0.7 mm]
\tikzstyle{gray label}=[draw,fill=gray!50!white,rectangle,inner sep=0.7 mm]
\tikzstyle{black label}=[draw,fill=black,rectangle,inner sep=0.7 mm]

\tikzstyle{dkpoint}=[kpoint,doubled]
\tikzstyle{wide dkpoint}=[wide kpoint,doubled]
\tikzstyle{dkpointdag}=[kpoint adjoint,doubled]
\tikzstyle{dkcopoint}=[kpoint adjoint,doubled]
\tikzstyle{dkpointadj}=[kpoint adjoint,doubled]
\tikzstyle{dkpointconj}=[kpoint conjugate,doubled]
\tikzstyle{dkpointtrans}=[kpoint transpose,doubled]

\tikzstyle{kscalar}=[kpoint common, shape=EBox, inner xsep=-1pt, inner ysep=3pt,font=\small]
\tikzstyle{kscalarconj}=[kpoint common, shape=WBox, inner xsep=-1pt, inner ysep=3pt,font=\small]

 \tikzstyle{upground}=[circuit ee IEC,thick,ground,rotate=90,scale=2.5]
 \tikzstyle{downground}=[circuit ee IEC,thick,ground,rotate=-90,scale=2.5]
 \tikzstyle{bigground}=[regular polygon,regular polygon sides=3,draw=gray,scale=0.50,inner sep=-0.5pt,minimum width=10mm,fill=gray]

\tikzstyle{arrs}=[-latex,font=\small,auto]
\tikzstyle{arrow plain}=[arrs]
\tikzstyle{arrow dashed}=[dashed,arrs]
\tikzstyle{arrow bold}=[very thick,arrs]
\tikzstyle{arrow hide}=[draw=white!0,-]
\tikzstyle{arrow reverse}=[latex-]
\tikzstyle{cdnode}=[]